\documentclass{amsart}
\usepackage[margin=1in]{geometry}
\usepackage{amsmath, amsfonts, amssymb}
\usepackage{color}
\usepackage{algpseudocode}
\usepackage{algorithm}
\usepackage{siunitx}
\usepackage{graphicx}
\usepackage{subfig}
\graphicspath{{images/}}

\newcommand\bdf{\boldsymbol{f}}
\newcommand\bx{\boldsymbol{x}}
\newcommand\bv{\boldsymbol{v}}
\newcommand\bu{\boldsymbol{u}}
\newcommand\bn{\boldsymbol{n}}
\newcommand\bw{\boldsymbol{w}}
\newcommand\lu{\overline{\boldsymbol{u}}}
\newcommand\lth{\overline{\theta}}
\newcommand\mm{\mathfrak{m}}
\newcommand\Kn{\mathit{Kn}}

\DeclareMathOperator\re{Re}
\DeclareMathOperator\im{Im}

\newtheorem{theorem}{Theorem}
\newtheorem{lemma}{Lemma}

\theoremstyle{remark}
\newtheorem*{remark}{Remark}

\title{Burnett Spectral Method for High-Speed Rarefied Gas Flows}

\author{Zhicheng Hu}
\address[Zhicheng Hu]{Department of Mathematics, College of Science,
    Nanjing University of Aeronautics and Astronautics, Nanjing
    210016, China}
\email{huzhicheng@nuaa.edu.cn}

\author{Zhenning Cai}
\address[Zhenning Cai]{Department of Mathematics, National University of Singapore,
  Level 4, Block S17, 10 Lower Kent Ridge Road, Singapore 119076}
\email{matcz@nus.edu.sg}

\thanks{Zhicheng Hu's work is partially supported by the National
  Natural Science Foundation of China (11601229), and the Natural
  Science Foundation of Jiangsu Province of China
  (BK20160784). Zhenning Cai's work was supported by National
  University of Singapore Startup Fund under grant
  No. R-146-000-241-133. The computational resources are supported by
  the High-performance Computing Platform of Peking University, China}


\begin{document}
\maketitle

\begin{abstract}
We introduce a numerical solver for the spatially inhomogeneous Boltzmann
equation using the Burnett spectral method. The modelling and discretization of
the collision operator are based on the previous work \cite{Cai2019}, which is
the hybridization of the BGK operator for higher moments and the quadratic
collision operator for lower moments. To ensure the preservation of the
equilibrium state, we introduce an additional term to the discrete collision
operator, which equals zero when the number of degrees of freedom tends to
infinity. Compared with the previous work \cite{Hu2019}, the computational cost
is reduced by one order. Numerical experiments such as shock structure
calculation and Fourier flows are carried out to show the efficiency and
accuracy of our numerical method.
\end{abstract}

\section{Introduction}
In rarefied gas dynamics, models based on continuum hypothesis such as Euler
equations and Navier-Stokes equations cannot provide accurate prediction of the
flow structure. To correctly describe the motion of fluids, one needs to employ
gas kinetic theory to capture the rarefaction effects. However, switching from
continuum models to kinetic models may greatly increase the computational
difficulty, since the kinetic theory uses the distribution function to describe
the velocity distribution of gas molecules, which doubles the dimensions of the
unknown function in the differential equations. Thus, solving kinetic models
deterministically has long been considered intractable, and the DSMC (direct
simulation of Monte Carlo) method has been playing an important role in the
simulation of rarefied gases \cite{Bird1963, Bird1994}.

Nowadays, due to the fast improvement of the CPU performance, researchers have
tried to solve the six-dimensional kinetic equations deterministically
\cite{Huang2012, Dechriste2016, Dimarco2018}. At the same time, a number of new
ideas have been proposed to accelerate the numerical solver \cite{Filbet2011,
Xu2011, Dimarco2019}. In particular, for the Boltzmann equation,
which has an additional difficulty due to its highly complicated binary
collision term, a significant progress on fast algorithms has been made in
recent years \cite{Wu2013, Alekseenko2014, Gamba2017, Gamba2018, Alekseenko2019,
Kitzler2019}. These works have shown great promise for practical applications of
these deterministic Boltzmann solvers in the near future. Our work also
contributes to this research field. In this paper, we are going to introduce a
new numerical solver for the spatially inhomogeneous Boltzmann equation.

Currently, the most popular numerical solver for the binary collision term of
the Boltzmann equation is likely to be the Fourier spectral method
\cite{Pareschi1996, Bobylev1997} and its variations \cite{Gamba2009, Mouhot2006,
Filbet2015}. However, for general gas molecules, the time complexity for
evaluating the collision term is quadratic in the number of degrees of freedom
in the velocity space. Thus, when a large number of Fourier modes are needed to
resolve the distribution function (e.g., when the distribution function is
discontinuous), this still introduces a large computational cost. In
\cite{Cai2015, Wang2019, Cai2019}, the authors have been trying to reduce the
computational cost by restricting the evaluation of the expensive collision term
only to a small number of degrees of freedom, so that the computational of the
collision term may still be affordable even for a relatively large number of
degrees of freedom in the velocity space. To achieve such a goal,
Hermite/Burnett polynomials are used instead of Fourier basis functions, so that
the idea of BGK-type modeling can be applied. Such a strategy has been verified
in the spatially inhomogeneous problems \cite{Cai2018, Hu2019}. The price to pay
is a higher time complexity for the expensive part compared with the Fourier
spectral method. In \cite{Cai2019}, it has been observed in the numerical
solution of spatially homogeneous Boltzmann equation that the implementation
using Burnett polynomials (orthogonal polynomials based on spherical
coordinates) is much faster than Hermite polynomials (orthogonal polynomials
based on Cartesian coordinates), despite their mathematical equivalence. This
work can be considered as a test of its performance with the presence of spatial
variables.

For the special implementation of the collision term \cite{Wang2019, Cai2019},
the transition from spatially homogeneous case to spatially inhomogeneous case
is not as straightforward as the Fourier spectral method (as discussed in
\cite{Hu2019}). Such a discrete collision term automatically conserves mass,
momentum and energy, while it does not preserve the equilibrium state without
additional numerical tricks. In this work, we propose a remedy of such a
problem, which is an improved version of \cite{Hu2019} with lower time
complexity. The method will be tested in several examples in one spatial
dimension. The results show both good efficiency and high accuracy.

The rest of this paper is organized as follows. Section \ref{sec:Boltzmann} is
a review of the Boltzmann equation and Burnett's expansion of the distribution
function. Our numerical method will be introduced in detail in Section
\ref{sec:algorithm}. Numerical tests will be given in Section
\ref{sec:numerical}, and we conclude this paper by a brief summary in Section
\ref{sec:conclusion}.

\section{Boltzmann equation and the Burnett spectral method}
\label{sec:Boltzmann}
To better describe our numerical algorithm, we would like to first clarify the
mathematical model to be solved and the framework of the numerical method. Some
relevant previous research works will also be reviewed in this section.

\subsection{Boltzmann equation}
The governing equation for the rarefied gas flow is the Boltzmann equation:
\begin{displaymath}
\frac{\partial f}{\partial t} + \bv \cdot \nabla_{\bx} f = Q[f,f],
\end{displaymath}
where $\bx = (x,y,z)^T \in \mathbb{R}^3$ is the spatial variable, $\bv \in (v_x,
v_y, v_z)^T \in \mathbb{R}^3$ is the velocity variable, and $f(t,\bx,\bv)$ is
the distribution function. The right-hand side $Q[f,f]$ is the collision term,
which will be detailed later. In this paper, we restrict ourselves to one
spatial dimension, so that the Boltzmann equation can be written as
\begin{equation} \label{eq:1D_Boltzmann}
\frac{\partial f}{\partial t} + v_x \frac{\partial f}{\partial x} = Q[f,f].
\end{equation}
With the initial condition
\begin{equation} \label{eq:init}
f(0,x,\bv) = f_0(x,\bv),
\end{equation}
the equation \eqref{eq:1D_Boltzmann} holds for any $t > 0$, $x \in I \subset
\mathbb{R}$, and $\bv \in \mathbb{R}^3$, where the interval $I = (a,b)$ can
either be finite or infinite. When $a > -\infty$, we consider Maxwell's wall
boundary condition at $x=a$. The solid wall at $x = a$ may have a velocity
$\bu_a^W(t)$. In this paper, we assume that the $x$-component of $\bu_a^W(t)$
is zero so that the computational domain does not change. Maxwell's wall
boundary condition assumes that among all the particles that hit the wall, some
particles undergo specular reflection, and others undergo diffusive reflection.
The proportion of the particles that undergo diffusive reflection is called the
accommodation coefficient $\chi_a$. For a solid wall with temperature
$T_a^W(t)$ at time $t$, the boundary condition can be formulated as
\begin{equation} \label{eq:Maxwell}
f(t,a,\bv) = \chi_a f_a^W(t,\bv) + (1-\chi_a) f(t,a,\overline{\bv})
  \qquad \text{for} \quad v_x > 0,
\end{equation}
where $\overline{\bv} = (-v_x, v_y, v_z)^T$ and
\begin{equation} \label{eq:wall_Maxwellian}
f_a^W(t,\bv) = \frac{n_a^W(t)}{[2\pi \theta_a^W(t)]^{3/2}}
  \exp \left( -\frac{|\bv - \bu_a^W(t)|^2}{2\theta_a^W(t)} \right).
\end{equation}
In \eqref{eq:wall_Maxwellian}, the quantity $\theta_a^W(t)$ is the temperature
of the gas represented in the unit of specific energy, defined by
\begin{displaymath}
\theta_a^W(t) = k_B T_a^W(t) / \mm,
\end{displaymath}
where $k_B$ is the Boltzmann constant, and $\mm$ is the mass of a single
molecule.  The quantity $n_a^W(t)$ is chosen such that the ``no mass flux''
boundary condition is satisfied. Its precise expression is
\begin{equation} \label{eq:nW}
n_a^W(t) = \sqrt{\frac{2\pi}{\theta_a^W(t)}}
  \int_{-\infty}^{+\infty} \int_{-\infty}^{+\infty} \int_{-\infty}^0
    v_x f(t,a,\bv) \,\mathrm{d}v_x \,\mathrm{d}v_y \,\mathrm{d}v_z.
\end{equation}
Similarly, if $b <
+\infty$ and the solid wall at $x=b$ has velocity $\bu_b^W(t)$, temperature
$\theta_b^W(t)$, and accommodation coefficient $\chi_b$, then the boundary
condition at $x=b$ is
\begin{displaymath}
f(t,b,\bv) = \frac{\chi_b n_b^W(t)}{[2\pi \theta_b^W(t)]^{3/2}}
  \exp \left( -\frac{|\bv - \bu_b^W(t)|^2}{2\theta_b^W(t)} \right)
  + (1-\chi_b) f(t,b,\overline{\bv}) \qquad \text{for} \quad v_x < 0,
\end{displaymath}
where $n_b^W(t)$ is given by
\begin{displaymath}
n_b^W(t) = \sqrt{\frac{2\pi}{\theta_b^W(t)}}
  \int_{-\infty}^{+\infty} \int_{-\infty}^{+\infty} \int_0^{+\infty}
    v_x f(t,b,\bv) \,\mathrm{d}v_x \,\mathrm{d}v_y \,\mathrm{d}v_z.
\end{displaymath}
Here we remind the readers that the boundary conditions need to be prescribed
only for a half of the distribution function which describes the particles
moving into the domain $I$.

The collision term $Q[f,f]$ is the most complicated part in the Boltzmann
equation, whose general form is
\begin{equation} \label{eq:collision}
Q[f,f](t,x,\bv) = \int_{\mathbb{R}^3} \int_{\mathbb{S}_+^2}
  B(\bv - \bv_1, \bn) [f(t,x,\bv_1') f(t,x,\bv') - f(t,x,\bv_1) f(t,x,\bv)]
  \,\mathrm{d}\bn \,\mathrm{d}\bv_1.
\end{equation}
Here $\bv_1'$ and $\bv'$ are post-collisional velocities:
\begin{displaymath}
\bv_1' = \bv_1 - [(\bv_1 - \bv) \cdot \bn] \bn, \qquad
\bv' = \bv - [(\bv - \bv_1) \cdot \bn] \bn,
\end{displaymath}
and $B(\cdot,\cdot)$ is the collision kernel determined by the potential
function between gas molecules. It can be seen that for any function $g(\bv)$
defined on the velocity space, the collision term $Q[g,g](\bv)$ can be defined
similar to \eqref{eq:collision} with $t$ and $x$ removed. For any distribution
function $g(\bv)$, the corresponding collision term satisfies the following
conservation property:
\begin{displaymath}
\int_{\mathbb{R}^3}
  \begin{pmatrix} 1 \\ \bv \\ \frac{1}{2} |\bv|^2 \end{pmatrix}
  Q[g,g](\bv) \,\mathrm{d}\bv = 0,
\end{displaymath}
which indicates the local conservation of mass, momentum and energy. Another
important property is
\begin{displaymath}
Q[\mathcal{M}, \mathcal{M}] = 0  \quad \Longleftrightarrow \quad
  \mathcal{M}(\bv) = \exp(\alpha + \boldsymbol{\beta} \cdot \bv + \gamma |\bv|^2)
  \text{ and } \gamma < 0.
\end{displaymath}
It shows that the manifold of local equilibrium states are formed by isotropic
Gaussian distribution functions, which are also called Maxwellians.

In this paper, we will mainly focus on the specific collision models induced by
the inverse power laws, in which the force between two molecules is always
repulsive, and the magnitude is proportional to a certain negative power of
the distance between them ($F = \kappa r^{-\eta}$ with $F$, $r$, $\eta$
and $\kappa$ being the force, distance, power index and the intensity constant,
respectively). By choosing a reference velocity $\lu$ and reference temperature
$\overline{T}$, it has been demonstrated in \cite{Hu2019} that the collision
term for inverse power laws can be nondimensionalized by
\begin{equation} \label{eq:nondim}
Q[g,g](\lu + \sqrt{\lth} \bv^*) = \frac{\rho^2}{\mm^2 \lth^{\frac{3}{2}}}
  \left( \frac{2\kappa}{\mm} \right)^{\frac{2}{\eta-1}}
  \lth^{\frac{\eta-5}{2(\eta-1)}} Q^*[g^*,g^*](\bv^*),
\end{equation}
where $\lth = k_B \overline{T} / \mm$, and the density of the gas $\rho$ as
well as the nondimensionalized distribution function $g^*$ are given by
\begin{equation} \label{eq:g_star}
\rho = \mm \int_{\mathbb{R}^3} g(\bv) \,\mathrm{d}\bv, \qquad
g^*(\bv^*) = \frac{\mm \lth^{3/2}}{\rho} g(\lu + \sqrt{\lth} \bv^*).
\end{equation}
In \eqref{eq:nondim}, the dimensionless collision operator $Q^*$ is independent
of $\kappa$ and $\mm$. The only parameter in $Q^*$ is the power index $\eta$.
This formula can be further simplified by introducing the reference viscosity
coefficient (see \eqref{eq:A_2} in the appendix for the definition of
$A_2(\eta)$):
\begin{equation} \label{eq:mu}
\overline{\mu} = \frac{5\mm(\lth/\pi)^{1/2} (2\mm\lth/\kappa)^{2/(\eta-1)}}%
  {8 A_2(\eta) \Gamma(4-2(\eta-1))},
\end{equation}
by which we find that
\begin{displaymath}
Q[g,g](\lu + \sqrt{\lth} \bv^*) = \frac{\rho^2}{\mm \overline{\mu} \sqrt{\lth}}
  \,\hat{Q}[g^*,g^*](\bv^*),
\end{displaymath}
where $\hat{Q}$ is the new dimensionless collision operator given by
\begin{displaymath}
\hat{Q}[g^*,g^*] = \frac{5}%
  {4^{\frac{3}{2} - \frac{2}{\eta-1}} \sqrt{\pi} A_2(\eta)\Gamma(4-2(\eta-1))}
  Q^*[g^*,g^*].
\end{displaymath}
As will be seen later, such a nondimensionalization is closely related to the
choice of parameters in our numerical scheme.

\begin{remark}
Here we have introduced two dimensionless collision terms $Q^*[f^*,f^*]$
and $\hat{Q}[f^*,f^*]$, which differ only by a constant. In \cite{Cai2019,
Wang2019}, the dimensionless collision term $Q^*[f^*,f^*]$ is used, while in
this work, we are going to use $\hat{Q}[f^*,f^*]$ in our further discussion. Since
our numerical method is built based on the work \cite{Cai2019}, we point out
the difference here to avoid confusion.
\end{remark}

\subsection{Burnett spectral method for the spatially homogeneous Boltzmann
equation} \label{sec:spectral}
For the spatially homogeneous Boltzmann equation, the Burnett spectral method
has been introduced in \cite{Cai2019}, where the Burnett method is introduced in
the dimensionless setting. In what follows, we will provide a brief review of
this method. Since the flow is assumed to be spatially homogeneous, the
variable $x$ will be omitted temporarily in this subsection.

We will present the method based on the dimensionless collision term
$\hat{Q}[f^*,f^*]$, where $f^*$ is the dimensionless distribution function
defined similar to \eqref{eq:g_star}:
\begin{equation} \label{eq:nondim_f}
f^*(t^*,\bv^*) = \frac{\mm \lth^{3/2}}{\rho}
  f \left( \frac{\overline{\mu} t^*}{\rho \lth}, \lu + \sqrt{\lth} \bv^* \right),
  \qquad \rho = \mm \int_{\mathbb{R}^3} f(t,\bv) \,\mathrm{d}\bv,
\end{equation}
Note that $\rho$ is independent of $t$ since collision does not change the
number density. To write down the spectral expansion, we first define the
Burnett polynomials \cite{Burnett1936}:
\begin{equation} \label{eq:p_lmn}
p_{lmn}(\bv^*) = \sqrt{\frac{2^{1-l} \pi^{3/2} n!}{\Gamma(n+l+3/2)}}
  L_n^{(l+1/2)} \left( \frac{|\bv^*|^2}{2} \right) |\bv^*|^l
  Y_l^m \left( \frac{\bv^*}{|\bv^*|} \right),
\qquad l,n = 0,1,\cdots,\quad m = -l,\cdots,l.
\end{equation}
where $L_n^{(\alpha)}(\cdot)$ and $Y_l^m(\cdot)$ are, respectively, the Laguerre
polynomials and spherical harmonics, whose definitions are given in detail in
the appendix (see \eqref{eq:LY}). Let $\omega(\bv^*)$ be the three-dimensional
standard normal distribution
\begin{displaymath}
\omega(\bv^*) = \frac{1}{(2\pi)^{3/2}}
  \exp \left( -\frac{|\bv^*|^2}{2} \right).
\end{displaymath}
Then the following orthogonality holds:
\begin{displaymath}
\int_{\mathbb{R}^3} [p_{l_1 m_1 n_1}(\bv^*)]^{\dagger}
  p_{l_2 m_2 n_2}(\bv^*) \omega(\bv^*) \,\mathrm{d}\bv^*
  = \delta_{l_1 l_2} \delta_{m_1 m_2} \delta_{n_1 n_2},
\end{displaymath}
where $\dagger$ denotes the complex conjugate. The Petrov-Galerkin spectral
method can be derived by approximating the dimensionless distribution function
$f^*(t^*,\bv^*)$ by
\begin{displaymath}
f^*_M(t^*,\bv^*) =
  \sum_{l=0}^M \sum_{m=-l}^l \sum_{n=0}^{\lfloor (M-l)/2 \rfloor}
  \tilde{f}_{lmn}^*(t^*) p_{lmn}(\bv^*) \omega(\bv^*),
\end{displaymath}
and then obtaining equations for the coefficients by calculating
\begin{multline} \label{eq:spectral}
\int_{\mathbb{R}^3} [p_{lmn}(\bv^*)]^{\dagger}
  \, \frac{\partial f^*_M(t^*,\bv^*)}{\partial t^*} \,\mathrm{d}\bv^* =
\int_{\mathbb{R}^3} [p_{lmn}(\bv^*)]^{\dagger}
  \, \hat{Q}[f^*_M, f^*_M](t^*,\bv^*) \,\mathrm{d}\bv^*, \\
l = 0,\cdots,M, \qquad m = -l,\cdots,l, \qquad
  n = 0,\cdots,\lfloor (M-l)/2 \rfloor.
\end{multline}
The general result is
\begin{equation} \label{eq:ode}
\frac{\mathrm{d} \tilde{f}_{lmn}^*}{\mathrm{d}t^*} =
  \sum_{l_1=0}^M \sum_{m_1=-l_1}^{l_1} \sum_{n_1=0}^{\lfloor (M-l_1)/2 \rfloor}
  \sum_{l_2=0}^M \sum_{m_2=-l_2}^{l_2} \sum_{n_2=0}^{\lfloor (M-l_2)/2 \rfloor}
  A_{lmn}^{l_1 m_1 n_1, l_2 m_2 n_2}
  \tilde{f}_{l_1 m_1 n_1}^* \tilde{f}_{l_2 m_2 n_2}^*,
\end{equation}
where the constant coefficients $A_{lmn}^{l_1 m_1 n_1, l_2 m_2 n_2}$ depend on
the collision model. In \cite{Wang2019, Cai2019}, the authors introduced an
algorithm computing these coefficients for all inverse power law models, which
is briefly described in the appendix. The computational cost of \eqref{eq:ode}
with all $l,m,n$ is $O(M^8)$ since $A_{lmn}^{l_1 m_1 n_1, l_2 m_2 n_2} = 0$ when
$m \neq m_1 + m_2$; and the discretization \eqref{eq:ode} automatically
conserves mass, momentum and energy.

Our numerical method is based on the discretization \eqref{eq:ode}. However,
such discretization is not readily applicable for the spatially inhomogeneous
case. The main reason is that \eqref{eq:ode} does not preserve the equilibrium
state. In detail, when $f^*(0,\bv^*) = \exp(\alpha + \boldsymbol{\beta} \cdot
\bv^* + \gamma |\bv^*|^2)$, after computing $f_M^*(0,\bv^*)$ by projection, the
right-hand side of \eqref{eq:ode} is nonzero. As mentioned in the introduction,
one of the main contributions of this paper is to fix such a problem.

\section{Numerical method} \label{sec:algorithm}
In this section, our numerical method will be provided in detail. To begin with,
we will resume the discussion at the end of Section \ref{sec:spectral}, and
develop an algorithm which preserves the steady state.

\subsection{Restoring the dimension and preserving the steady state}
In this section, we will still keep the spatial variable $x$ omitted, and focus
on the homogeneous Boltzmann equation. From \eqref{eq:spectral}, it can be seen
that the spectral method in Section \ref{sec:spectral} works only when $f^*(t^*, \cdot)
\in L^2(\mathbb{R}^3; \omega^{-1} \,\mathrm{d}\bv^*)$ for all $t^*$, i.e.,
\begin{displaymath}
\int_{\mathbb{R}^3} |f^*(t^*, \bv^*)|^2 [\omega(\bv^*)]^{-1} \,\mathrm{d}\bv^*
  < +\infty, \qquad \forall t^* \geqslant 0.
\end{displaymath}
By \eqref{eq:nondim_f}, it can be seen that the original distribution function
$f(t,\bv)$ must satisfy
\begin{equation} \label{eq:scaling}
\int_{\mathbb{R}^3} |f(t, \bv)|^2
  \left[ \omega \left( \frac{\bv - \lu}{\sqrt{\lth}} \right) \right]^{-1}
\,\mathrm{d}\bv < +\infty, \qquad \forall t \geqslant 0.
\end{equation}
This equation shows that the parameter $\lth$ is not only a parameter in the
nondimensionalization, but also playing the role of the scaling factor in the
spectral method for problems on unbounded domains \cite{Tang1993}. It is easy to
see that when $\lth$ is larger, the equation \eqref{eq:scaling} allows more
distribution functions. Therefore, we need to choose a sufficiently large $\lth$
to include all possible distribution functions. The existence of such $\lth$ has
been theoretically guaranteed in \cite{Alonso2017}. To emphasize the role of
$\lth$ in our algorithm, we will present our algorithm using the original
distribution function $f(t,\bv)$. Thus the approximate distribution function is
\begin{equation} \label{eq:f_M}
f_M(t,\bv) = \sum_{l=0}^M \sum_{m=-l}^l \sum_{n=0}^{\lfloor (M-l)/2 \rfloor}
  \tilde{f}_{lmn}(t) p_{lmn}^{[\lu,\lth]}(\bv) \omega^{[\lu,\lth]}(\bv),
\end{equation}
where
\begin{equation} \label{eq:p_omega}
p_{lmn}^{[\lu,\lth]}(\bv) = \lth^{-(l+2n)/2}
  p_{lmn} \left( \frac{\bv - \lu}{\sqrt{\lth}} \right), \qquad
\omega^{[\lu,\lth]}(\bv) = \frac{1}{\mm (2\pi \lth)^{3/2}}
  \exp \left( -\frac{|\bv - \lu|^2}{2\lth} \right).
\end{equation}
The discretization \eqref{eq:ode} becomes
\begin{displaymath}
\frac{\mathrm{d} \tilde{f}_{lmn}}{\mathrm{d}t} = \frac{\,\lth\,}{\overline{\mu}}
  \sum_{l_1=0}^M \sum_{m_1=-l_1}^{l_1} \sum_{n_1=0}^{\lfloor (M-l_1)/2 \rfloor}
  \sum_{l_2=0}^M \sum_{m=-l_2}^{l_2} \sum_{n_2=0}^{\lfloor (M-l_2)/2 \rfloor}
    A_{lmn}^{l_1 m_1 n_1, l_2 m_2 n_2} \lth^{n - n_1 - n_2 + (l - l_1 - l_2)/2}
    \tilde{f}_{l_1 m_1 n_1} \tilde{f}_{l_2 m_2 n_2}.
\end{displaymath}
For simplicity, the right-hand side of the above equation will be named
$\tilde{Q}_{lmn}[f_M]$ hereafter.

Now we are going to change the right-hand side of the above scheme such that the
method preserves the steady state. When the distribution function $f(t,\bv)$ is
a Maxwellian $\mathcal{M}(t,\bv)$, the Maxwellian can be determined by the first
few coefficients in the series expansion:
\begin{displaymath}
\mathcal{M}(t,\bv) = \frac{\rho(t)}{\mm [2\pi \theta(t)]^{3/2}}
  \exp \left( -\frac{|\bv - \bu(t)|^2}{2 \theta(t)} \right),
\end{displaymath}
where
\begin{gather*}
\rho(t) = \tilde{f}_{000}(t), \quad
\bu(t) = \lu(t) + \left( \frac{\tilde{f}_{100}(t)}{\tilde{f}_{000}(t)}, ~
  - \sqrt{2} \re \frac{\tilde{f}_{110}(t)}{\tilde{f}_{000}(t)}, ~
  \sqrt{2} \im \frac{\tilde{f}_{110}(t)}{\tilde{f}_{000}(t)} \right)^T, \\
\theta(t) = \lth - \sqrt{\frac{2}{3}} \frac{\tilde{f}_{001}(t)}{\tilde{f}_{000}(t)}
  - \frac{|\tilde{f}_{100}(t)|^2 + 2|\tilde{f}_{110}(t)|^2}{3|\tilde{f}_{000}(t)|^2}.
\end{gather*}
Since $Q[\mathcal{M},\mathcal{M}]$ always equals zero, one can rewrite
the collision term as $Q[f,f] - Q[\mathcal{M},\mathcal{M}]$.
Thereby, the corresponding discretization turns out to be
\begin{displaymath}
\frac{\mathrm{d} \tilde{f}_{lmn}}{\mathrm{d}t} =
  \tilde{Q}_{lmn}[f_M] - \tilde{Q}_{lmn}[\mathcal{M}_M],
\end{displaymath}
where $\mathcal{M}_M$ is the projection of the Maxwellian:
\begin{displaymath}
\mathcal{M}_M(t,\bv) =
  \sum_{l=0}^M \sum_{m=-l}^l \sum_{n=0}^{\lfloor (M-l)/2 \rfloor}
  \tilde{\mathcal{M}}_{lmn}(t) p_{lmn}^{[\lu,\lth]}(\bv) \omega^{[\lu,\lth]}(\bv),
\end{displaymath}
and the coefficients $\tilde{\mathcal{M}}_{lmn}(t)$ can be obtained by the
following theorem:
\begin{theorem} \label{thm:Maxwellian}
For $\rho, \theta > 0$ and $\bu = (u_x, u_y, u_z) \in \mathbb{R}^3$, let
\begin{equation} \label{eq:general_Maxwellian}
\mathcal{M}(\bv) = \frac{\rho}{\mm(2\pi \theta)^{3/2}}
  \exp \left( -\frac{|\bv - \bu|^2}{2\theta} \right).
\end{equation}
If $\lth > \theta/2$, then for any $\lu = (\bar{u}_x, \bar{u}_y, \bar{u}_z) \in
\mathbb{R}^3$, there exist coefficients $\tilde{\mathcal{M}}_{lmn}$ such that
\begin{equation} \label{eq:M_expansion}
\mathcal{M}(\bv) = \sum_{l=0}^{+\infty} \sum_{m=-l}^l \sum_{n=0}^{+\infty}
  \tilde{\mathcal{M}}_{lmn} p_{lmn}^{[\lu,\lth]}(\bv) \omega^{[\lu,\lth]}(\bv)
\end{equation}
holds in $L^2(\mathbb{R}^3, [\omega^{[\lu,\lth]}(\bv)]^{-1} \,\mathrm{d}\bv)$,
and when $n \geqslant 1$, the coefficients satisfy the following recursive
formula:
\begin{equation} \label{eq:recur}
\begin{split}
\tilde{\mathcal{M}}_{lmn} &= \frac{1}{\sqrt{n} \gamma_{l+1,m}} \Big(
  \sqrt{n+l+3/2} \gamma_{l+2,m} \tilde{\mathcal{M}}_{l+2,m,n-1} +
  (\bar{u}_x - u_x) \tilde{\mathcal{M}}_{l+1,m,n-1} \\
& \qquad + \sqrt{n+l+1/2} \gamma_{l+1,m} (\lth - \theta) \tilde{\mathcal{M}}_{l,m,n-1}
  - \sqrt{n-1} \gamma_{l+2,m} (\lth - \theta) \tilde{\mathcal{M}}_{l+2,m,n-2}
\big),
\end{split}
\end{equation}
where the last term $\tilde{\mathcal{M}}_{l+2,m,n-2}$ is regarded as zero when
$n=1$, and the $\gamma$ symbol is defined by
\begin{equation} \label{eq:gamma}
\gamma_{lm} = \sqrt{\frac{2(l-m)(l+m)}{(2l+1)(2l-1)}}.
\end{equation}
When $n = 0$ and $|m| < l$, the recurrence formula is
\begin{equation} \label{eq:tildeM0}
\tilde{\mathcal{M}}_{lm0} = \frac{1}{\sqrt{l+1/2} \gamma_{lm}} \left(
  (u_x - \bar{u}_x) \tilde{\mathcal{M}}_{l-1,m,0} -
  \sqrt{\frac{1}{l-1/2}} \gamma_{l-1,m}
    \frac{|\bu - \lu|^2}{2} \tilde{\mathcal{M}}_{l-2,m,0}
\right),
\end{equation}
where $\tilde{\mathcal{M}}_{l-2,m,0}$ is regarded as zero if $|m| = l-1$. When
$n=0$ and $m = \pm l$, we have
\begin{equation} \label{eq:tildeMll0}
\tilde{\mathcal{M}}_{ll0} = \sqrt{\frac{1}{2^l l!}} \rho
  [(\bar{u}_y - u_y) - \mathrm{i} (\bar{u}_z - u_z)]^l, \qquad
\tilde{\mathcal{M}}_{l,-l,0} = (-1)^l \tilde{\mathcal{M}}_{ll0}^{\dagger}.
\end{equation}
\end{theorem}
The proof of this theorem is to be found in the appendix. By this theorem, we
see that the computational cost for every coefficient is $O(1)$. Therefore
the time complexity for evaluating all the coefficients
$\tilde{\mathcal{M}}_{lmn}$ with $l+2n \leqslant M$ is $O(M^3)$. The detailed
algorithm is as follows:

\begin{algorithm}[!ht]
\renewcommand{\thealgorithm}{ }
\caption{Calculation of $\tilde{\mathcal{M}}_{lmn}$ for $l+2n \leqslant M$}
\begin{algorithmic}[1]
\begin{minipage}{0.8\textwidth}
\For {$m$ from $0$ to $M$}
\State Compute $\tilde{\mathcal{M}}_{mm0}$ from \eqref{eq:tildeMll0}
\For {$l$ from $m+1$ to $M$}
\State Compute $\tilde{\mathcal{M}}_{lm0}$ from \eqref{eq:tildeM0}
\EndFor
\For {$deg$ from $m$ to $M$}
\For {$n$ from $1$ to $\lfloor (deg-m) / 2 \rfloor$}
\State $l \leftarrow deg - 2n$
\State Compute $\tilde{\mathcal{M}}_{lmn}$ from \eqref{eq:recur}
\EndFor
\EndFor
\EndFor

\For {$m$ from $-M$ to $-1$}
\For {$l$ from $|m|$ to $M$}
\For {$n$ from $0$ to $\lfloor (M-l) / 2 \rfloor$}
\State $\tilde{\mathcal{M}}_{lmn} \leftarrow (-1)^m \tilde{\mathcal{M}}_{l,-m,n}^{\dagger}$
\EndFor
\EndFor
\EndFor
\end{minipage}
\end{algorithmic}
\end{algorithm}

The above algorithm gives a working order of computation to ensure that when
the formulas in Theorem \ref{thm:Maxwellian} are applied, no recursion is
needed.

\subsection{Modelling of the collision term}
A complete algorithm has been described in the above subsection for the
discretization of the collision term. However, the computational complexity for
this algorithm is as high as $O(M^8)$, due to the nine indices appearing in the
coefficients $A_{lmn}^{l_1 m_1 n_1, l_2 m_2 n_2}$ and the constraint $m = m_1 +
m_2$. This makes the simulation difficult when $M$ is large. To reduce the
computational cost, it has been proposed in \cite{Cai2015, Wang2019, Cai2019}
to introduce the ``BGK modelling technique'' to the collision, meaning that we
only apply the quadratic collision to lower moments, which are considered to be
important in the numerical computation, while for higher moments, we model
their evolution by letting them decay to their equilibrium values at a constant
rate. For monatomic gases, the famous BGK model \cite{Bhatnagar1954} can be
derived from the linearized collision model using such an idea. As discussed in
\cite[Section 5.2]{Cai2015}, if the linearized collision model is applied only
up to the second moments (stress tensor), and other higher moments are set as
simple convergence to their equilibrium values, we can obtain the BGK model.
Another well-known model is the Shakhov model \cite{Shakhov1968}, which adds
heat fluxes to the part modeled by the linearized collision operator.

In this work, we are going to adopt the same technique when $M$ is large.
Assume that we want to apply quadratic modelling for the first $M_0$th moments,
where $M_0$ is chosen as a constant positive integer less than or equal to $M$.
Thus the spatially homogeneous Boltzmann equation is discretized by
\begin{equation} \label{eq:model}
\frac{\mathrm{d} \tilde{f}_{lmn}}{\mathrm{d}t} = \overline{Q}_{lmn}^{M_0} :=
\left\{ \begin{array}{ll}
  \tilde{Q}_{lmn}^{M_0} [f_{M_0}] - \tilde{Q}_{lmn}^{M_0} [\mathcal{M}_{M_0}],
    & \text{if } l+2n \leqslant M_0, \\[5pt]
  \nu (\tilde{\mathcal{M}}_{lmn} - \tilde{f}_{lmn}),
    & \text{if } l+2n > M_0.
\end{array} \right.
\end{equation}
In the equations for $\tilde{f}_{lmn}$ with $l+2n > M_0$, the coefficient $\nu$
indicates the rate of convergence to the equilibrium value
$\tilde{\mathcal{M}}_{lmn}$, which is chosen following \cite{Cai2015,Hu2019} as
\begin{displaymath}
\nu = \frac{\rho \theta}{\overline{\mu}}
  \left( \frac{\lth}{\theta} \right)^{\frac{1}{2} + \frac{2}{\eta-1}}\varrho_{M_0}.
\end{displaymath}
To define $\varrho_{M_0}$ in the above equation, we first define a sequence of
matrices $A^l = (a_{nn'}^l) \in \mathbb{R}^{(N_l+1) \times (N_l+1)}$ for $l =
0,1,\cdots,M_0$, where
\begin{displaymath}
N_l = \lfloor (M_0-l)/2 \rfloor, \qquad
a_{nn'}^l = A_{l0n}^{000,l0n'} + A_{l0n}^{l0n',000}, \quad n,n'=0,1,\cdots,N_l.
\end{displaymath}
Thus the definition of $\varrho_{M_0}$ is
\begin{displaymath}
\varrho_{M_0} = \max \{ \varrho(A^l) \mid l = 0,1,\cdots,M_0 \},
\end{displaymath}
where $\varrho(\cdot)$ is the spectral radius of the matrix. As is detailed in
\cite{Cai2015}, such a $\nu$ is in fact the spectral radius of truncated
linearized collision operator.

The total computational cost for \eqref{eq:model} is $O(M_0^8 + M^3)$, which
is obviously an improvement of the authors' previous work \cite{Hu2019} using
Hermite polynomials and a different technique to preserve the steady state, where
the computational cost was $O(M_0^9 + M^4)$.

\subsection{Adding back the spatial variable}
From this section, we are going to recover the spatial variable $x$. Thus in
\eqref{eq:f_M}, the function $f_M$ on the left-hand side and the coefficients
$\tilde{f}_{lmn}$ on the right-hand side should contain the parameter $x$. To
discretize the advection term, we just need to compute $v_x \partial_x f_M$:
\begin{equation} \label{eq:adv}
\begin{split}
v_x \frac{\partial f_M(t,x,\bv)}{\partial x} &=
  \sum_{l=0}^M \sum_{m=-l}^l \sum_{n=0}^{\lfloor (M-l)/2 \rfloor}
    \frac{\partial \tilde{f}_{lmn}(t,x)}{\partial x}
    \left[ v_x p_{lmn}^{[\lu,\lth]}(\bv) \right] \omega^{[\lu,\lth]}(\bv),
\end{split}
\end{equation}
where the term in the square bracket can be expanded by
\begin{equation} \label{eq:recursion}
\begin{split}
v_x p_{lmn}^{[\lu,\lth]}(\bv) = \overline{u}_x p_{lmn}^{[\lu,\lth]}(\bv) &+
  \lth \left(
    \sqrt{n+l+3/2} \gamma_{l+1,m} p_{l+1,m,n}^{[\lu,\lth]}(\bv) -
    \sqrt{n+1} \gamma_{-l,m} p_{l-1,m,n+1}^{[\lu,\lth]}(\bv)
  \right) \\
  & +\sqrt{n+l+1/2} \gamma_{-l,m} p_{l-1,m,n}^{[\lu,\lth]}(\bv) -
    \sqrt{n} \gamma_{l+1,m} p_{l+1,m,n-1}^{[\lu,\lth]}(\bv).
\end{split}
\end{equation}
By now, we can combine \eqref{eq:f_M}, \eqref{eq:adv} and \eqref{eq:recursion}
to get the complete semidiscrete equations:
\begin{multline*}
\frac{\partial \tilde{f}_{lmn}}{\partial t} +
  \overline{u}_x \frac{\partial \tilde{f}_{lmn}}{\partial x} + \lth \left(
    \sqrt{n+l+1/2} \gamma_{lm} \frac{\partial \tilde{f}_{l-1,m,n}}{\partial x}
    - \sqrt{n} \gamma_{-l-1,m} \frac{\partial \tilde{f}_{l+1,m,n-1}}{\partial x}
  \right) \\
+ \sqrt{n+l+3/2} \gamma_{-l-1,m}
  \frac{\partial \tilde{f}_{l+1,m,n}}{\partial x} -
  \sqrt{n+1} \gamma_{lm} \frac{\partial \tilde{f}_{l-1,m,n+1}}{\partial x} =
  \overline{Q}_{lmn}^{M_0}, \\
l = 0,1,\cdots,M, \quad m=-l,\cdots,l,
  \quad n = 0,1,\cdots,\lfloor (M-l)/2 \rfloor.
\end{multline*}
Here $\tilde{f}_{l'm'n'}$ is regarded as zero when its indices are not in the
range given by the last line of the above equations. Let $\bdf$ denote the
vector whose components are all the coefficients $\tilde{f}_{lmn}$ appearing in
the expansion of $f_M$. Then the above equations can be written as
\begin{equation} \label{eq:hyp}
\frac{\partial \bdf}{\partial t} + {\bf A} \frac{\partial \bdf}{\partial x} =
  \boldsymbol{Q}(\bdf),
\end{equation}
where $\bf A$ is a sparse constant matrix whose diagonal entries are
$\overline{u}_x$. And each row of $\bf A$ has at most five nonzero entries.

To complete the problem, we need to add initial and boundary conditions for
\eqref{eq:hyp}. The initial condition can be obtained by expanding
\eqref{eq:init} into series. Alternatively, we can use the orthogonality of
basis functions to write down the initial condition as
\begin{displaymath}
\tilde{f}_{lmn}(0,x) = \mm \lth^{l+2n} \int_{\mathbb{R}^3}
  \left[ p_{lmn}^{[\lu,\lth]}(\bv) \right]^{\dagger} f_0(x,\bv) \,\mathrm{d}\bv,
\quad l=0,1,\cdots,M, \quad m = -l,\cdots,l,
  \quad n = 0,1,\cdots, \left\lfloor \frac{M-l}{2} \right\rfloor.
\end{displaymath}
When the solid wall exists in the problem, we need to formulate wall boundary
conditions for \eqref{eq:hyp}, which will be detailed in the next subsection.

\subsection{Wall boundary conditions}
We only consider the wall boundary condition of \eqref{eq:hyp} at $x=a$. The
basic idea is the same as the construction of initial condition. We multiply
\eqref{eq:Maxwell} by $\mm \lth^{l+2n} \left[ p_{lmn}^{[\lu,\lth]}(\bv)
\right]^{\dagger}$ and integrate with respect to $\bv$. Note that
\eqref{eq:Maxwell} holds only for $v_x > 0$, and therefore the integral domain
is the half space:
\begin{equation} \label{eq:bc}
\begin{split}
& \mm \lth^{l+2n} \int_{-\infty}^{+\infty} \int_{-\infty}^{+\infty} \int_0^{+\infty}
  \left[ p_{lmn}^{[\lu,\lth]}(\bv) \right]^{\dagger} f_M(t,a,\bv)
  \,\mathrm{d}v_x \,\mathrm{d}v_y \,\mathrm{d}v_z = \\
& \qquad \mm \lth^{l+2n}
  \int_{-\infty}^{+\infty} \int_{-\infty}^{+\infty} \int_0^{+\infty}
  \left[ p_{lmn}^{[\lu,\lth]}(\bv) \right]^{\dagger}
  \left( \chi_a f_a^W(t,a,\bv) + (1-\chi_a) f_M(t,a,\overline{\bv}) \right)
  \,\mathrm{d}v_x \,\mathrm{d}v_y \,\mathrm{d}v_z.
\end{split}
\end{equation}
Here, the distribution function $f(t,a,\bv)$ in \eqref{eq:Maxwell} has been
replaced by the numerical solution $f_M(t,a,\bv)$, and the ``wall Maxwellian''
$f_a^W(t,\bv)$ is still defined by \eqref{eq:wall_Maxwellian}, while in the
definition of $n_a^W(t)$ \eqref{eq:nW}, the distribution function $f(t,a,\bv)$
should be again replaced by $f_M(t,a,\bv)$. However, \eqref{eq:bc} does not
complete the statement of the boundary conditions, since if \eqref{eq:bc} with
all $l,m,n$ satisfying $l+2n \leqslant M$ are imposed at $x=a$, the number of
boundary conditions will exceed the number required by the hyperbolicity, which
should be the number of characteristics pointing into the domain. To fix the
issue, we first choose $\lu$ such that $\overline{u}_x = 0$. Thus, as in
\cite{Grad1949, Cai2018, Hu2019}, we can get the correct number of boundary
conditions if we only take into account \eqref{eq:bc} for $l,m,n$ satisfying
$p_{lmn}^{[\lu,\lth]}(\bv) = -p_{lmn}^{[\lu,\lth]}(\overline{\bv})$ and $l+2n
\leqslant M$. The symmetry condition requires that $l+m$ must be odd. Below we
are going to omit the spatial variable, which is fixed at $x=a$.

To make the boundary conditions \eqref{eq:bc} more explicit, we adopt the idea
in \cite{Torrilhon2015} to split the distribution function into an odd part and
an even part:
\begin{align}
f_M^{\mathrm{(odd)}}(t,\bv) &= \frac{f_M(\bv) - f_M(\overline{\bv})}{2} =
  \sum_{l=0}^M \sum_{\substack{m=-l\\l+m\text{ odd}}}^l
    \sum_{n=0}^{\lfloor (M-l)/2 \rfloor} \tilde{f}_{lmn}(t)
      p_{lmn}^{[\lu,\lth]}(\bv) \omega^{[\lu,\lth]}(\bv), \\
f_M^{\mathrm{(even)}}(t,\bv) &= \frac{f_M(\bv) + f_M(\overline{\bv})}{2} =
  \sum_{l=0}^M \sum_{\substack{m=-l\\l+m\text{ even}}}^l
    \sum_{n=0}^{\lfloor (M-l)/2 \rfloor} \tilde{f}_{lmn}(t)
      p_{lmn}^{[\lu,\lth]}(\bv) \omega^{[\lu,\lth]}(\bv).
\end{align}
Thus the boundary condition \eqref{eq:bc} can be rewritten as
\begin{equation} \label{eq:odd_even}
\begin{split}
& \mm \lth^{l+2n} \int_{-\infty}^{+\infty} \int_{-\infty}^{+\infty} \int_0^{+\infty}
  \left[ p_{lmn}^{[\lu,\lth]}(\bv) \right]^{\dagger} f_M^{\mathrm{(odd)}}(t,\bv)
  \,\mathrm{d}v_x \,\mathrm{d}v_y \,\mathrm{d}v_z = \\
& \qquad \frac{\chi_a \mm \lth^{l+2n}}{2-\chi_a}
  \int_{-\infty}^{+\infty} \int_{-\infty}^{+\infty} \int_0^{+\infty}
  \left[ p_{lmn}^{[\lu,\lth]}(\bv) \right]^{\dagger}
  \left( f_a^W(t,\bv) - f_M^{\mathrm{(even)}}(t,\bv) \right)
  \,\mathrm{d}v_x \,\mathrm{d}v_y \,\mathrm{d}v_z, \qquad l+m \text{ is odd}.
\end{split}
\end{equation}
Further simplification requires the following result:
\begin{theorem} \label{lem:integral}
Suppose $l+m$ is odd. Then
\begin{displaymath}
\mm \lth^{l+2n} \int_{-\infty}^{+\infty} \int_{-\infty}^{+\infty} \int_0^{+\infty}
  \left[ p_{lmn}^{[\lu,\lth]}(\bv) \right]^{\dagger} f_M^{\mathrm{(odd)}}(t,\bv)
  \,\mathrm{d}v_x \,\mathrm{d}v_y \,\mathrm{d}v_z
= \frac{1}{2} \tilde{f}_{lmn}(t).
\end{displaymath}
\end{theorem}
\begin{proof}
Since $l + m$ is odd, we have $p_{lmn}^{[\lu,\lth]}(\bv) =
-p_{lmn}^{[\lu,\lth]}(\overline{\bv})$. Note that $f_M^{\mathrm{(odd)}}(t,\bv)
= -f_M^{\mathrm{(odd)}}(t,\overline{\bv})$, we obtain by change of variables that
\begin{equation} \label{eq:half_mnt}
\begin{split}
& \mm \lth^{l+2n} \int_{-\infty}^{+\infty} \int_{-\infty}^{+\infty} \int_0^{+\infty}
  \left[ p_{lmn}^{[\lu,\lth]}(\bv) \right]^{\dagger} f_M^{\mathrm{(odd)}}(t,\bv)
  \,\mathrm{d}v_x \,\mathrm{d}v_y \,\mathrm{d}v_z \\
={} & \mm \lth^{l+2n} \int_{-\infty}^{+\infty} \int_{-\infty}^{+\infty} \int_{-\infty}^0
  \left[ p_{lmn}^{[\lu,\lth]}(\overline{\bv}) \right]^{\dagger}
  f_M^{\mathrm{(odd)}}(t,\overline{\bv})
  \,\mathrm{d}\overline{v}_x \,\mathrm{d}\overline{v}_y \,\mathrm{d}\overline{v}_z \\
={} & \mm \lth^{l+2n} \int_{-\infty}^{+\infty} \int_{-\infty}^{+\infty} \int_{-\infty}^0
  \left[ p_{lmn}^{[\lu,\lth]}(\bv) \right]^{\dagger} f_M^{\mathrm{(odd)}}(t,\bv)
  \,\mathrm{d}v_x \,\mathrm{d}v_y \,\mathrm{d}v_z.
\end{split}
\end{equation}
Let $J_{lmn}$ be the above quantity. Then by adding up the first and third
lines in \eqref{eq:half_mnt}, we obtain 
\begin{displaymath}
2J_{lmn} = \mm \lth^{l+2n} \int_{\mathbb{R}^3}
  \left[ p_{lmn}^{[\lu,\lth]}(\bv) \right]^{\dagger} f_M^{\mathrm{(odd)}}(t,\bv)
  \,\mathrm{d}\bv = \tilde{f}_{lmn}(t),
\end{displaymath}
which implies the conclusion of the theorem.
\end{proof}
The above theorem gives the left-hand side of \eqref{eq:odd_even}. To proceed, we
first consider a special case $(l,m,n) = (1,0,0)$. In this case,
$p_{100}^{[\lu,\lth]}(\bv) = v_x / \lth$, and the right-hand side of
\eqref{eq:odd_even} can be computed by
\begin{displaymath}
\begin{split}
& \frac{\chi_a \mm}{2-\chi_a}
  \int_{-\infty}^{+\infty} \int_{-\infty}^{+\infty} \int_0^{+\infty}
  v_x \left( f_a^W(t,\bv) - f_M^{\mathrm{(even)}}(t,\bv) \right)
  \,\mathrm{d}v_x \,\mathrm{d}v_y \,\mathrm{d}v_z \\
={} & \frac{\chi_a \mm}{2-\chi_a}
  \int_{-\infty}^{+\infty} \int_{-\infty}^{+\infty} \int_{-\infty}^0
  v_x \left( f_M^{\mathrm{(even)}}(t,\bv) - f_a^W(t,\bv) \right)
  \,\mathrm{d}v_x \,\mathrm{d}v_y \,\mathrm{d}v_z \\
={} & \frac{\chi_a \mm}{2-\chi_a}
  \int_{-\infty}^{+\infty} \int_{-\infty}^{+\infty} \int_{-\infty}^0
  v_x \left( f_M(\bv) - f_M^{\mathrm{(odd)}}(t,\bv) - f_a^W(t,\bv) \right)
  \,\mathrm{d}v_x \,\mathrm{d}v_y \,\mathrm{d}v_z \\
={} & \frac{\chi_a \mm}{2-\chi_a}
  \int_{-\infty}^{+\infty} \int_{-\infty}^{+\infty} \int_{-\infty}^0
  v_x \left( f_M(t,\bv) - f_a^W(t,\bv) \right)
  \,\mathrm{d}v_x \,\mathrm{d}v_y \,\mathrm{d}v_z
  - \frac{\chi_a}{2-\chi_a} \frac{1}{2} \tilde{f}_{100}(t).
\end{split}
\end{displaymath}
Here the first equality uses the symmetry of $f_a^W$ and
$f_M^{\mathrm{(even)}}$; the second equality uses the decomposition of $f_M$;
and the third equality uses Theorem \ref{lem:integral} and the symmetry of
$f_M^{\mathrm{(odd)}}$. Now, by using \eqref{eq:wall_Maxwellian} and
\eqref{eq:nW}, straightforward calculation yields
\begin{displaymath}
\int_{-\infty}^{+\infty} \int_{-\infty}^{+\infty} \int_{-\infty}^0
  v_x \left( f_M(t,\bv) - f_a^W(t,\bv) \right)
  \,\mathrm{d}v_x \,\mathrm{d}v_y \,\mathrm{d}v_z = 0.
\end{displaymath}
Thus the boundary condition \eqref{eq:odd_even} for $(l,m,n) = (1,0,0)$ turns
out to be
\begin{displaymath}
\frac{1}{2} \tilde{f}_{100}(t) = 
  - \frac{\chi_a}{2-\chi_a} \frac{1}{2} \tilde{f}_{100}(t),
\end{displaymath}
which is equivalent to
\begin{displaymath}
\tilde{f}_{100}(t) = 0.
\end{displaymath}
Such a result agrees with the requirement that the perpendicular momentum or
velocity must equal zero. By this result, we also know that
\begin{displaymath}
\begin{split}
n_a^W(t) &= \sqrt{\frac{2\pi}{\theta_a^W(t)}} \left(
  \int_{-\infty}^{+\infty} \int_{-\infty}^{+\infty} \int_{-\infty}^0
    v_x f_M^{\mathrm{(odd)}}(t,\bv) \,\mathrm{d}v_x \,\mathrm{d}v_y \,\mathrm{d}v_z
  + \int_{-\infty}^{+\infty} \int_{-\infty}^{+\infty} \int_{-\infty}^0
    v_x f_M^{\mathrm{(even)}}(t,\bv) \,\mathrm{d}v_x \,\mathrm{d}v_y \,\mathrm{d}v_z
\right) \\
&= \sqrt{\frac{2\pi}{\theta_a^W(t)}} \left(
  \frac{1}{2m} \tilde{f}_{100}(t)
  + \int_{-\infty}^{+\infty} \int_{-\infty}^{+\infty} \int_{-\infty}^0
    v_x f_M^{\mathrm{(even)}}(t,\bv) \,\mathrm{d}v_x \,\mathrm{d}v_y \,\mathrm{d}v_z
\right) \\
&= \sqrt{\frac{2\pi}{\theta_a^W(t)}}
  \int_{-\infty}^{+\infty} \int_{-\infty}^{+\infty} \int_{-\infty}^0
    v_x f_M^{\mathrm{(even)}}(t,\bv) \,\mathrm{d}v_x \,\mathrm{d}v_y \,\mathrm{d}v_z,
\end{split}
\end{displaymath}
which means that the right-hand side of \eqref{eq:odd_even} is completely
independent of $f_M^{\mathrm{(odd)}}$.

By the above results, in general, the equation \eqref{eq:odd_even} can be
written as
\begin{displaymath}
\tilde{f}_{lmn}(t) = \frac{2\chi_a}{2-\chi_a}
  \sum_{l'=0}^M \sum_{\substack{m'=-l'\\ l'+m'\text{ even}}}^{l'}
  \sum_{n'=0}^{\lfloor (M-l')/2 \rfloor} B_{lmn}^{l'm'n'}
  \tilde{f}_{l'm'n'}(t), \qquad l+m \text{ is odd}.
\end{displaymath}
The constants $B_{lmn}^{l'm'n'}$ are given by
\begin{displaymath}
\begin{split}
B_{lmn}^{l'm'n'} &= \tilde{\mathcal{M}}_{lmn}^W \sqrt{\frac{2\pi}{\theta_a^W(t)}}
  \int_{-\infty}^{+\infty} \int_{-\infty}^{+\infty} \int_{-\infty}^0
    v_x p_{l'm'n'}^{[\lu,\lth]} \omega^{[\lu,\lth]}(\bv)
  \,\mathrm{d}v_x \,\mathrm{d}v_y \,\mathrm{d}v_z \\
& \qquad - m \lth^{l+2n}
  \int_{-\infty}^{+\infty} \int_{-\infty}^{+\infty} \int_0^{+\infty}
  \left[ p_{lmn}^{[\lu,\lth]}(\bv) \right]^{\dagger}
  p_{l'm'n'}^{[\lu,\lth]}(\bv) \omega^{[\lu,\lth]}(\bv)
  \,\mathrm{d}v_x \,\mathrm{d}v_y \,\mathrm{d}v_z,
\end{split}
\end{displaymath}
where
\begin{displaymath}
\tilde{\mathcal{M}}_{lmn}^W = \frac{\lth^{l+2n}}{[2\pi \theta_a^W(t)]^{3/2}}
  \int_{-\infty}^{+\infty} \int_{-\infty}^{+\infty} \int_0^{+\infty}
  \left[ p_{lmn}^{[\lu,\lth]}(\bv) \right]^{\dagger}
  \exp \left( -\frac{|\bv - \bu_a^W(t)|^2}{2\theta_a^W(t)} \right)
  \,\mathrm{d}v_x \,\mathrm{d}v_y \,\mathrm{d}v_z.
\end{displaymath}

\section{Numerical examples} \label{sec:numerical}
By now, we are ready to carry out numerical tests to see the performance of the
method. In all our numerical tests, we choose $\eta = 10$, $\mm =
\SI{6.63e-26}{\kg}$, and $\kappa =
\SI{3.46946e-113}{\kg\meter\tothe{11}\per\second\squared}$. To define the
Knudsen number, we employ the variable hard sphere (VHS) model \cite{Bird1994}.
At the reference temperature \SI{273.15}{\kelvin}, if the reference diameter of
the gas molecule is $d = \SI{4.17e-10}{\meter}$, then the viscosity of the VHS
gas matches the viscosity of the IPL gas. Thus the mean free path and the
Knudsen number can be defined by
\begin{equation} \label{eq:mfp_Kn}
\lambda = \frac{\mm}{\sqrt{2} \pi \rho d^2}, \qquad
\mathit{Kn} = \lambda / L,
\end{equation}
where $\rho$ is the reference gas density and $L$ is the characteristic length.
For spatial discretization, we use discontinuous Galerkin method or finite
volume WENO scheme, to be specified below. Both steady-state and unsteady flows
are to be carried out in our numerical tests. Note that although only
$(1+3)$-dimensional flows are simulated, all the examples below are
quite challenging due to the existence of high Mach number or large
temperature ratio, which makes it difficult to capture the profile of
the distribution function over the whole computational domain.

\subsection{Simulation of steady-state flows}
To study the steady-state flows, we start from a given initial state and use
time-stepping to evolve the system for a sufficiently long time. To describe the
stopping criterion, we define
\begin{displaymath}
\|f(t_1, \cdot, \cdot) - f(t_2, \cdot, \cdot)\| := \left(
  \int_a^b \int_{\mathbb{R}^3} |f(t_1, x, \bv) - f(t_2, x, \bv)|^2
    \left[ \omega^{[\lu,\lth]}(\bv) \right]^{-1}
  \,\mathrm{d} \bv \,\mathrm{d}x \right)^{1/2}.
\end{displaymath}
We consider that the steady state is achieved at the $n$th time step if the
numerical solution satisfies
\begin{equation} \label{eq:stop_crit}
\frac{\|f((n+1)\Delta t, \cdot, \cdot) - f(n\Delta t, \cdot, \cdot)\|}
  {\|f(\Delta t, \cdot, \cdot) - f(0, \cdot, \cdot)\|} < \epsilon,
\end{equation}
where $\Delta t$ is the time step.

\subsubsection{Steady shock structure} \label{sec:shock}
The plane wave shock structure is a classical example frequently used to test
the kinetic models or the Boltzmann solver \cite{Xu2011, Timokhin2017}. The
domain is unbounded ($a = -\infty$ and $b = +\infty$) and the boundary
conditions are given by
\begin{align*}
\lim_{x\rightarrow -\infty} f(x,\bv) &= f_a(\bv) :=
  \frac{\rho_a}{\mm (2\pi \theta_a)^{3/2}}
  \exp \left( -\frac{|\bv - \bu_a|^2}{2 \theta_a} \right), \\
\lim_{x\rightarrow +\infty} f(x,\bv) &= f_b(\bv) :=
  \frac{\rho_b}{\mm (2\pi \theta_b)^{3/2}}
  \exp \left( -\frac{|\bv - \bu_b|^2}{2 \theta_b} \right),
\end{align*}
where
\begin{displaymath}
\frac{\rho_b}{\rho_a} = \frac{4\mathit{Ma}^2}{\mathit{Ma}^2+3}, \quad
\bu_a = \left( \sqrt{\frac{5\theta_a}{3}} \mathit{Ma}, 0, 0 \right)^T, \quad
\bu_b = \left( \sqrt{\frac{5\theta_a}{3}}
  \frac{\mathit{Ma}^2+3}{4\mathit{Ma}}, 0, 0 \right)^T, \quad
\frac{\theta_b}{\theta_a} =
  \frac{(5\mathit{Ma}^2 - 1)(\mathit{Ma}^2 + 3)}{16\mathit{Ma}^2},
\end{displaymath}
and $\mathit{Ma}$ is the Mach number of the shock wave. In our numerical tests,
we set
\begin{displaymath}
\rho_a = \SI{9.282e-6}{\kg\per\meter\cubed}, \qquad
\theta_a = k_B / \mm \times \SI{273.15}{\kelvin} =
  \SI{5.688e4}{\meter\squared\per\second\squared}.
\end{displaymath}
To obtain the steady state, we set the initial condition to be
\begin{displaymath}
f(0,x,\bv) = \left\{ \begin{array}{ll}
  f_a(\bv), & \text{if } x < 0, \\
  f_b(\bv), & \text{if } x > 0,
\end{array} \right.
\end{displaymath}
and we evolve the distribution functions until \eqref{eq:stop_crit} is achieved
with $\epsilon = 10^{-6}$. Numerically, the computational domain is set to be
$[-30\lambda, 30\lambda]$, where $\lambda$ is the mean free path defined in
\eqref{eq:mfp_Kn} with $\rho = \rho_a$. The computational domain is divided
into $60$ grid cells, and the fifth-order WENO finite volume method is applied
for the spatial discretization. For the velocity discretization, we use
\begin{displaymath}
\lu = (\bu_a + \bu_b) / 2, \qquad \lth = 0.7\theta_{b}.
\end{displaymath}

Two Mach numbers $\mathit{Ma} = 6.5$ and $\mathit{Ma} = 9.0$ are
considered in our numerical tests. The corresponding solution of
density $\rho$, temperature $T$, normal stress $\sigma_{xx}$ and heat
flux $q_{x}$, obtained by our method with $M_{0}=10$ and $M=20$, are
presented in \figurename~\ref{fig:shock-Ma65-binaryOM20} for
$\mathit{Ma} = 6.5$ and in \figurename~\ref{fig:shock-Ma90-binaryOM20}
for $\mathit{Ma} = 9.0$ respectively. Comparison of them insider the
shock layer are made with the solution obtained by the DSMC method
\cite{Bird1994}. Even for such high Mach number cases, the results
show that the shock structure profiles, including the shock thickness
and the peak values of heat flux and normal stress, coincide perfectly
well between the solution given by our method and the DSMC method.

Both simulations are performed on a cluster with the CPU model Intel
Xeon E5-2697A V4 @ 2.6GHz. Ten threads are used for each
simulation. Details of the simulations, including the number of time
steps and the total elapsed time, are listed in Table
\ref{tab:shock-cputime}, which shows the efficiency of the presented
method.

\begin{figure}[!htb]
  \centering
  {\includegraphics[width=0.49\textwidth,clip]{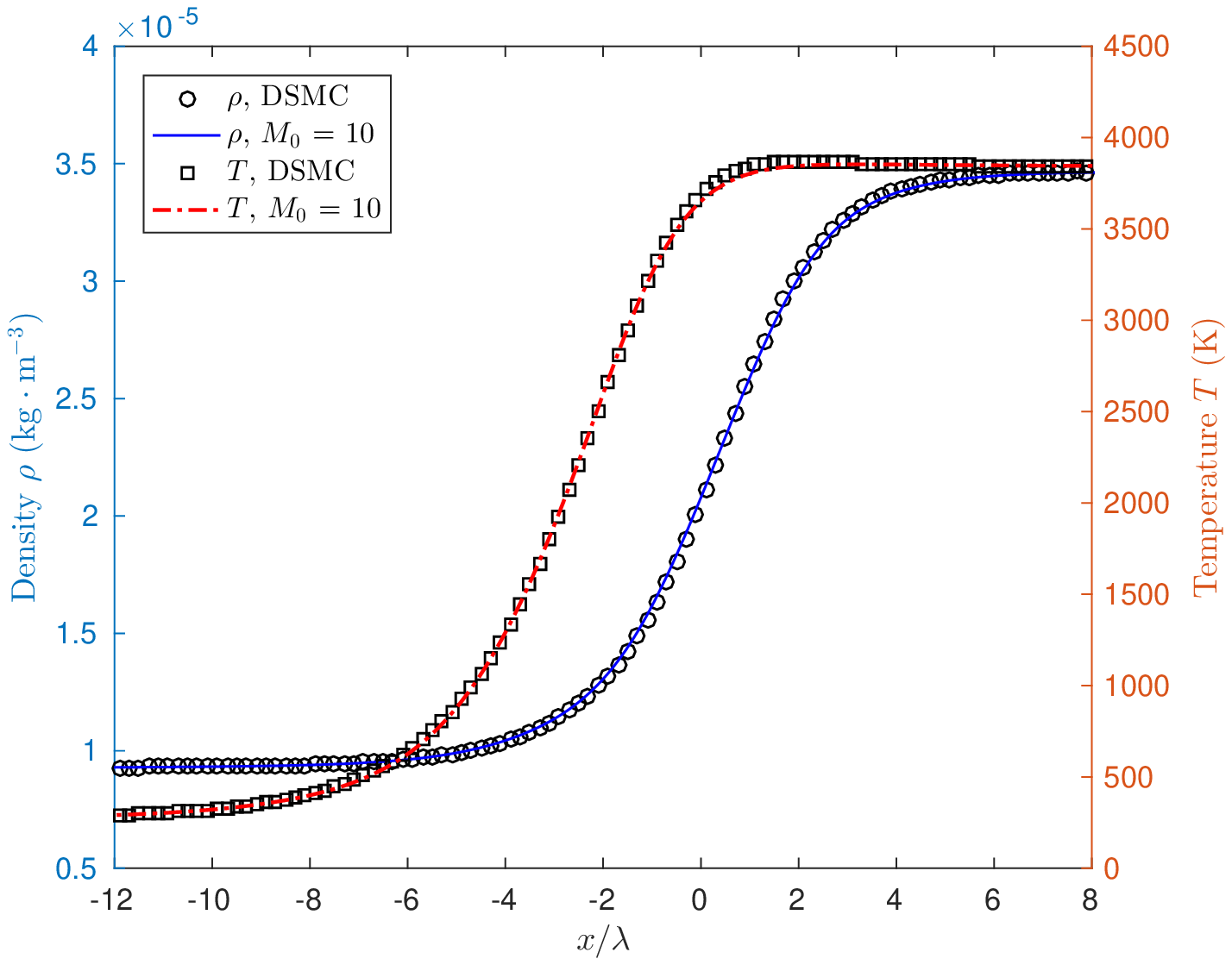}}\hfill
  {\includegraphics[width=0.49\textwidth,clip]{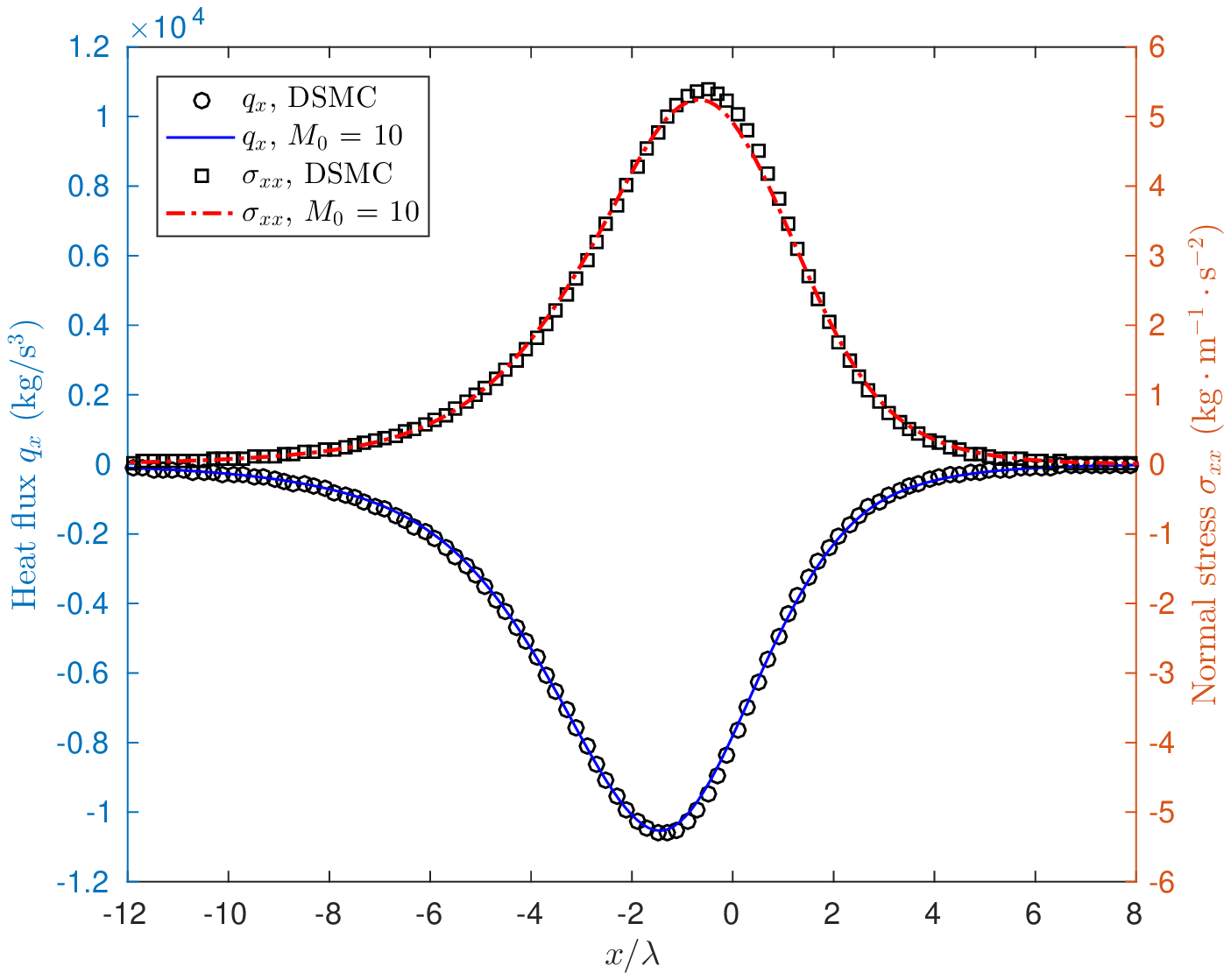}}
  \caption{Solution of the shock structure for $\mathit{Ma} = 6.5$ and $M=20$.}
  \label{fig:shock-Ma65-binaryOM20}
\end{figure}

\begin{figure}[!htb]
  \centering
  {\includegraphics[width=0.49\textwidth,clip]{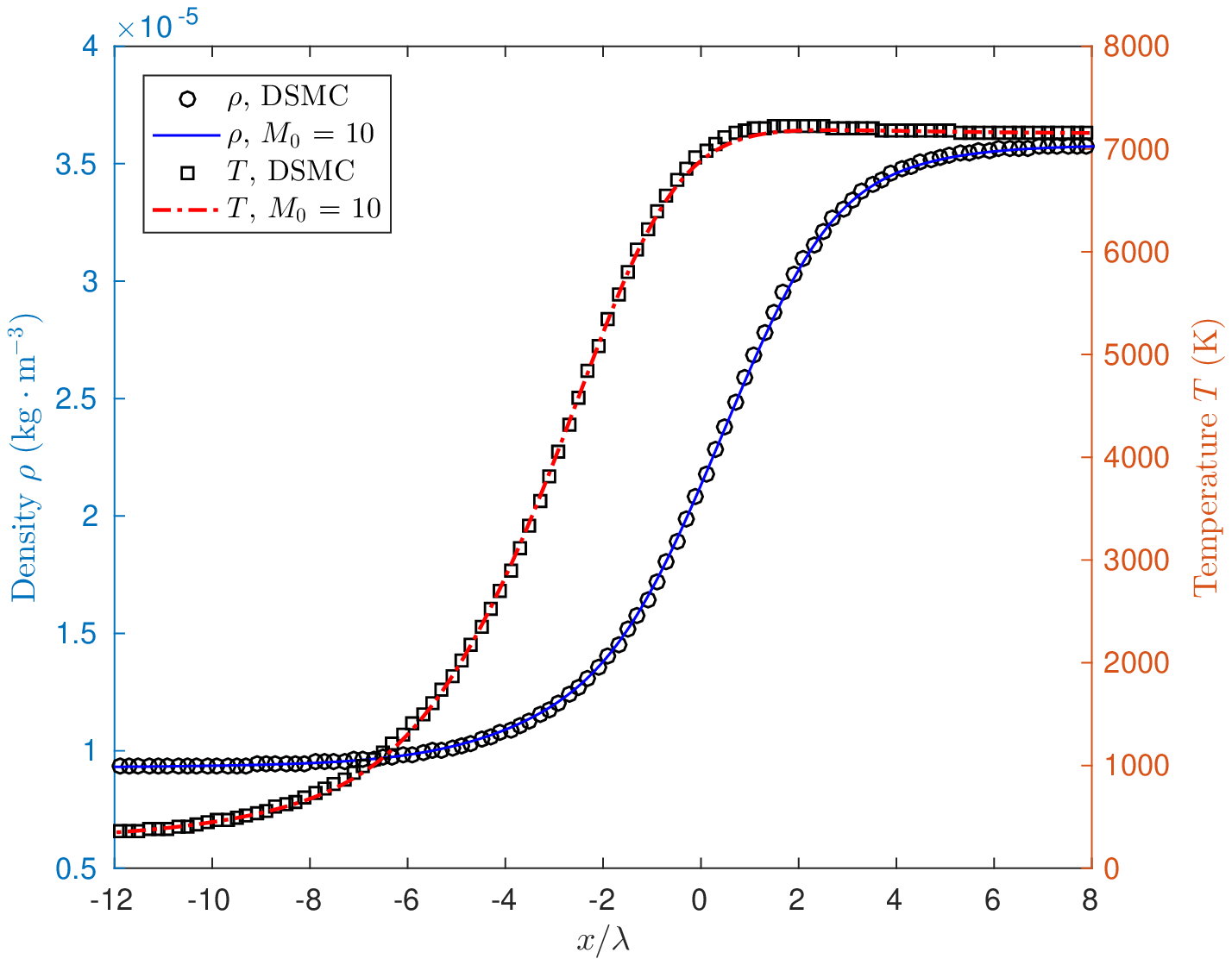}}\hfill
  {\includegraphics[width=0.49\textwidth,clip]{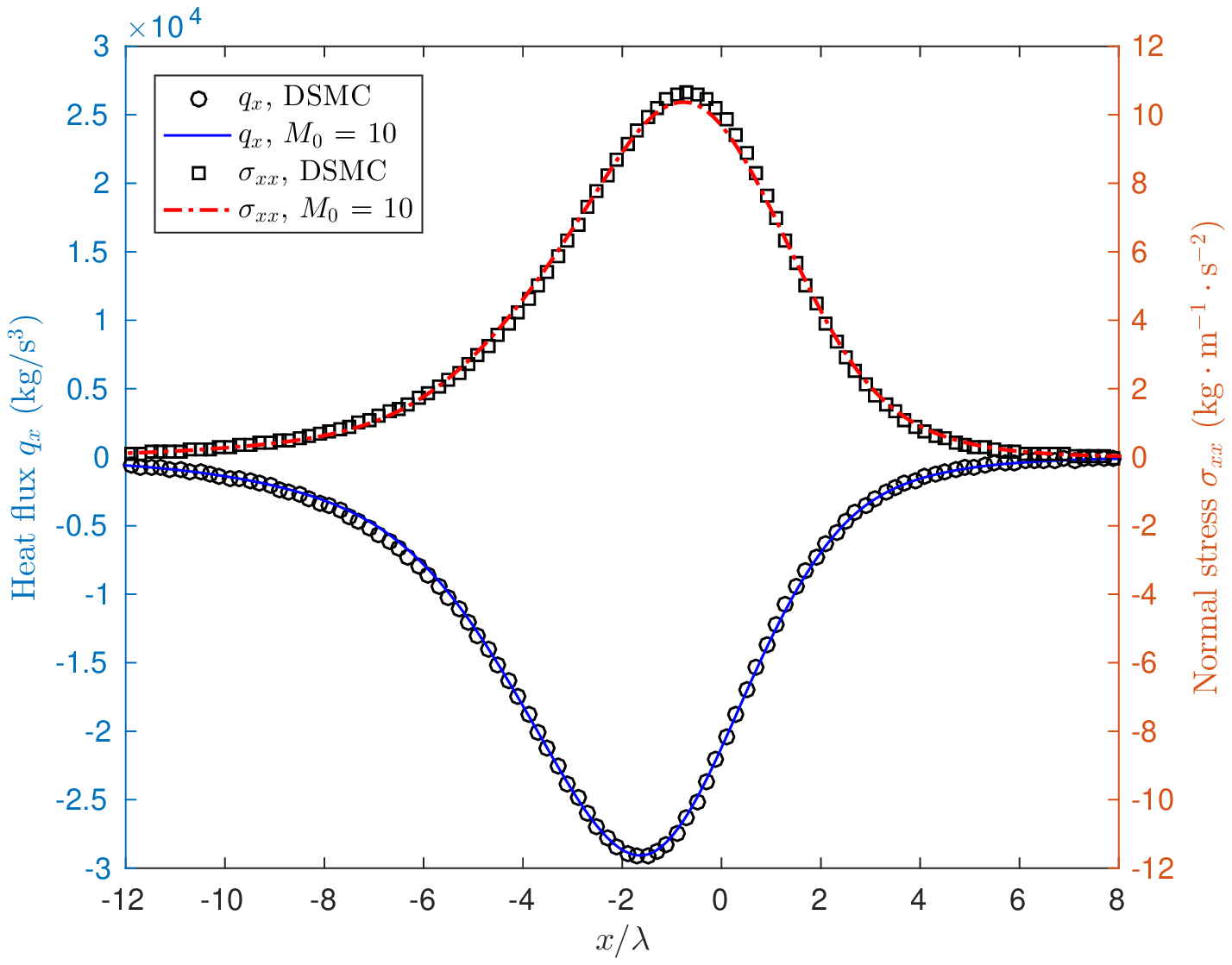}}
  \caption{Solution of the shock structure for $\mathit{Ma} = 9.0$ and $M=20$.}
  \label{fig:shock-Ma90-binaryOM20}
\end{figure}

\begin{table}
\centering
\caption{Run-time data for shock structure simulations with $60$ grid cells, $M_{0}=10$ and $M=20$.}
\label{tab:shock-cputime}
\begin{tabular}{ccc}
\hline
Test case & $\mathit{Ma} = 6.5$ & $\mathit{Ma} = 9.0$ \\
\hline
  Number of coefficients & $1771$ & $1771$ \\ 
  Time step ($\Delta t$) & $\SI{6.47e-7}{\second}$ & $\SI{4.74e-7}{\second}$ \\
  Number of time steps & $2605$ & $2690$ \\ 
  Total elapsed time & $\SI{96.31}{\second}$ & $\SI{99.99}{\second}$ \\ 
  Elapsed time per time step & $\SI{3.70e-2}{\second}$ & $\SI{3.72e-2}{\second}$ \\
\hline
\end{tabular}
\end{table}

\subsubsection{Fourier flow}
This is another benchmark test for problems with boundary conditions
\cite{Hu2019}. The fluid locates between two stationary and infinitely
large parallel plates with different temperature. At the steady state,
significant temperature jump can be observed for rarefied gases. The
parameters of this problem include
\begin{itemize}
\item $L$: distance between two plates;
\item $T_a^{W}$, $T_b^{W}$: the temperature of the left and right plates;
\item $\rho_0$: the average density of the fluid.
\end{itemize}
In our tests, we always choose $T_a^{W} = \SI{273.15}{\kelvin}$ and
$\rho_0 = \SI{9.282e-6}{\kg\per\meter\cubed}$. The computational
domain is defined by $a = -L/2$ and $b = L/2$ with the accommodation
coefficients in the boundary condition being
$\chi_{a} = \chi_{b} = 1$. The domain is decomposed into $10$ uniform
grid cells, and the fourth-order nodal discontinuous Galerkin method
\cite{Hesthaven} is used for spatial discretization. For velocity space,
we discretize it using $\lu = 0$ and
$\lth = (\theta_{a}^{W} + \theta_{b}^{W} ) / 2$, where
$\theta_a ^{W}= k_B T_a^{W} / \mm$ and
$\theta_b^{W}= k_B T_b^{W} / \mm$. We compute the steady state by
starting from the initial condition
\begin{displaymath}
f(0,x,\bv) = \frac{\rho_0}{\mm (2\pi \lth)^{3/2}} \exp
  \left( -\frac{|\bv|^2}{2\lth} \right),
\end{displaymath}
and the stopping criterion is again \eqref{eq:stop_crit} with
$\epsilon = 10^{-6}$. In order to compare our results, the DSMC method
\cite{Bird1994} is also employed to produce the reference
solution. Below we are going to consider two different choices of
$T_b^{W}$.

(1) $T_b^{W} = 4T_a^{W}$. We first set the temperature ratio of two
plates to be $4$.
Three distances $L = \SI{0.092456}{\meter}$, $\SI{0.018491}{\meter}$
and $\SI{0.003698}{\meter}$ are considered. They correspond to Knudsen
number $0.1$, $0.5$ and $2.5$, respectively. For numerical results
presented in this paper, we adopt $M=20$, $30$ and $40$, as the
Knudsen number increases from $0.1$ to $2.5$.

Numerical solutions of density $\rho$, temperature $T$, normal stress
$\sigma_{xx}$ and heat flux $q_{x}$, obtained by our method with three
choices of $M_{0}$, i.e., $M_{0}=5$, $10$ and $15$, are presented in
\figurename~\ref{fig:fourier-Twr4Kn01-binaryOM20}-\ref{fig:fourier-Twr4Kn25-binaryOM40}
for $\mathit{Kn} = 0.1$, $0.5$ and $2.5$ respectively. It can be seen
that all our results agree well with the DSMC solution. The relative
deviation of our solution away from the DSMC solution is actually
quite small. Moreover, our solution becomes much closer and closer to
the DSMC solution, as $M_{0}$ increases.

(2) $T_b^{W} = 10T_a^{W}$. Now we set the temperature ratio to be
$10$, which is obviously tougher to simulate due to the wide spread of
the distribution functions.
Two distances $L = \SI{0.018491}{\meter}$ and $\SI{0.003698}{\meter}$,
with the corresponding Knudsen number $0.5$ and $2.5$ respectively,
are considered. As the previous case, we set $M=30$ for
$\mathit{Kn}=0.5$ and $M=40$ for $\mathit{Kn}=2.5$ in our tests.

Numerical solutions obtained by our method with $M_{0}=5$, $10$ and
$15$, and the DSMC method, are given in
\figurename~\ref{fig:fourier-Twr10Kn05-binaryOM30} and
\ref{fig:fourier-Twr10Kn25-binaryOM40} for $\mathit{Kn}=0.5$ and
$\mathit{Kn}=2.5$, respectively. The results still show a good
agreement between our solutions and the DSMC solutions, although more
obvious deviation can be observed, especially for the case with
$\mathit{Kn}=2.5$. However, as more moments are modelled accurately by
the quadratic collision model in our method, which indicates $M_{0}$
is increased, remarkable improvement of our results can be obtained as
shown in these figures.

At last, the run-time data of partial simulations, which are obtained
on the same cluster with 10 threads for each simulation as shock
structure simulations, are also provided in Table
\ref{tab:fourier-cputime} to show the efficiency of our method.

\begin{figure}[!htb]
  \centering
  \subfloat[Density, $\rho~({\rm kg\cdot m^{-3}})$]{\includegraphics[width=0.49\textwidth,clip]{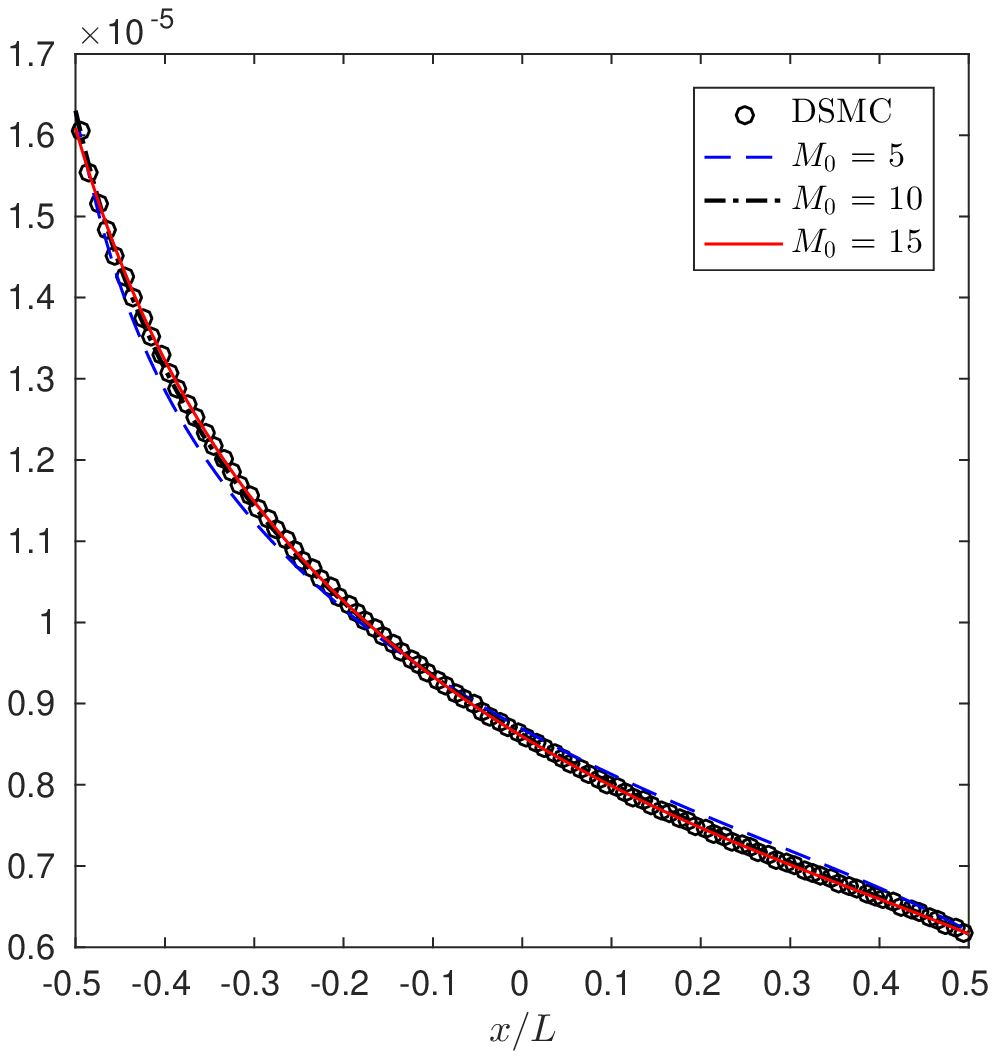}}\hfill
  \subfloat[Temperature, $T~({\rm K})$]{\includegraphics[width=0.49\textwidth,clip]{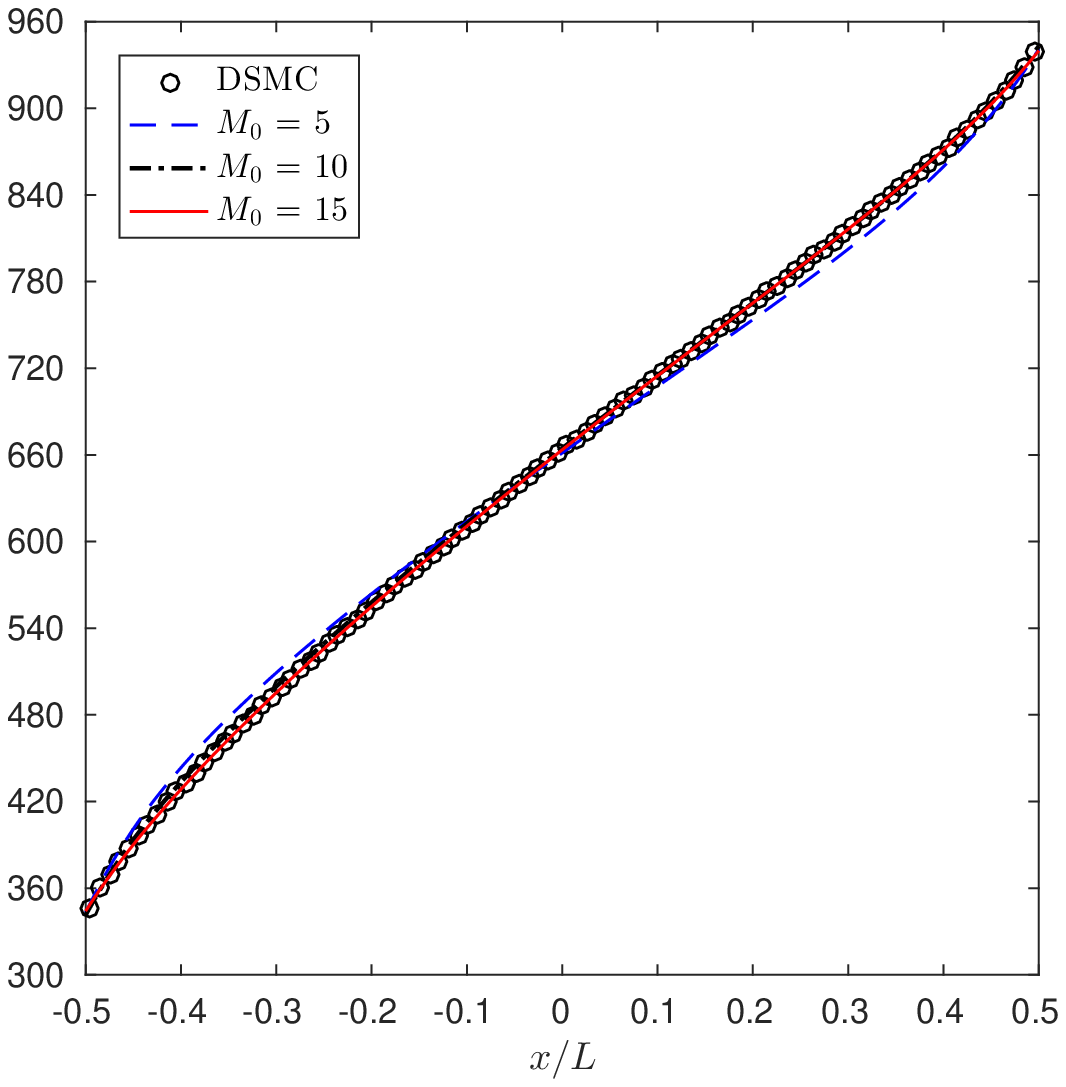}} \\
  \subfloat[Normal stress, $\sigma_{xx}~({\rm kg \cdot m^{-1}\cdot s^{-2}})$]{\includegraphics[width=0.49\textwidth,clip]{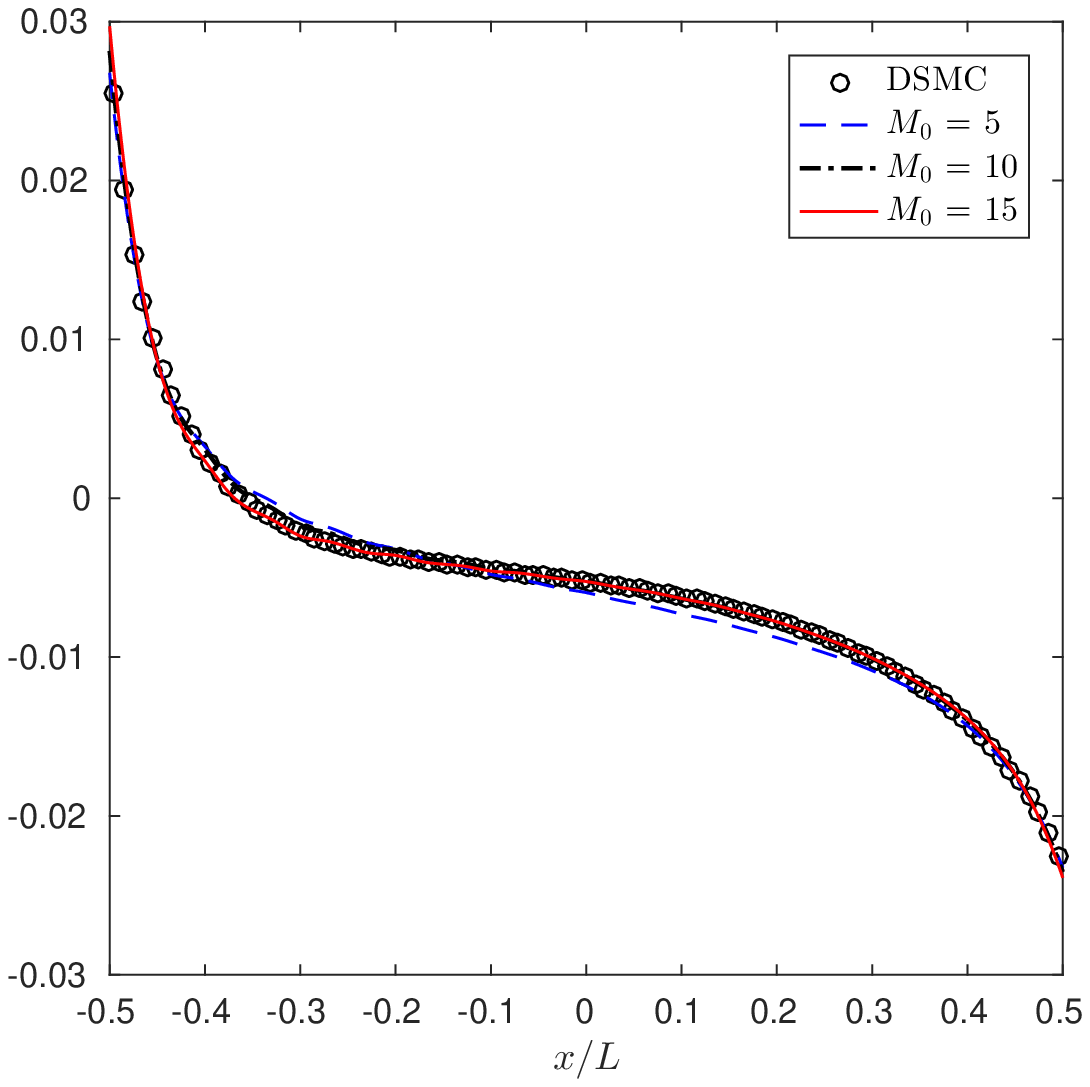}} \hfill
  \subfloat[Heat flux, $q_x~(\rm kg/s^{3})$]{\includegraphics[width=0.49\textwidth,clip]{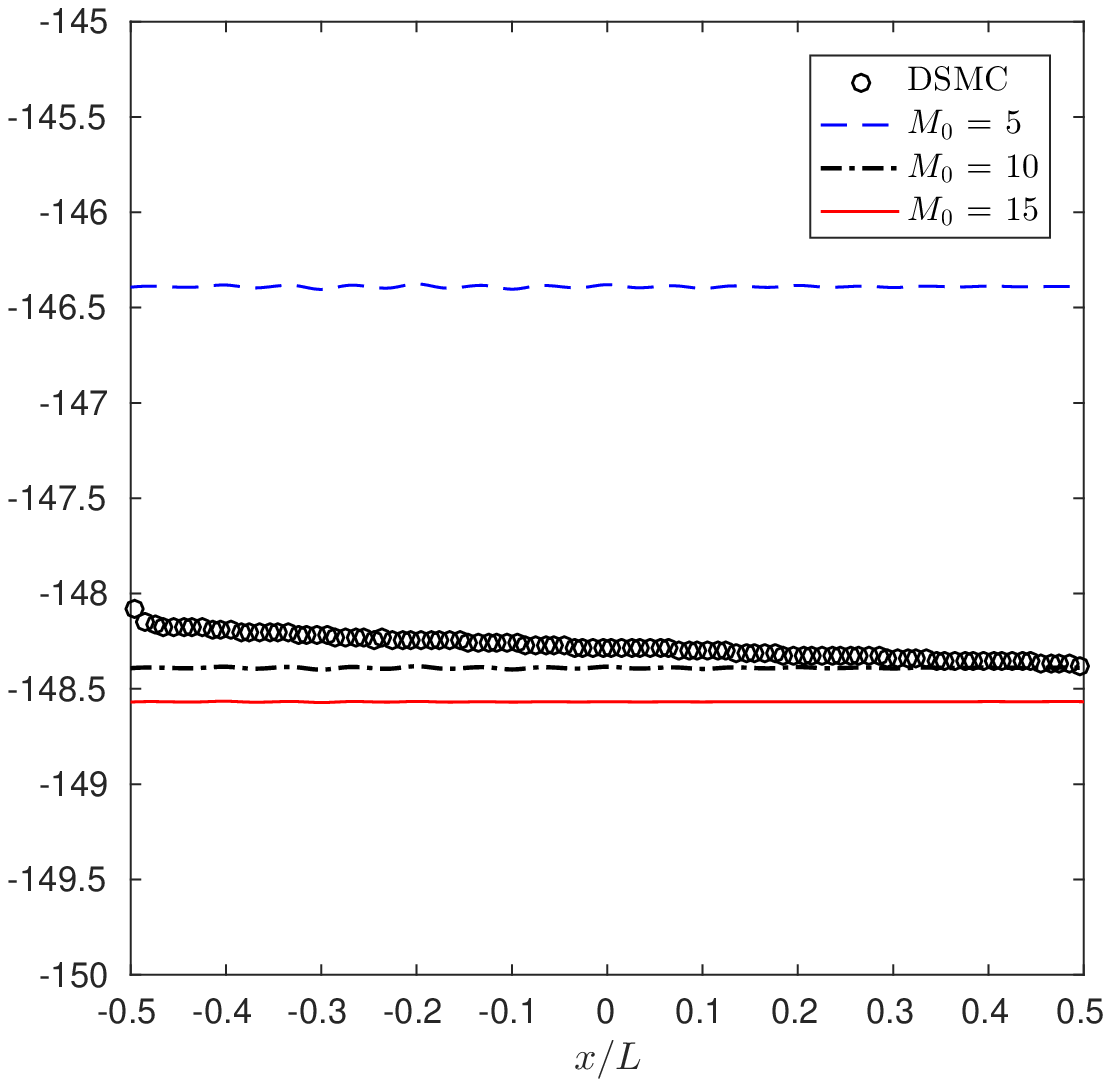}} 
  \caption{Solution of the Fourier flow for $T_{b}^{W}=4T_{a}^{W}$ with $\Kn =0.1$ and $M=20$.}
  \label{fig:fourier-Twr4Kn01-binaryOM20}
\end{figure}

\begin{figure}[!htb]
  \centering
  \subfloat[Density, $\rho~({\rm kg\cdot m^{-3}})$]{\includegraphics[width=0.49\textwidth,clip]{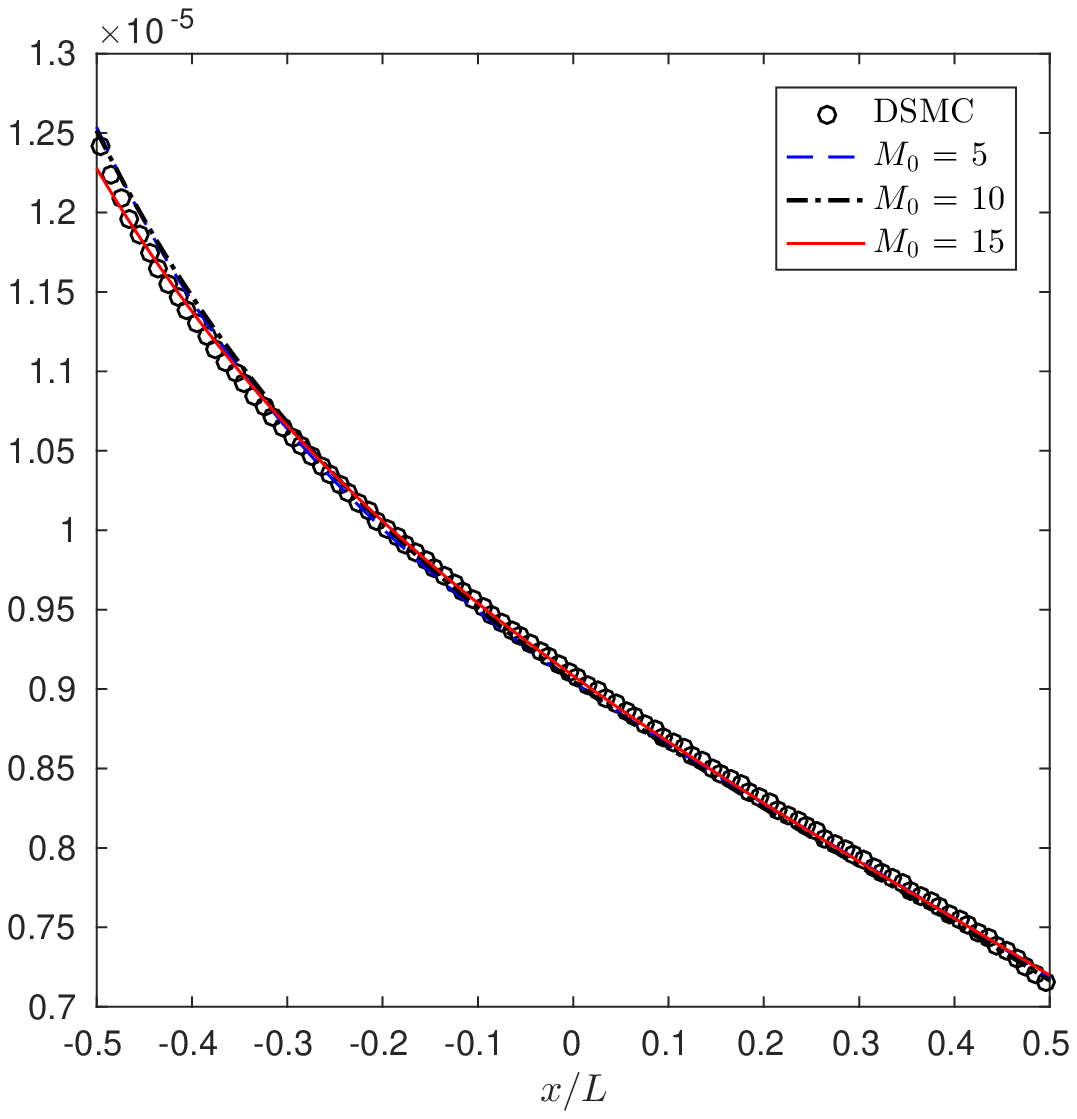}}\hfill
  \subfloat[Temperature, $T~({\rm K})$]{\includegraphics[width=0.49\textwidth,clip]{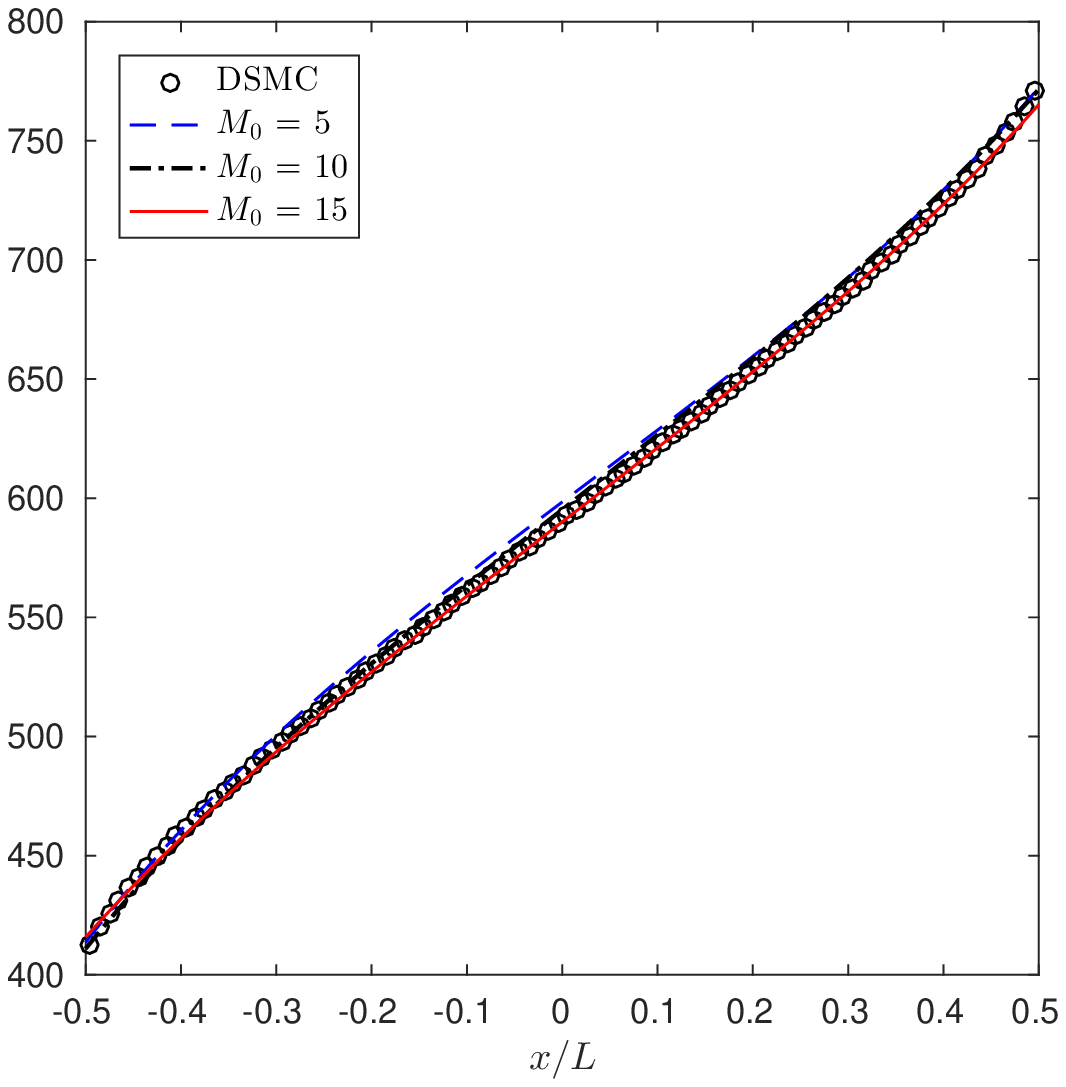}} \\
  \subfloat[Normal stress, $\sigma_{xx}~({\rm kg \cdot m^{-1}\cdot s^{-2}})$]{\includegraphics[width=0.49\textwidth,clip]{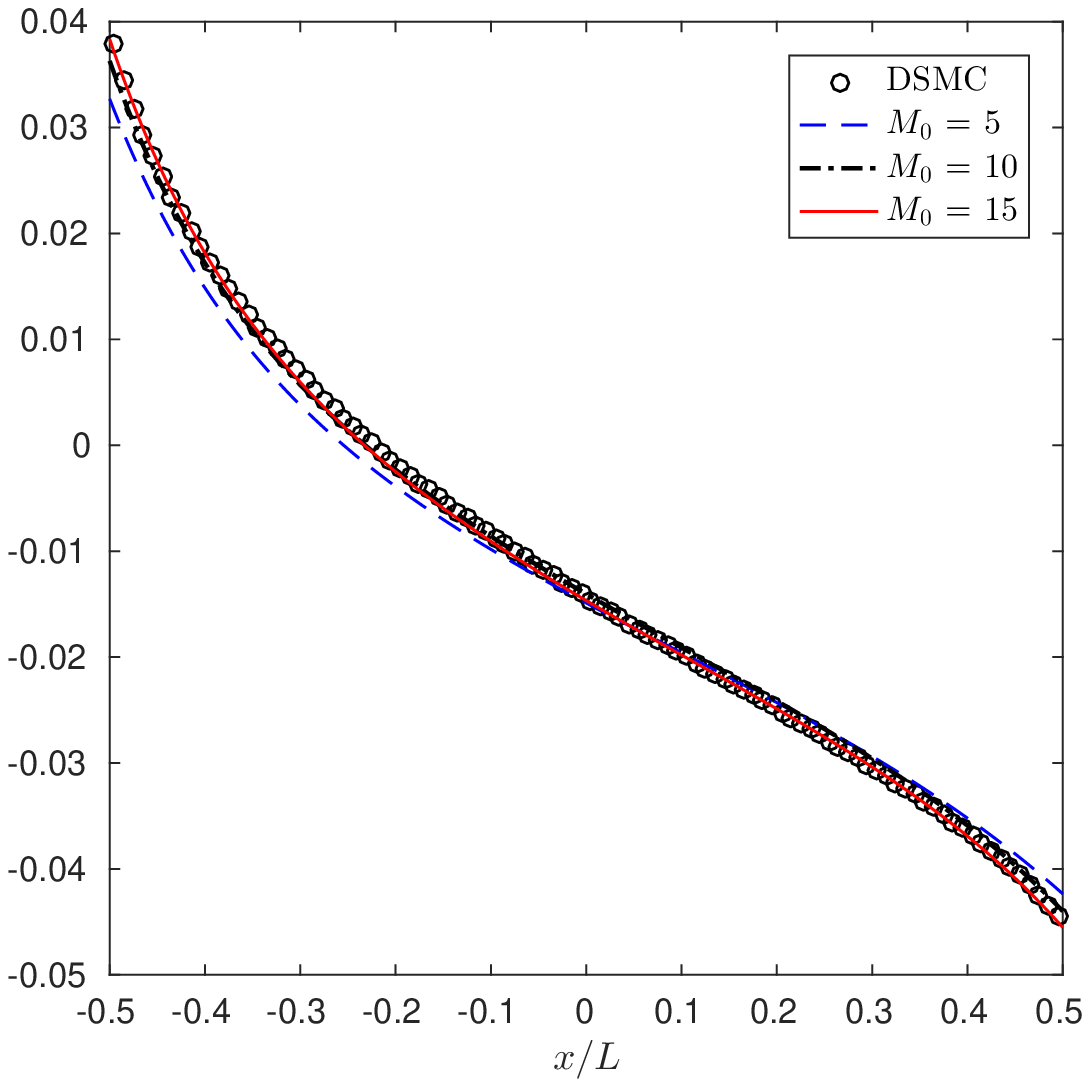}} \hfill
  \subfloat[Heat flux, $q_x~(\rm kg/s^{3})$]{\includegraphics[width=0.49\textwidth,clip]{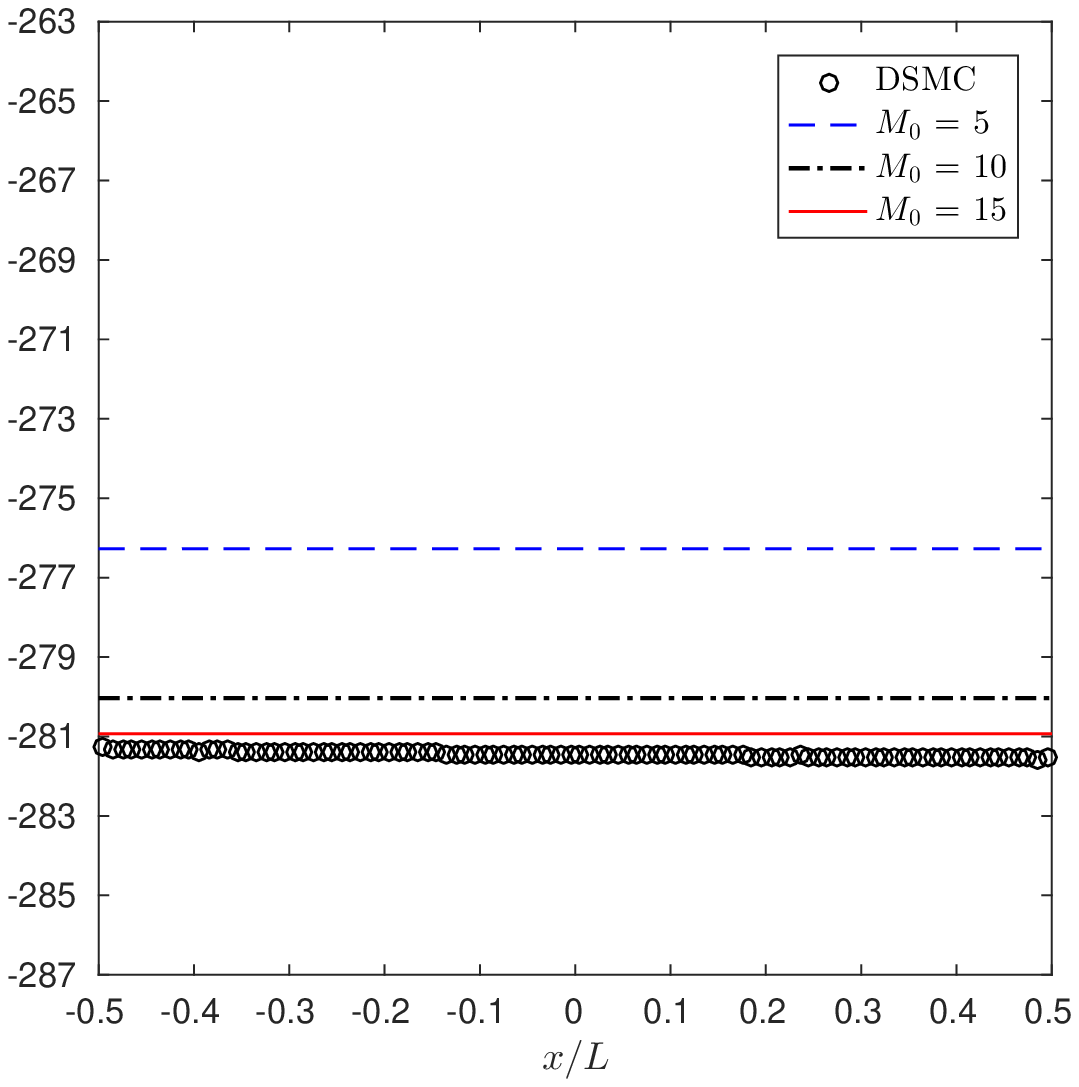}} 
  \caption{Solution of the Fourier flow for $T_{b}^{W}=4T_{a}^{W}$ with $\Kn =0.5$ and $M=30$.}
  \label{fig:fourier-Twr4Kn05-binaryOM30}
\end{figure}

\begin{figure}[!htb]
  \centering
  \subfloat[Density, $\rho~({\rm kg\cdot m^{-3}})$]{\includegraphics[width=0.49\textwidth,clip]{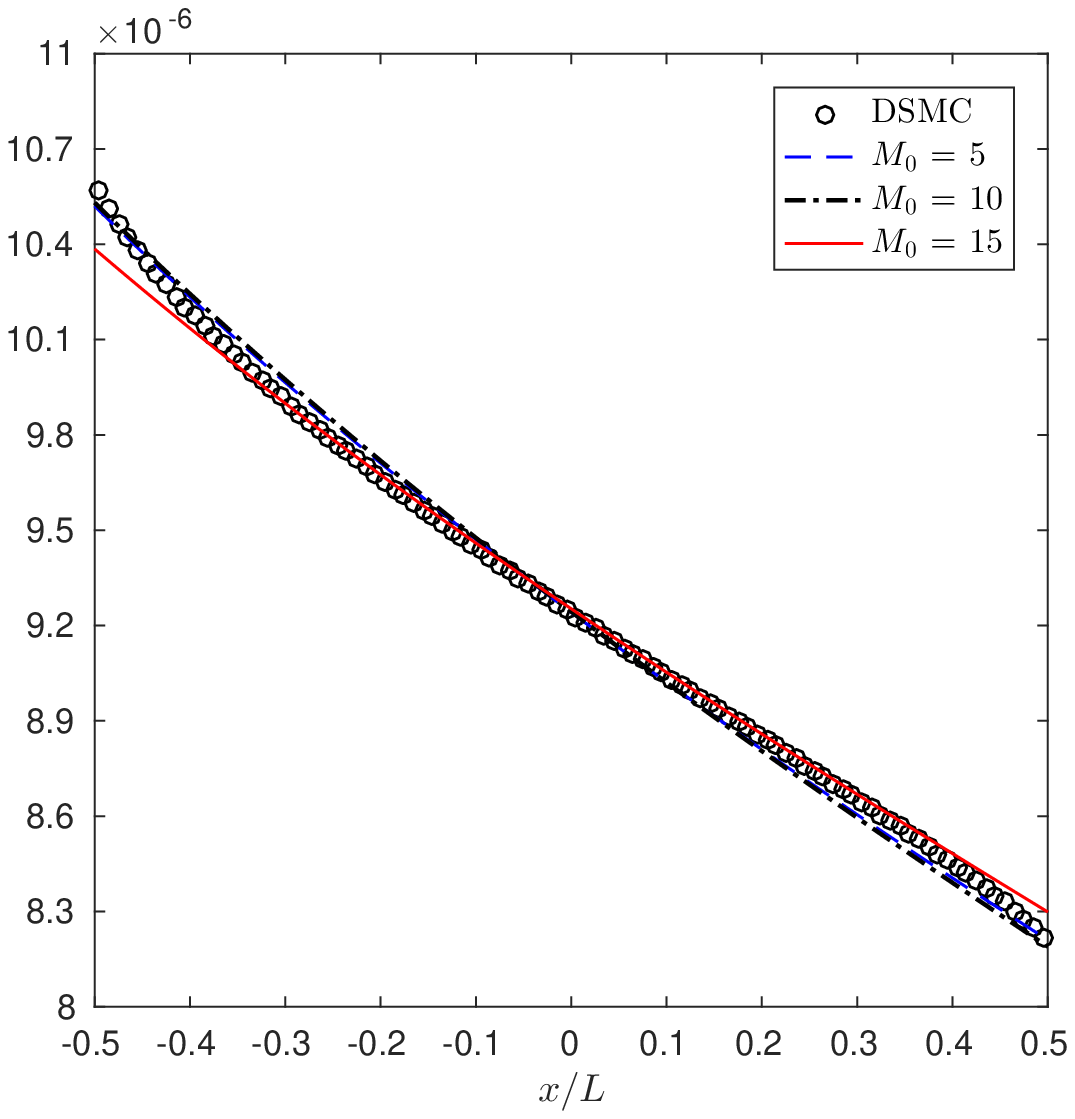}}\hfill
  \subfloat[Temperature, $T~({\rm K})$]{\includegraphics[width=0.49\textwidth,clip]{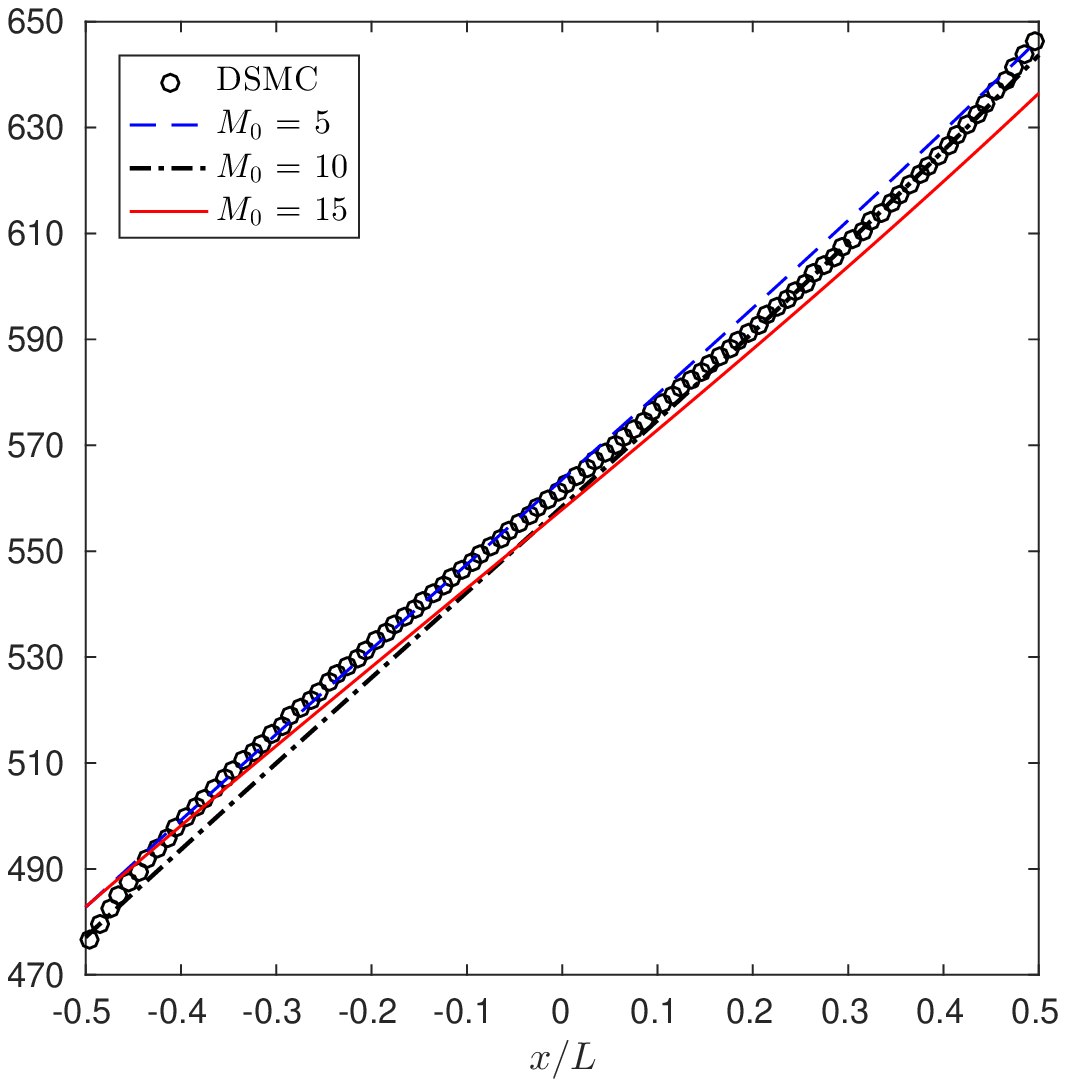}} \\
  \subfloat[Normal stress, $\sigma_{xx}~({\rm kg \cdot m^{-1}\cdot s^{-2}})$]{\includegraphics[width=0.49\textwidth,clip]{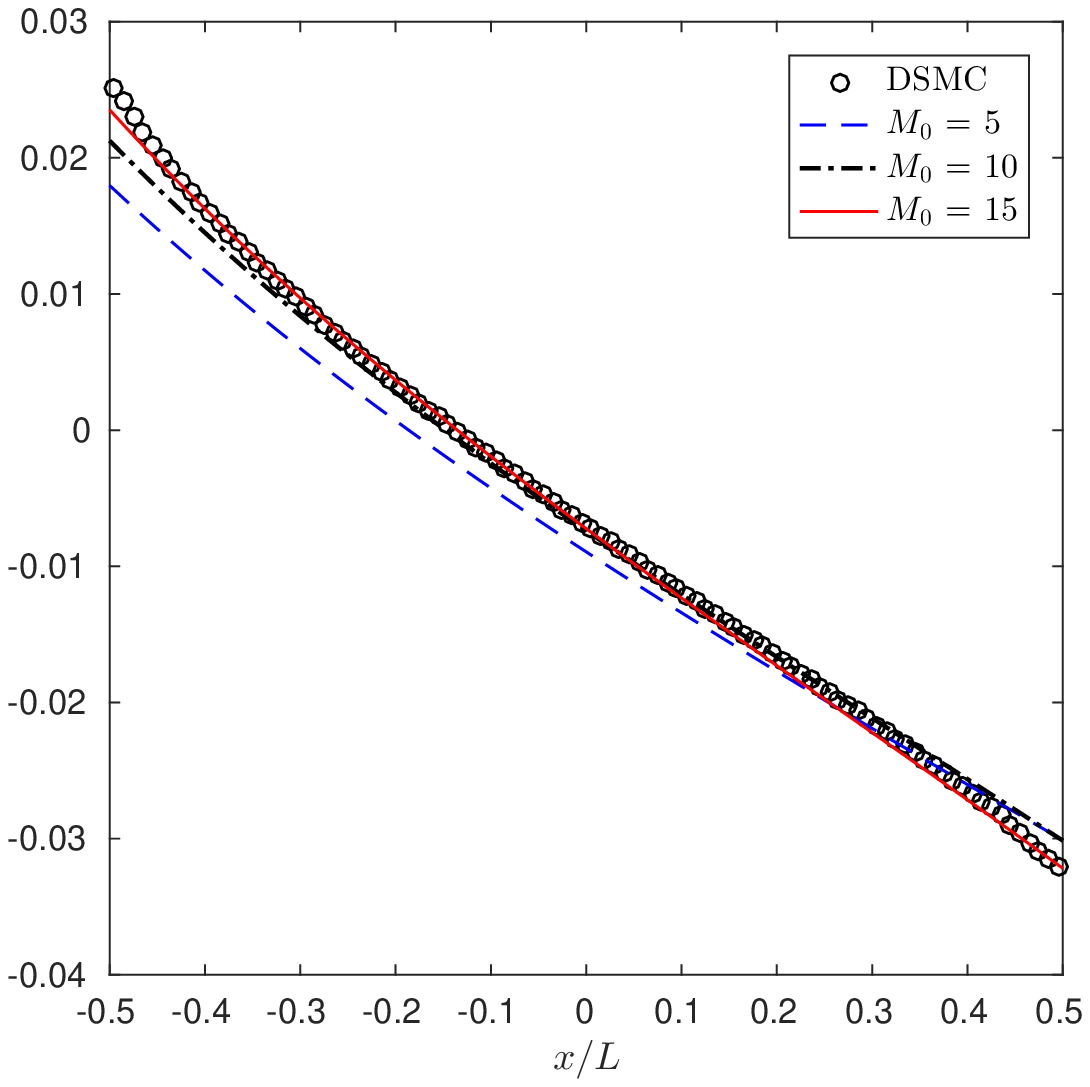}} \hfill
  \subfloat[Heat flux, $q_x~(\rm kg/s^{3})$]{\includegraphics[width=0.49\textwidth,clip]{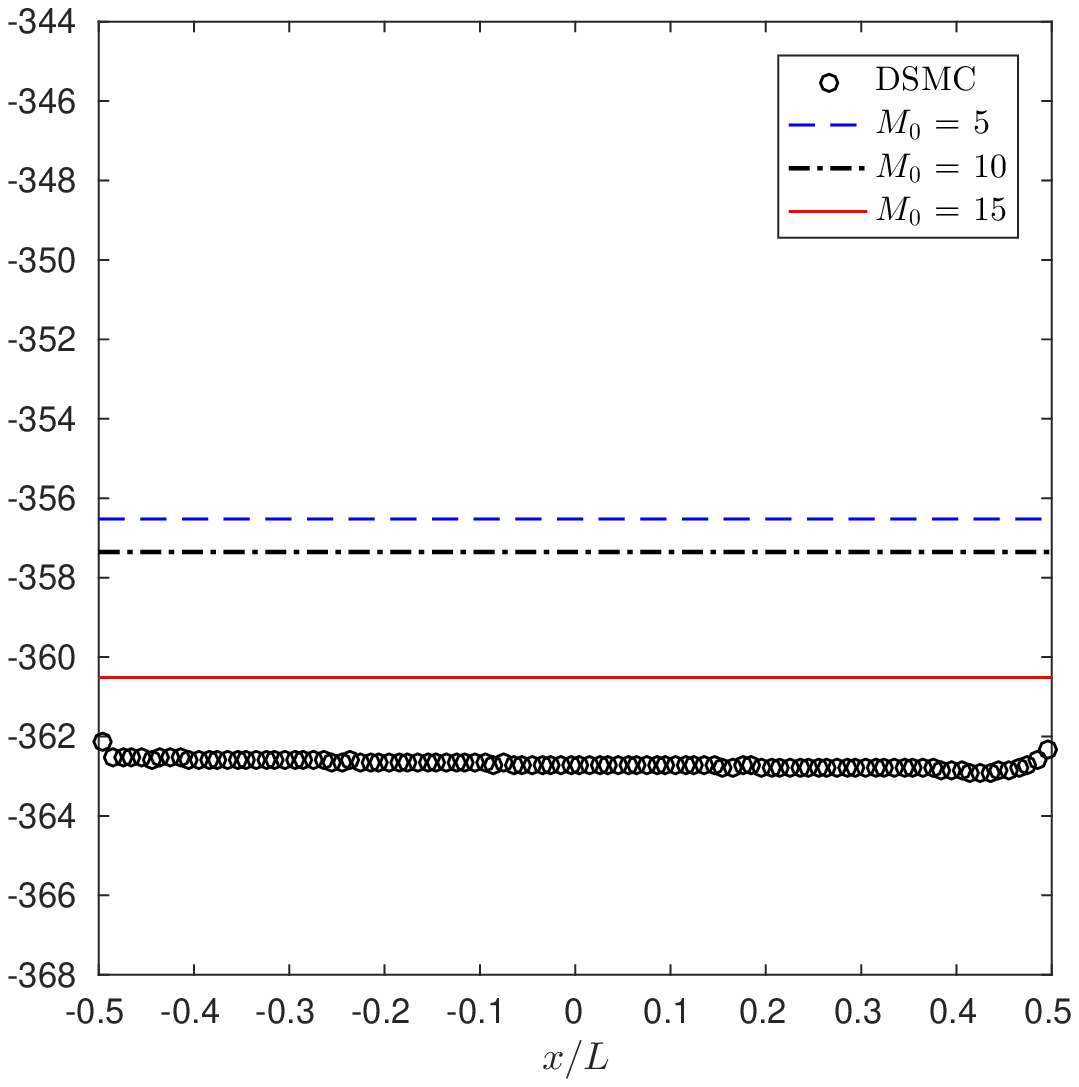}} 
  \caption{Solution of the Fourier flow for $T_{b}^{W}=4T_{a}^{W}$ with $\Kn =2.5$ and $M=40$.}
  \label{fig:fourier-Twr4Kn25-binaryOM40}
\end{figure}

\begin{figure}[!htb]
  \centering
  \subfloat[Density, $\rho~({\rm kg\cdot m^{-3}})$]{\includegraphics[width=0.49\textwidth,clip]{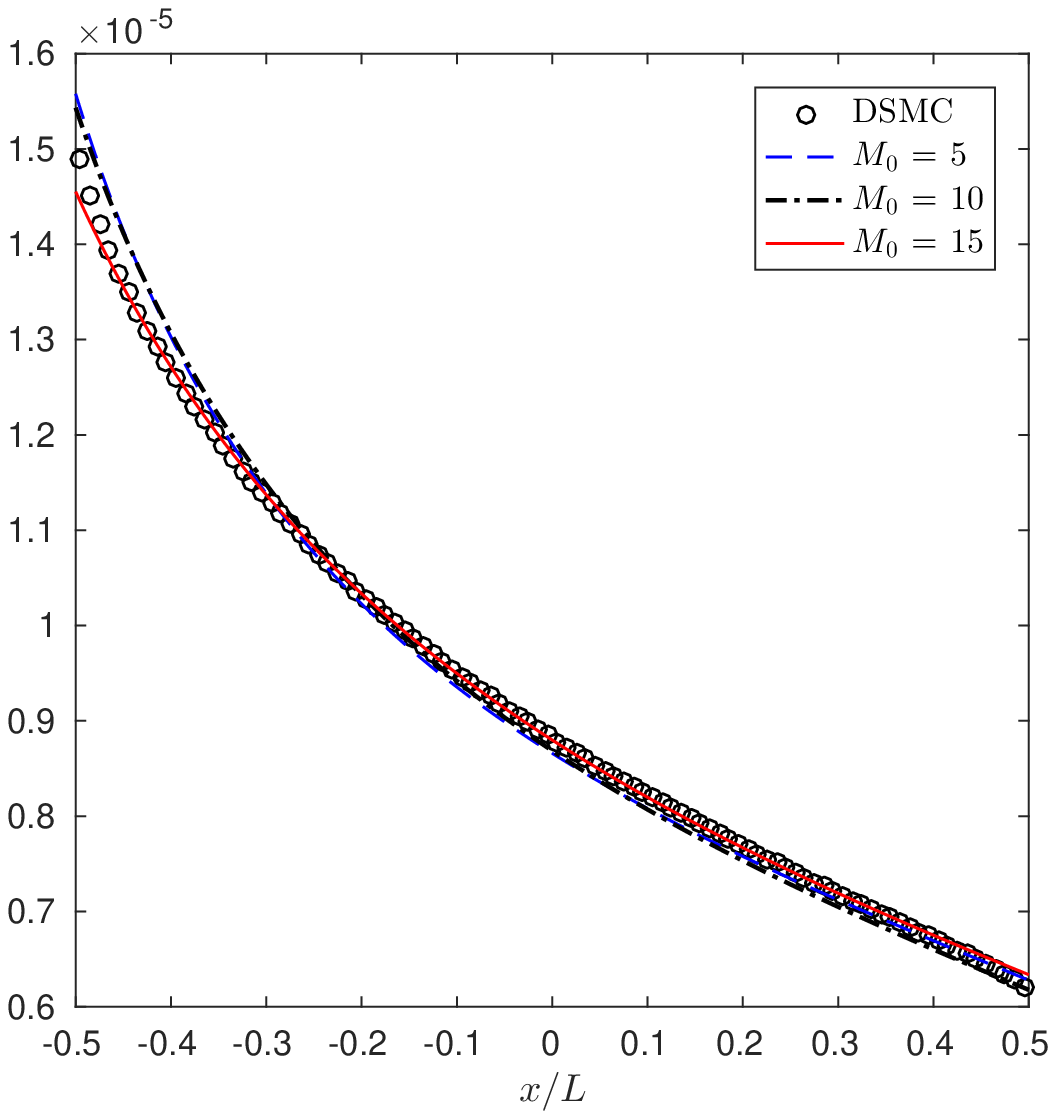}}\hfill
  \subfloat[Temperature, $T~({\rm K})$]{\includegraphics[width=0.49\textwidth,clip]{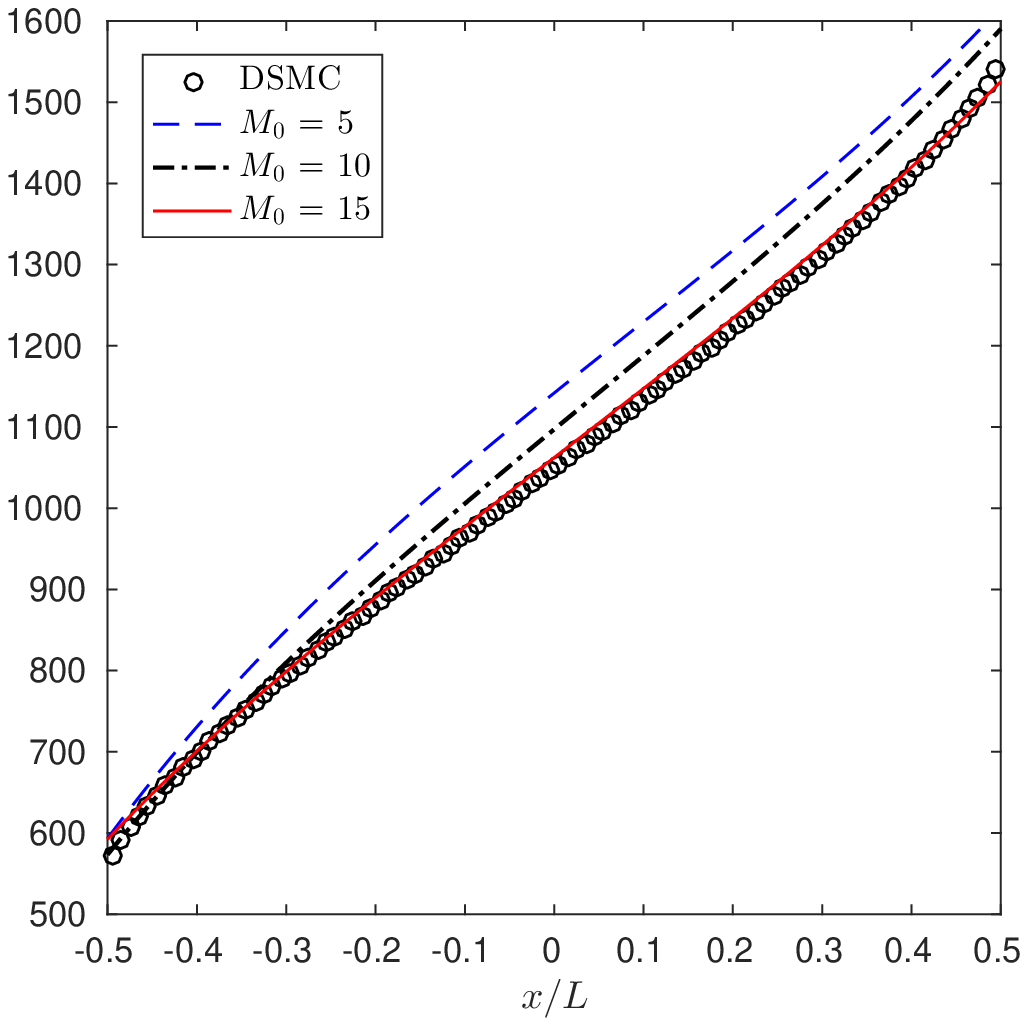}} \\
  \subfloat[Normal stress, $\sigma_{xx}~({\rm kg \cdot m^{-1}\cdot s^{-2}})$]{\includegraphics[width=0.49\textwidth,clip]{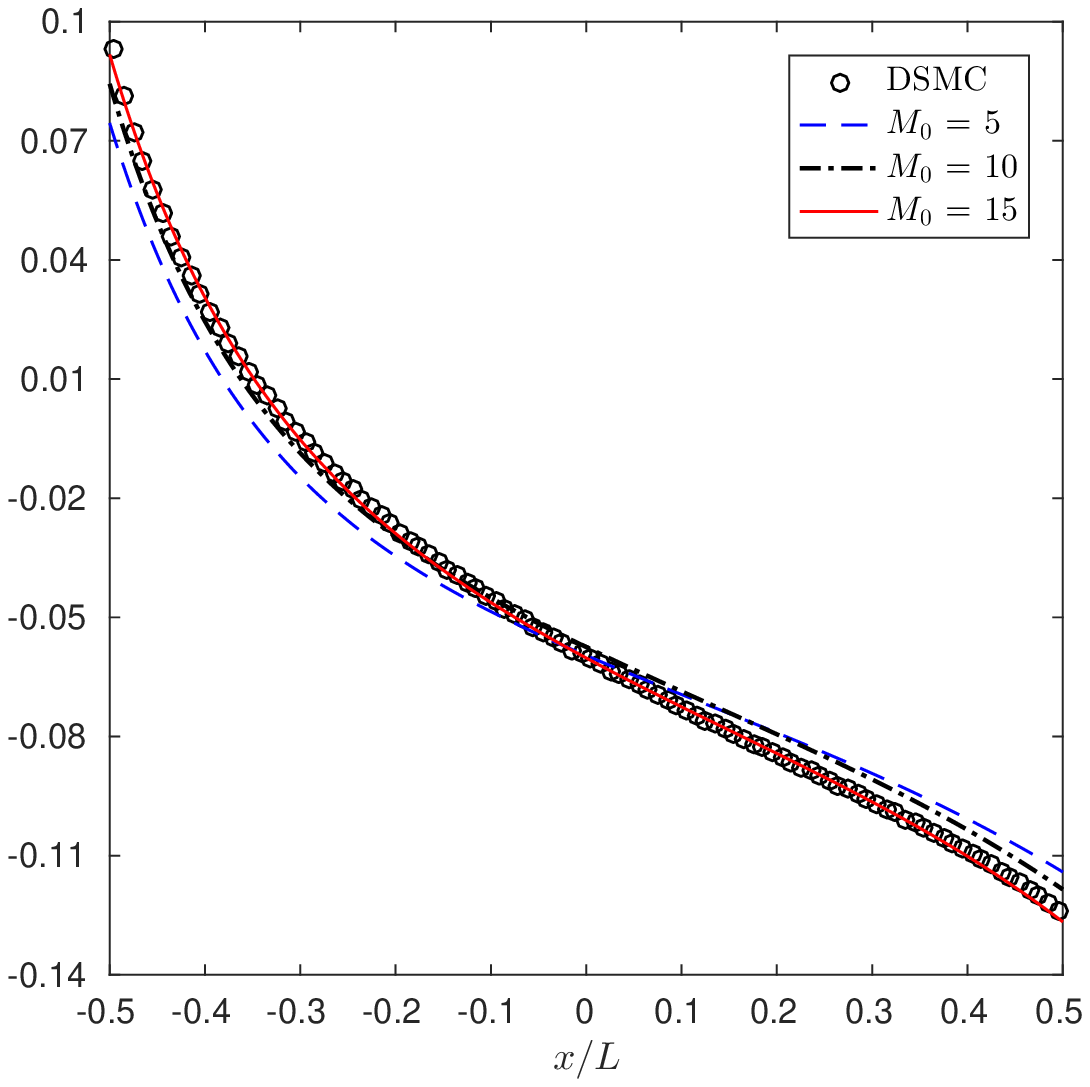}} \hfill
  \subfloat[Heat flux, $q_x~(\rm kg/s^{3})$]{\includegraphics[width=0.49\textwidth,clip]{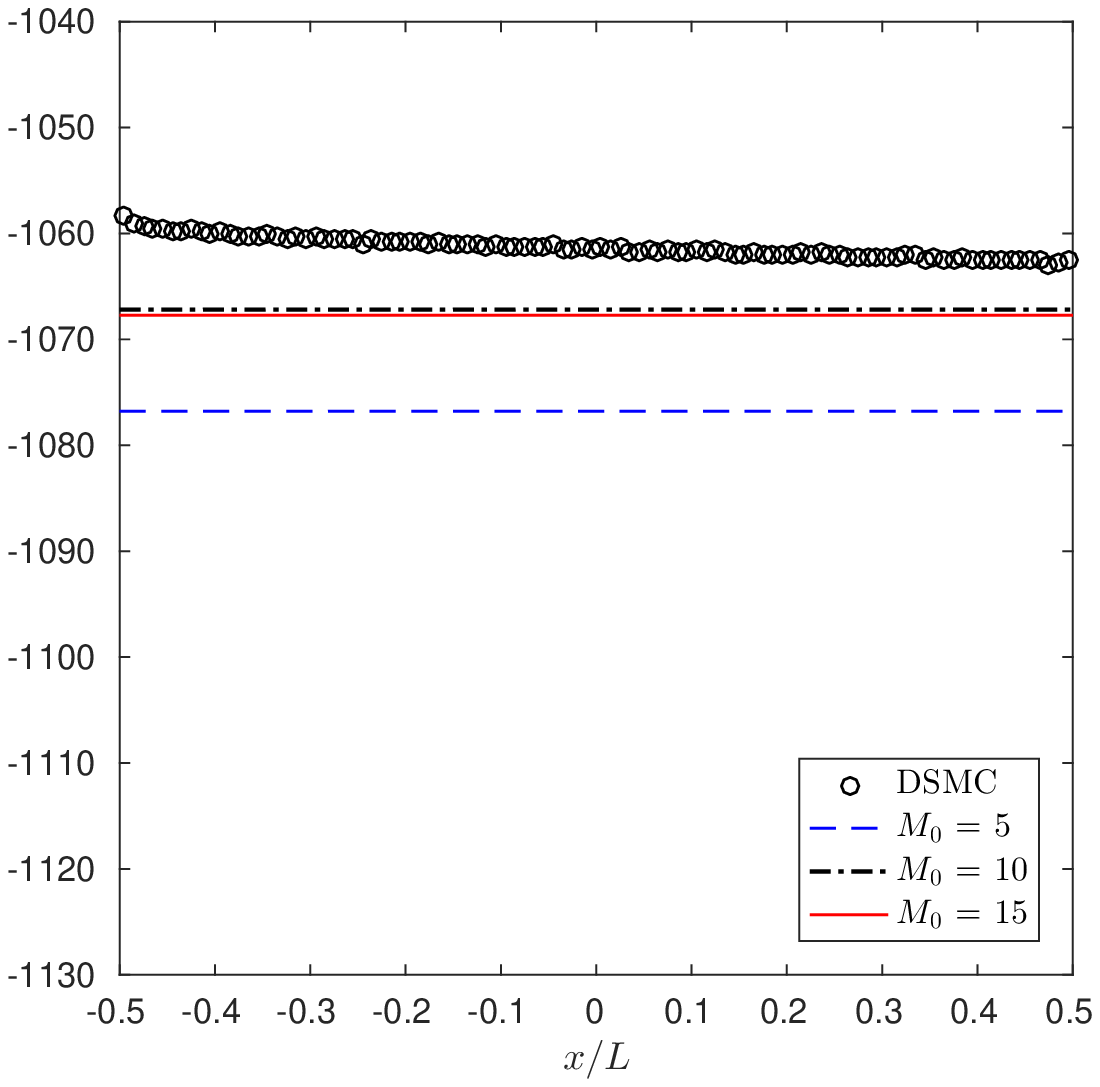}} 
  \caption{Solution of the Fourier flow for $T_{b}^{W}=10T_{a}^{W}$ with $\Kn =0.5$ and $M=30$.}
  \label{fig:fourier-Twr10Kn05-binaryOM30}
\end{figure}

\begin{figure}[!htb]
  \centering
  \subfloat[Density, $\rho~({\rm kg\cdot m^{-3}})$]{\includegraphics[width=0.49\textwidth,clip]{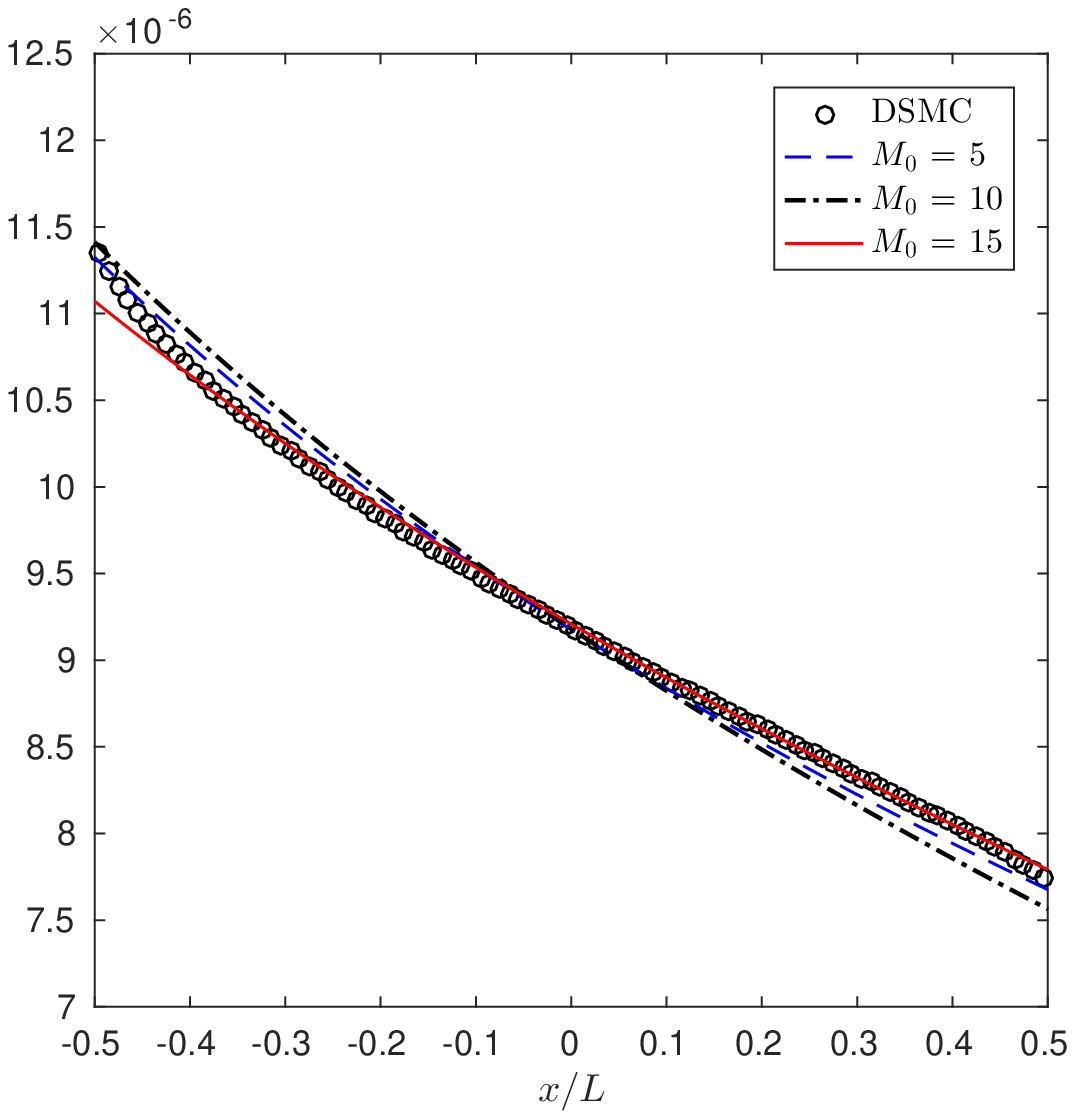}}\hfill
  \subfloat[Temperature, $T~({\rm K})$]{\includegraphics[width=0.49\textwidth,clip]{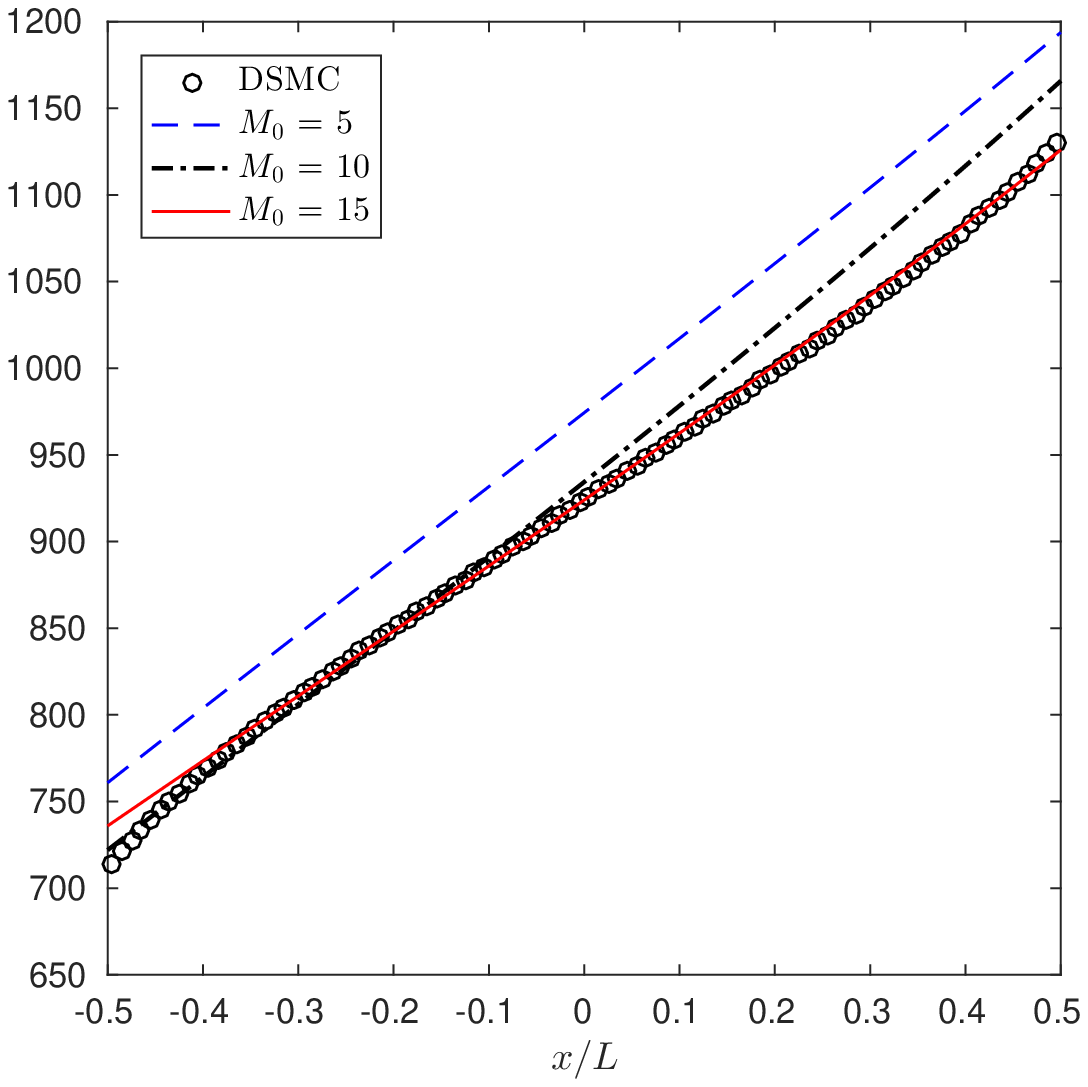}} \\
  \subfloat[Normal stress, $\sigma_{xx}~({\rm kg \cdot m^{-1}\cdot s^{-2}})$]{\includegraphics[width=0.49\textwidth,clip]{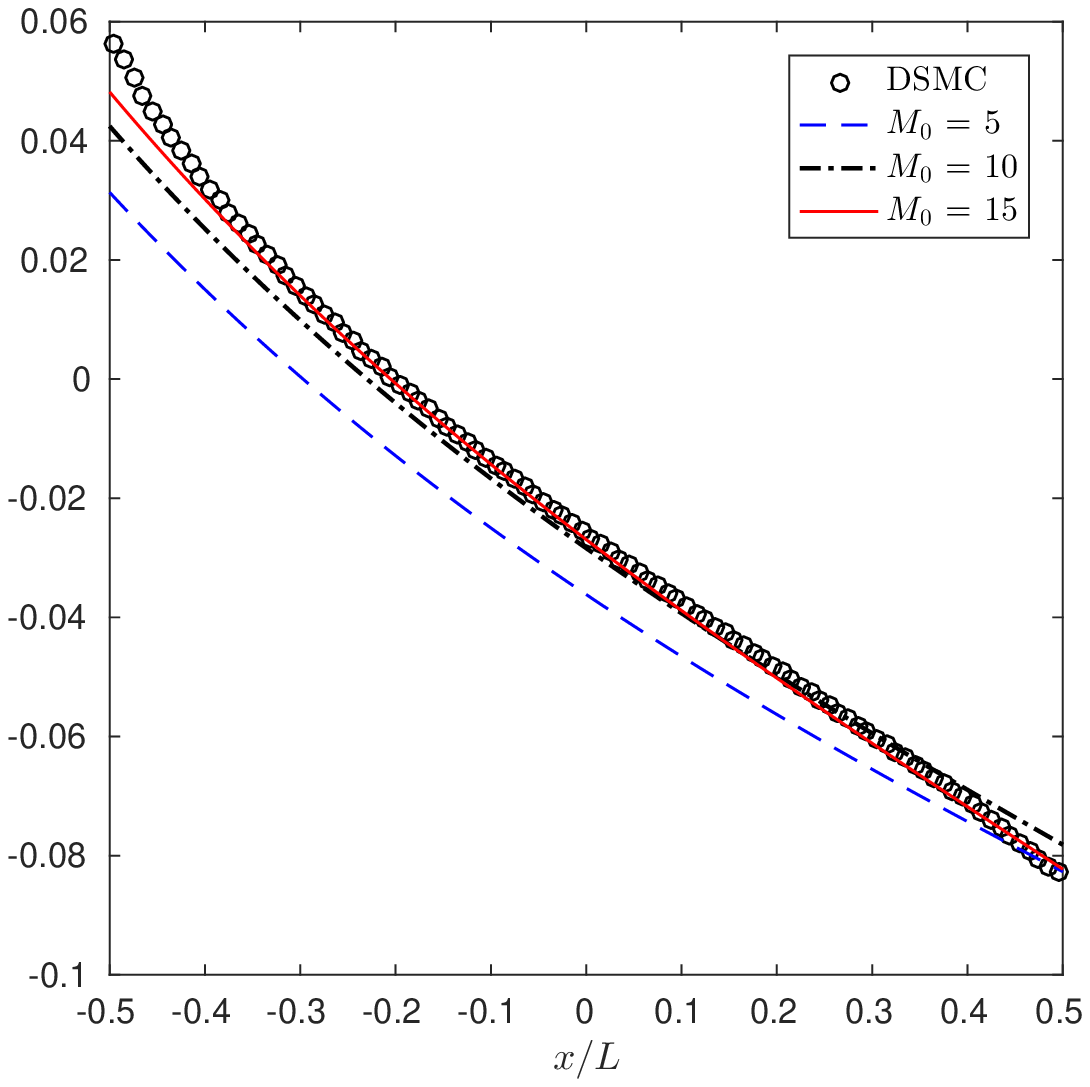}} \hfill
  \subfloat[Heat flux, $q_x~(\rm kg/s^{3})$]{\includegraphics[width=0.49\textwidth,clip]{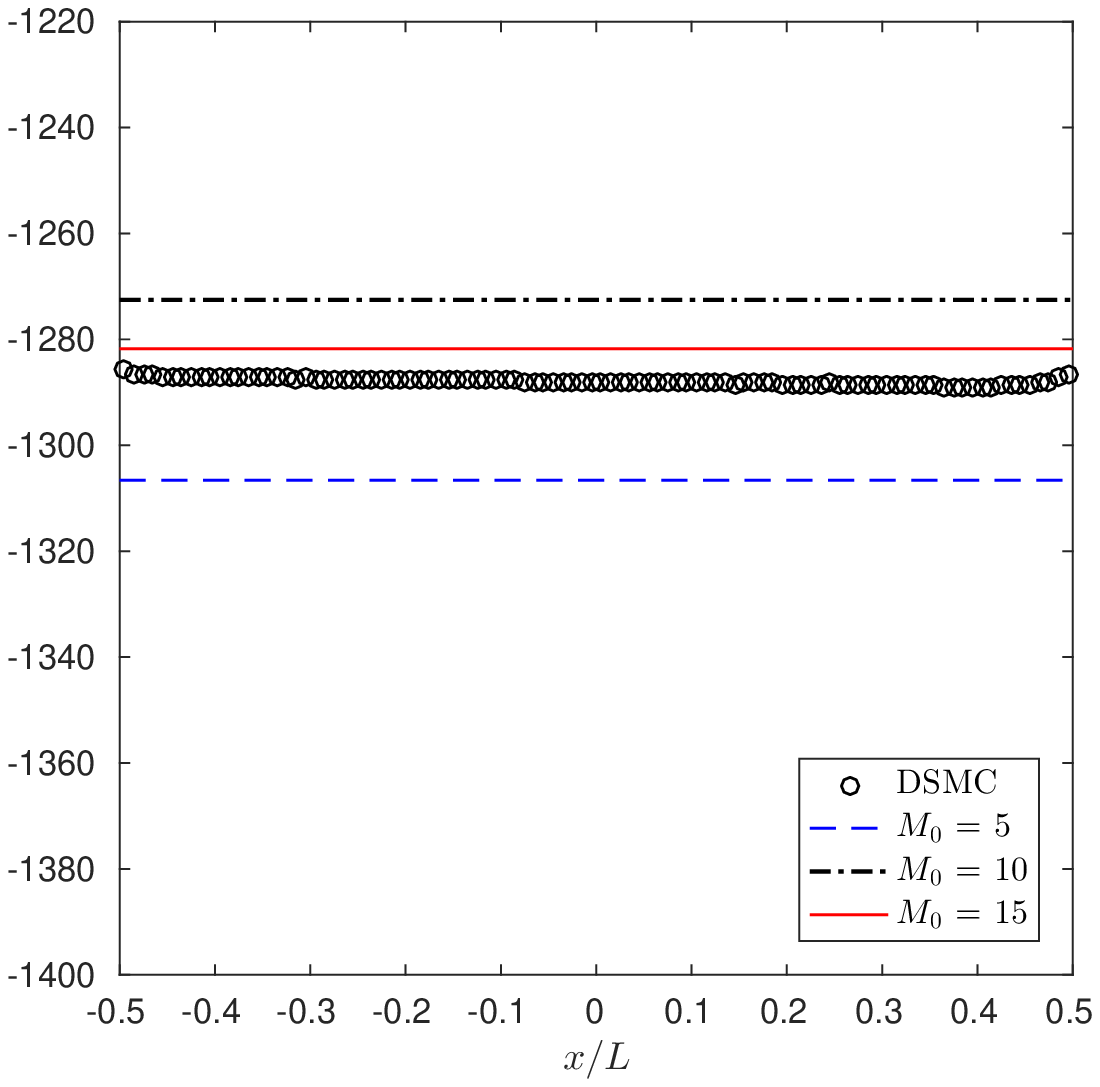}} 
  \caption{Solution of the Fourier flow for $T_{b}^{W}=10T_{a}^{W}$ with $\Kn =2.5$ and $M=40$.}
  \label{fig:fourier-Twr10Kn25-binaryOM40}
\end{figure}

\begin{table}
\centering
\caption{Run-time data for Fourier simulations with $10$ grid cells.}
\label{tab:fourier-cputime}
\begin{tabular}{c|ccc|cc}
\hline
Test case & \multicolumn{3}{c|}{$T_{b}^{W}=4T_{a}^{W}$} & \multicolumn{2}{c}{$T_{b}^{W}=10T_{a}^{W}$} \\
\hline
  Knudsen number ($\Kn$) & $0.1$ & $0.5$ & $2.5$ & $0.5$ & $2.5$ \\
  $M_{0}$ and $M$ & $10$, $20$ & $10$, $30$ & $10$, $40$ & $15$, $30$ & $15$, $40$ \\
  Number of coefficients & $1771$ & $5456$ & $12341$ & $5456$ & $12341$ \\ 
  Time step ($\Delta t$) & $\SI{1.73e-7}{\second}$ & $\SI{2.74e-8}{\second}$ & $\SI{4.67e-9}{\second}$& $\SI{1.85e-8}{\second}$ & $\SI{3.15e-9}{\second}$\\
  Number of time steps & $8754$ & $13230$ & $31899$ & $15403$ & $35888$ \\ 
  Total elapsed time & $\SI{164.68}{\second}$ & $\SI{251.13}{\second}$ & $\SI{1360.74}{\second}$ & $\SI{4367.15}{\second}$ & $\SI{13266}{\second}$ \\ 
  Elapsed time per time step & $\SI{1.88e-2}{\second}$ & $\SI{1.90e-2}{\second}$ & $\SI{4.27e-2}{\second}$ & $\SI{0.284}{\second}$ & $\SI{0.370}{\second}$ \\
\hline
\end{tabular}
\end{table}

\subsection{Simulation of unsteady flows}
Now we use another two numerical examples to study the evolution of the flow.
In both cases, we need to employ the steady shock structure simulated in Section
\ref{sec:shock} in the initial condition, while in this section, we expect that
the shock wave moves at the given Mach number. This can be achieved by the
following steps:
\begin{itemize}
\item Perform the same simulation as in Section \ref{sec:shock} using $\lu =
  (u_a, 0, 0)^T$.
\item In the numerical results, we perform a translation of the distribution
  function such that the fluid state in front of the shock has velocity zero.
  Such a translation can be simply implemented by setting $\lu = 0$.
\end{itemize}
The second step is in fact a transform of the frame of reference, which turns a
steady shock structure to a moving shock structure. Below we are going to study
the collision of two shock structures and the interaction between the shock and
the solid wall.

\subsubsection{Collision of two shocks}
We first consider the collision of two shock waves which move in
opposite directions with the Mach numbers $\mathit{Ma}=3.8$ and $6.5$
respectively. Precisely speaking, the shock with $\mathit{Ma}=3.8$ is
on the left of the domain and moves to the right, while the shock with
$\mathit{Ma}=6.5$ is on the right of the domain and moves to the
left. In our method, the parameters $\lu$ and $\lth$ must be constants
for all grid cells. Thus, the initial shock profiles are obtained by
the simulations in Section \ref{sec:shock} using $\lu = \bu_{a}$ with
the corresponding Mach numbers and $\lth = 0.8 \theta_{b}$, where
$\theta_{b}$ is computed with $\mathit{Ma}=6.5$. After re-setting
$\lu = 0$ and reversing the velocity of the shock with
$\mathit{Ma}=3.8$, we then obtain the initial state of this test. Due
to the existence of two shocks in opposite directions, the
distribution function can spread widely after the merge two shocks,
which makes the problem highly challenging.

Our numerical solutions of density $\rho$, temperature $T$, normal
stress $\sigma_{xx}$ and heat flux $q_{x}$ with $M_{0}=10$ and $M=20$
are presented in \figurename~\ref{fig:shock-col-binaryOM20} for
various time instants from $t=0$ to $1.2$. As shown in the figure, the
black solid line gives the corresponding initial state of these
quantities. After the collision of two shocks around the time $t=0.6$,
two new shock waves will be generated with a rarefaction wave standing
between them. The left shock wave moves to left from right with the
speed smaller than $\mathit{Ma}=6.5$, while the right shock wave moves
from left to right with the speed smaller than
$\mathit{Ma}=3.8$. Besides, the left shock wave moves much faster then
the right shock wave. These solution structures are consistent with
the corresponding solutions obtained by the classical Euler equations,
whose numerical results at $t=1.2$ are also given in
\figurename~\ref{fig:shock-col-binaryOM20} by gray solid line.

\begin{figure}[!htb]
  \centering
  \subfloat[Density, $\rho~({\rm kg\cdot m^{-3}})$]{\includegraphics[width=0.48\textwidth,clip]{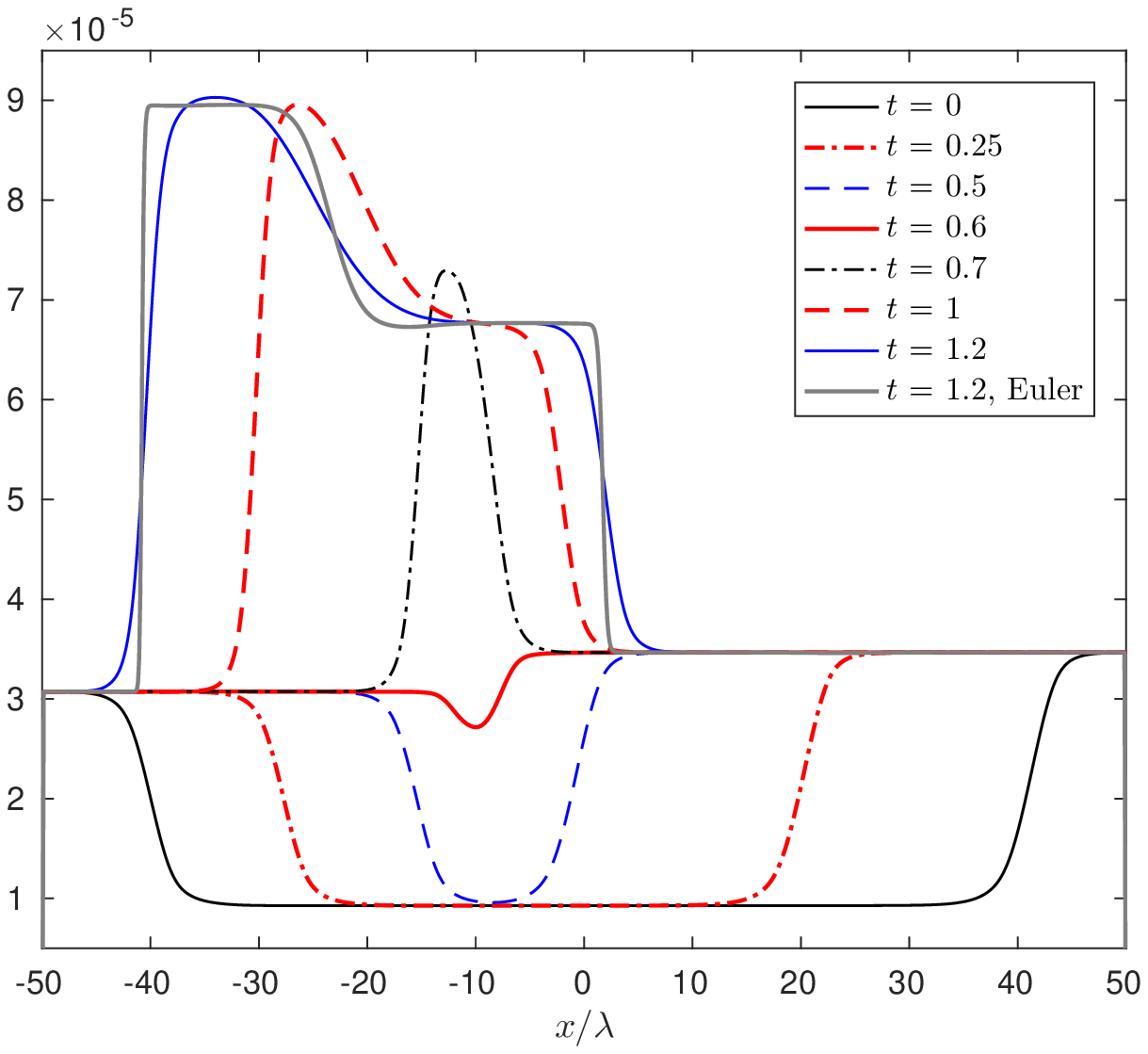}}\hfill
  \subfloat[Temperature, $T~({\rm K})$]{\includegraphics[width=0.5\textwidth,clip]{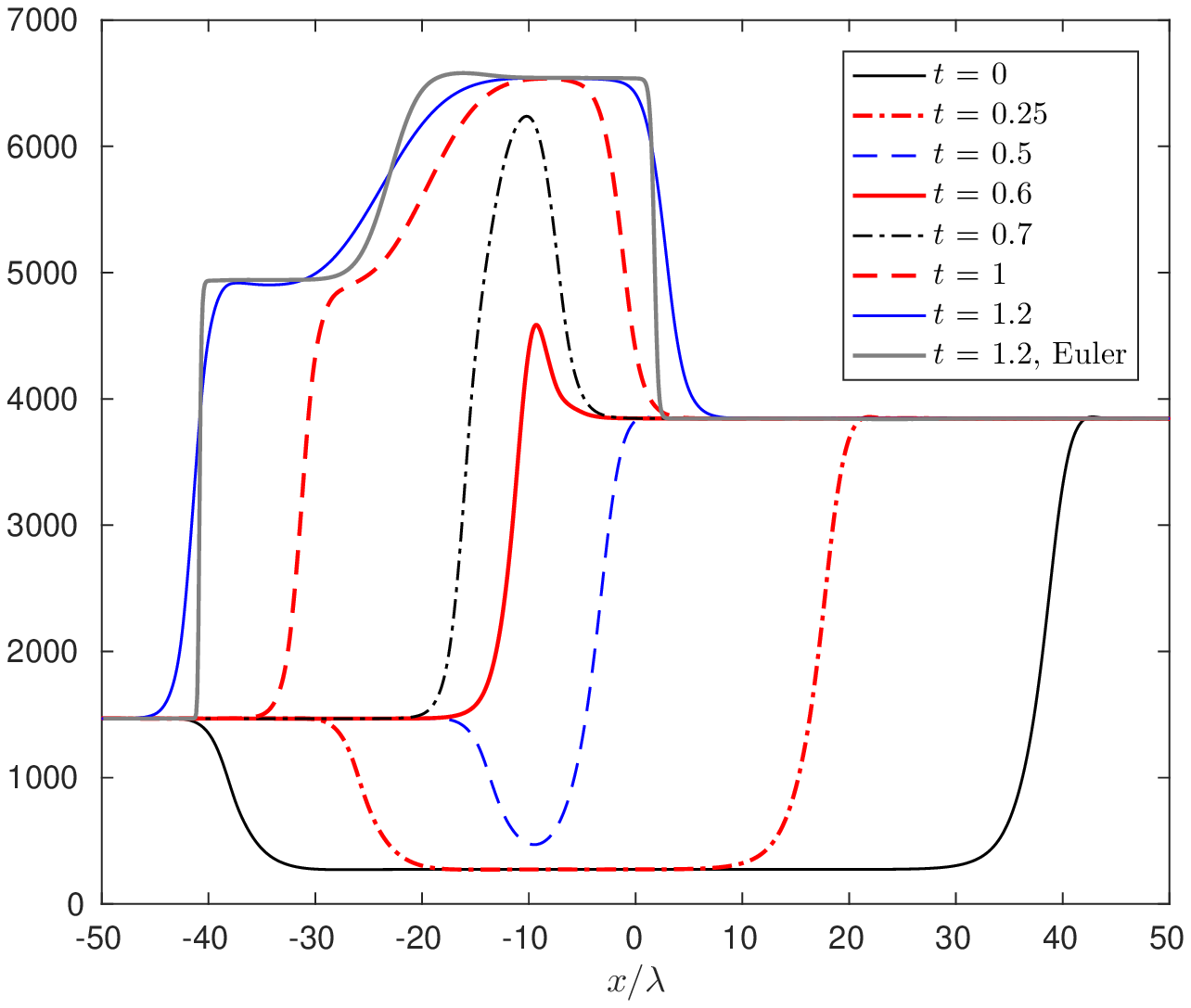}} \\
  \subfloat[Normal stress, $\sigma_{xx}~({\rm kg \cdot m^{-1}\cdot s^{-2}})$]{\includegraphics[width=0.48\textwidth,clip]{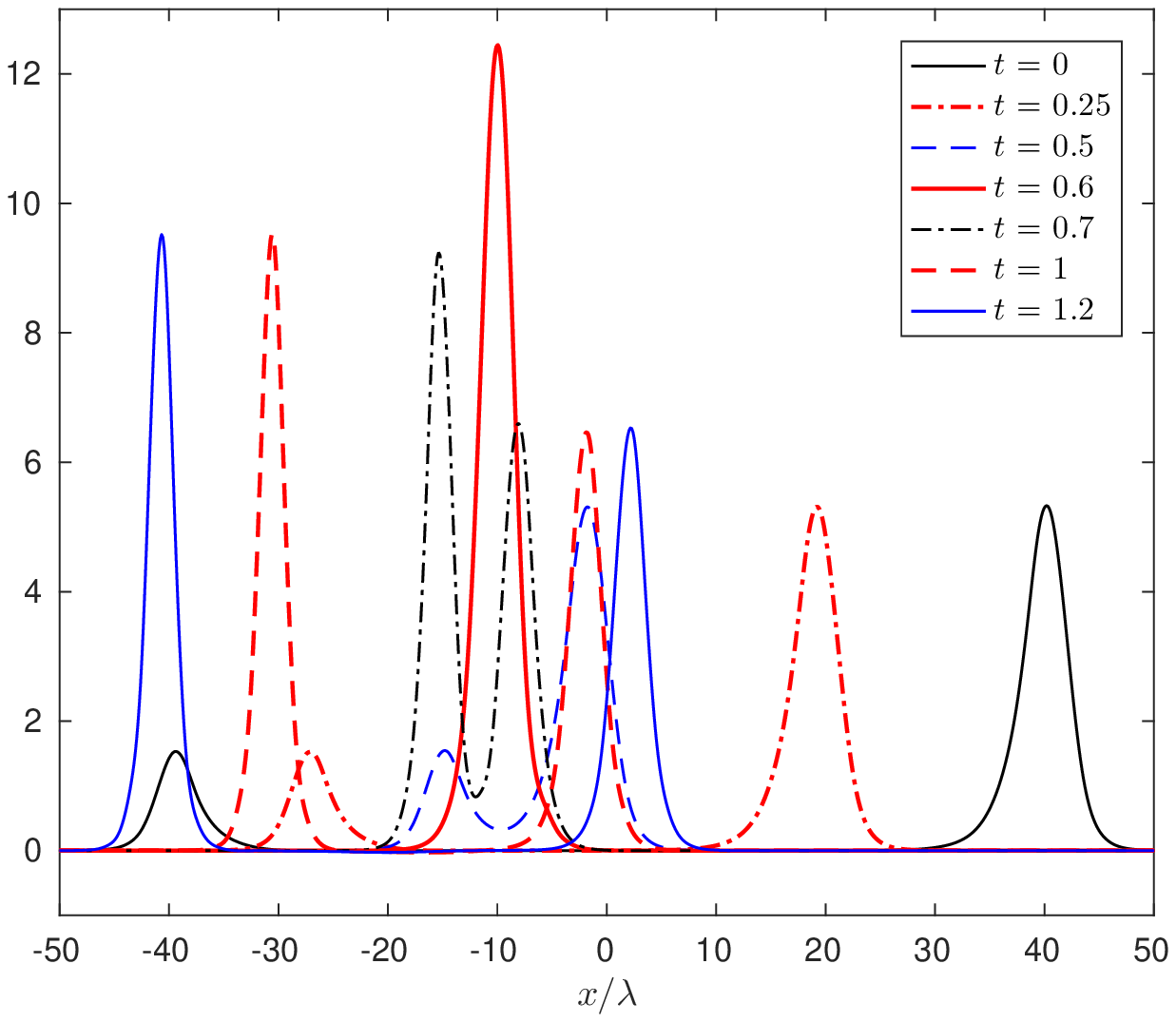}} \hfill
  \subfloat[Heat flux, $q_x~(\rm kg/s^{3})$]{\includegraphics[width=0.49\textwidth,clip]{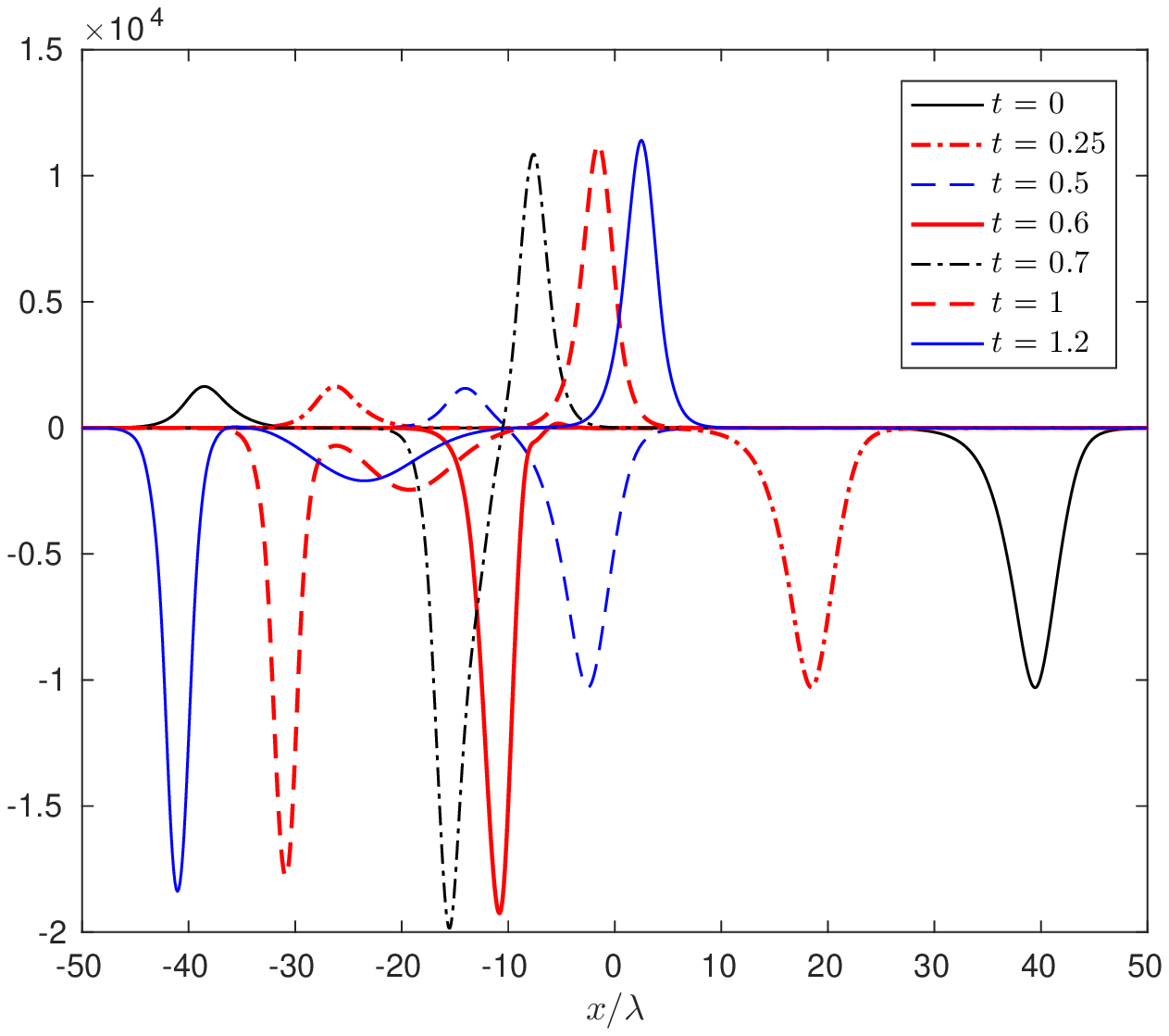}} 
  \caption{Solution for the collision of two shocks with
    $\mathit{Ma}=3.8$ and $6.5$ respectively at various time with
    $M_{0}=10$ and $M=20$. The gray solid line is the corresponding
    results from the Euler equations at $t=1.2$.}
  \label{fig:shock-col-binaryOM20}
\end{figure}

\subsubsection{Interaction of a shock and a solid wall}
This test investigates the interaction of a shock wave and a solid
wall. The shock wave moves from right to left, and the solid wall is
fixed on the left of the computational domain $[-30\lambda,
30\lambda]$ with the wall temperature $T_{a}^{W} =
\SI{273.15}{\kelvin}$ and the accommodation coefficient $\chi_{a}=1$.
It is observed that the temperature would be doubled after interacting
with the wall. So we use $\lu = \bu_{a}$ and $\lth = \theta_{b}$
together with $M_{0}=10$ and $M=40$ to prepare the initial shock
profile. Then the shock wave will move to left with the expected speed
by setting $\lu = 0$. For easier processing of the boundary condition,
a linearly reconstructed finite volume method with $600$ uniform grid
cells is performed, instead of the fifth-order WENO finite volume
method used in Section \ref{sec:shock}.

This example is even more challenging due to the existence of the
solid wall, which introduces discontinuity into the distribution
function. Numerical solutions of density $\rho$, temperature $T$,
normal stress $\sigma_{xx}$ and heat flux $q_{x}$ are presented in
\figurename~\ref{fig:shock-wall-Ma38} at various time instants from
$t=0$ to $1.2$ with $\mathit{Ma}=3.8$. The black solid line in the
figure represents the corresponding initial state of these
quantities. It is shown that the shock collides with the wall around
the time $t=0.6$. By the interaction with the wall, the particles
accumulate in a small region around the wall, resulting a significant
increase in density. A shock wave is then bounced back, with the speed
much slower than the original shock. Similar results can be obtained
for the case $\mathit{Ma}=6.5$, as given in
\figurename~\ref{fig:shock-wall-Ma65} at various time instants from
$t=0$ to $0.8$.

\begin{figure}[!htb]
  \centering
  \subfloat[Density, $\rho~({\rm kg\cdot m^{-3}})$]{\includegraphics[width=0.49\textwidth,clip]{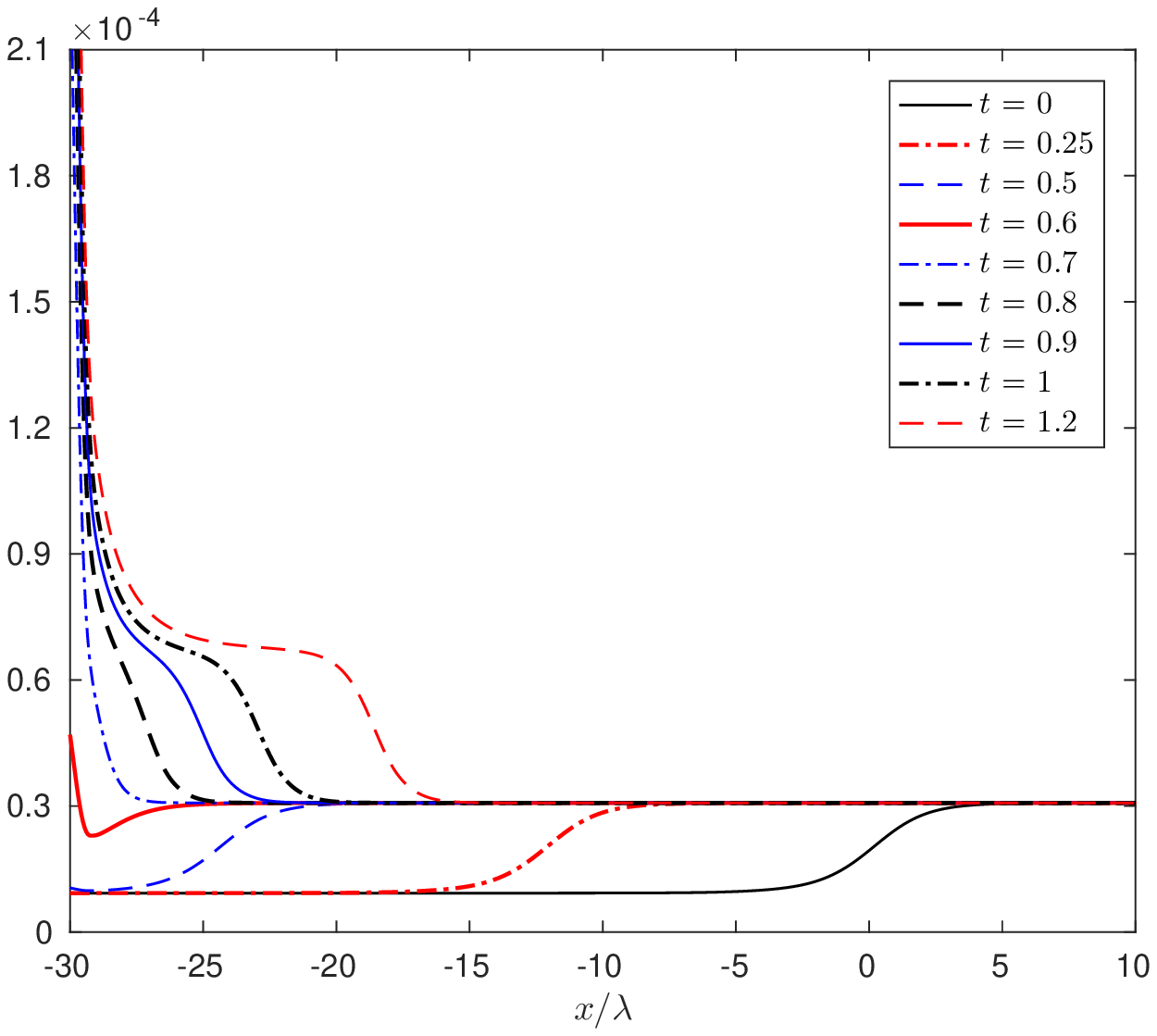}}\hfill
  \subfloat[Temperature, $T~({\rm K})$]{\includegraphics[width=0.5\textwidth,clip]{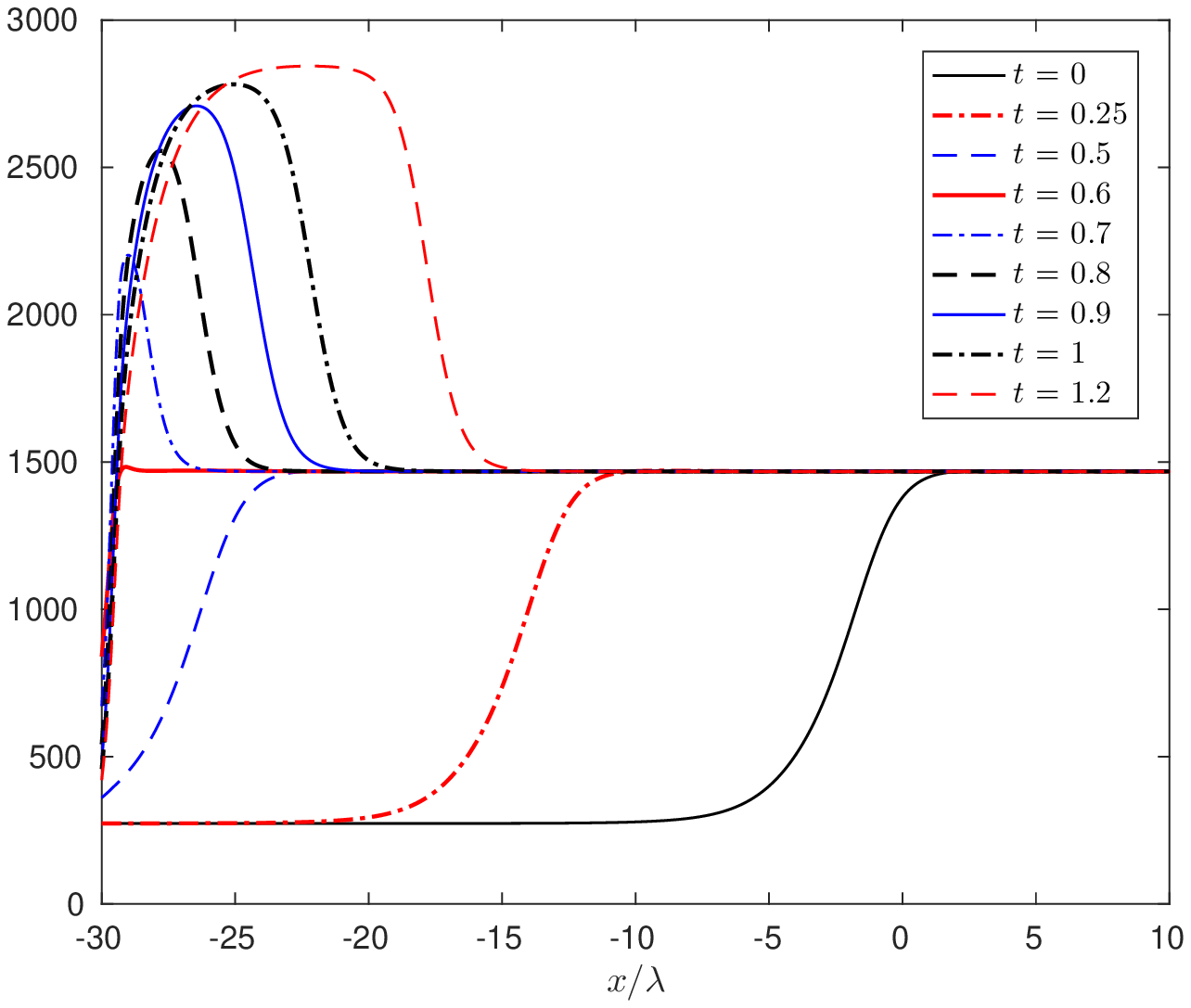}} \\
  \subfloat[Normal stress, $\sigma_{xx}~({\rm kg \cdot m^{-1}\cdot s^{-2}})$]{\includegraphics[width=0.48\textwidth,clip]{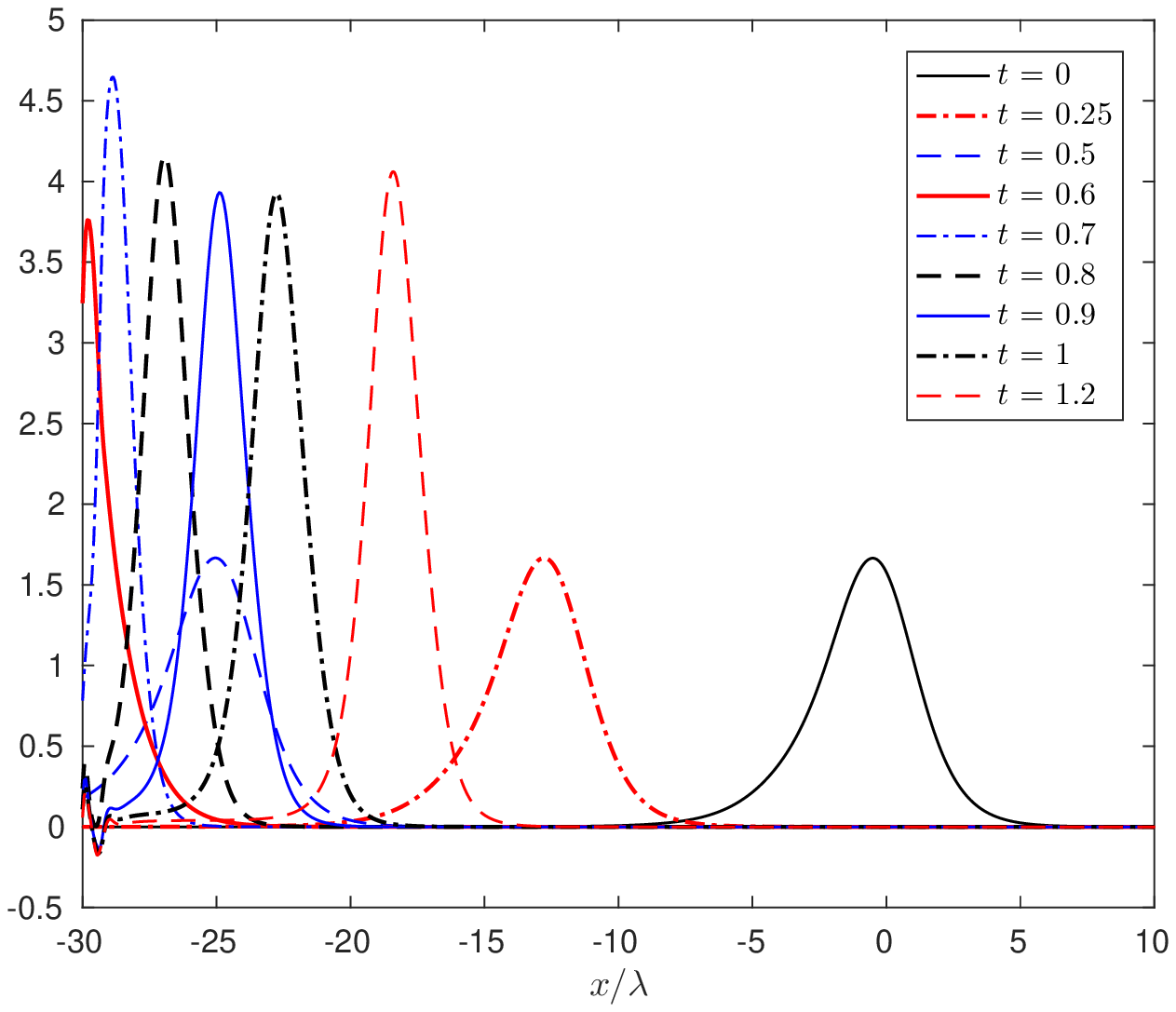}} \hfill
  \subfloat[Heat flux, $q_x~(\rm kg/s^{3})$]{\includegraphics[width=0.5\textwidth,clip]{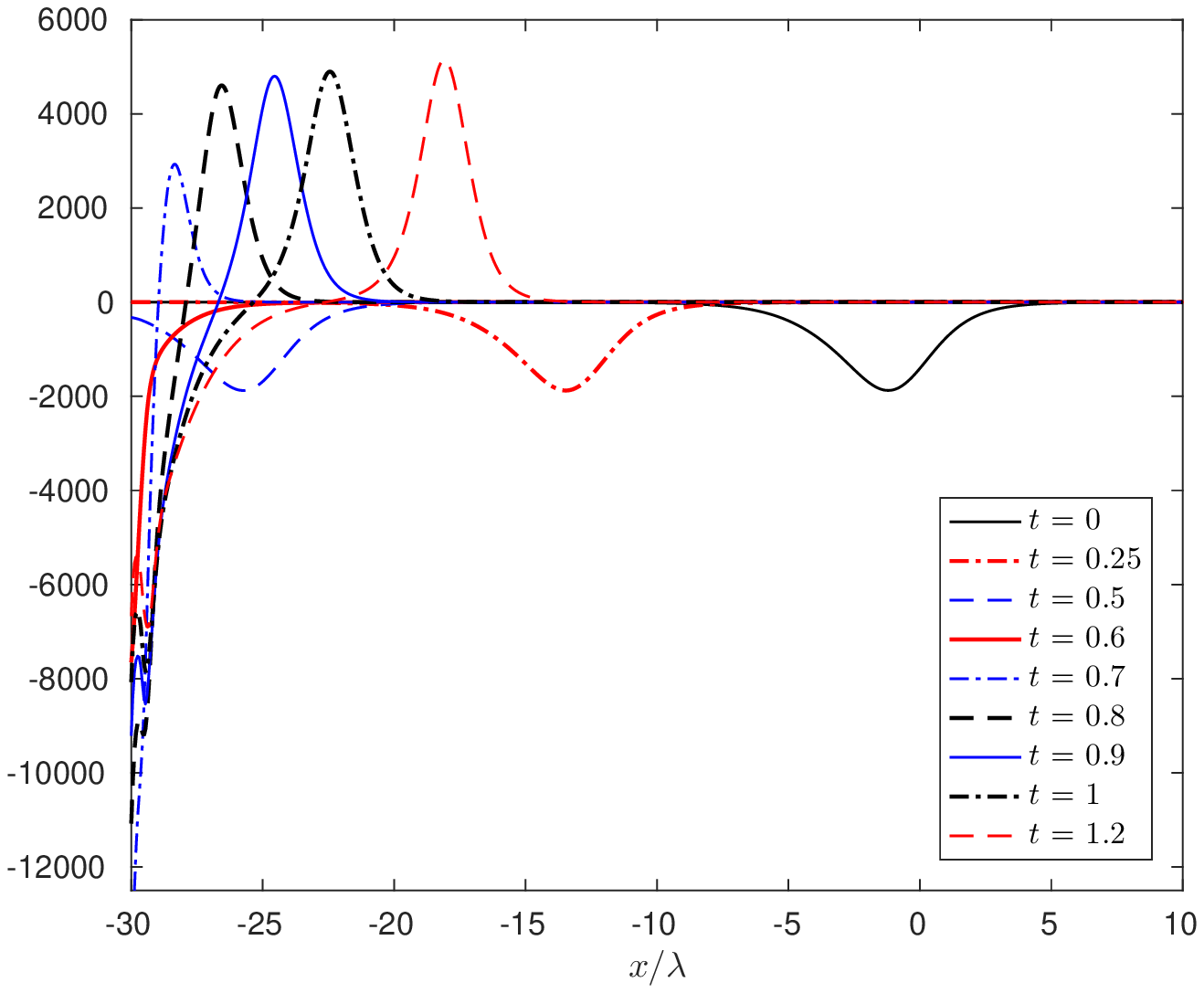}} 
  \caption{Solution for the interaction of a shock with
    $\mathit{Ma}=3.8$ and a solid wall at various time with $M_{0}=10$
    and $M=40$.}
  \label{fig:shock-wall-Ma38}
\end{figure}

\begin{figure}[!htb]
  \centering
  \subfloat[Density, $\rho~({\rm kg\cdot m^{-3}})$]{\includegraphics[width=0.49\textwidth,clip]{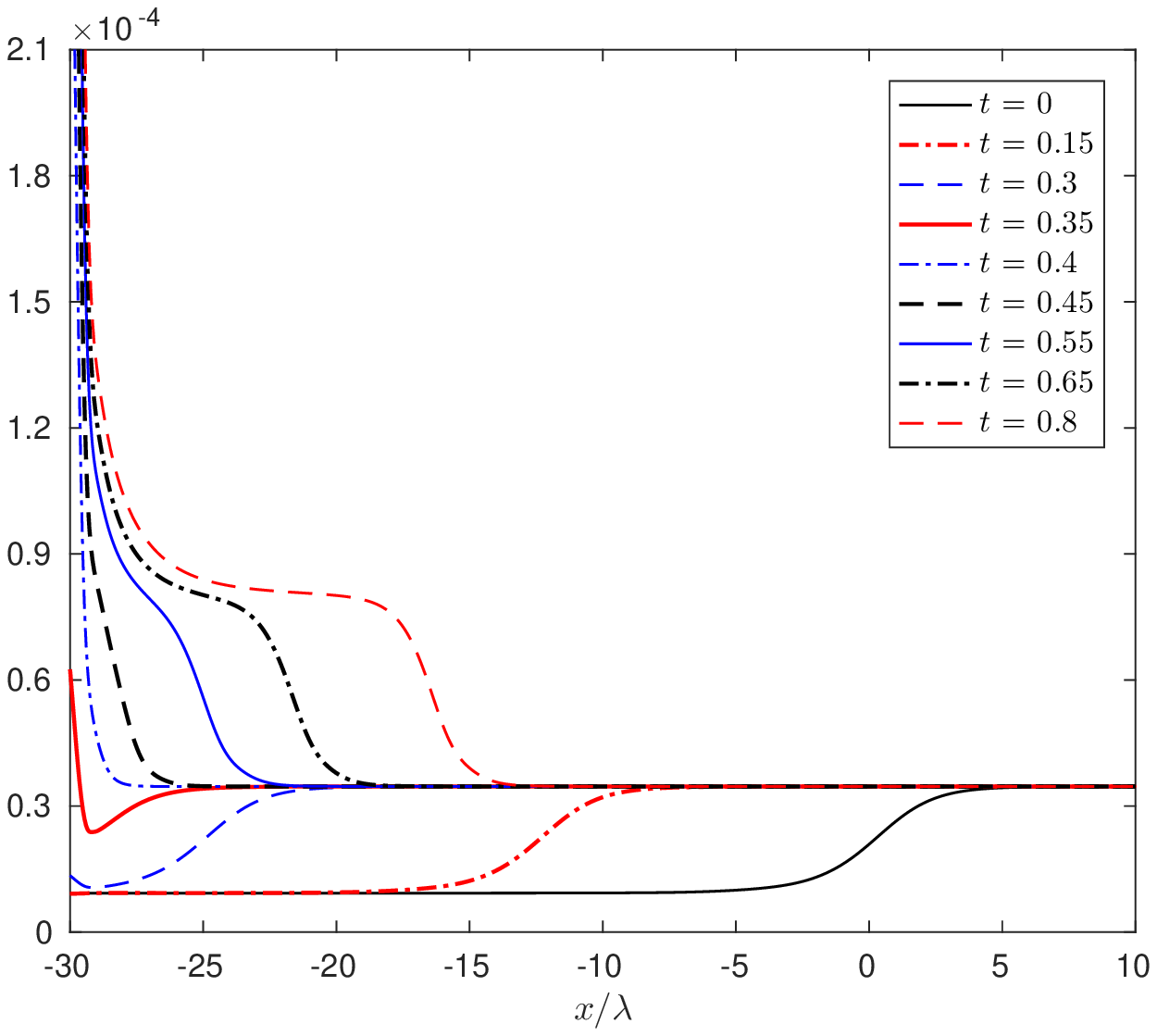}}\hfill
  \subfloat[Temperature, $T~({\rm K})$]{\includegraphics[width=0.5\textwidth,clip]{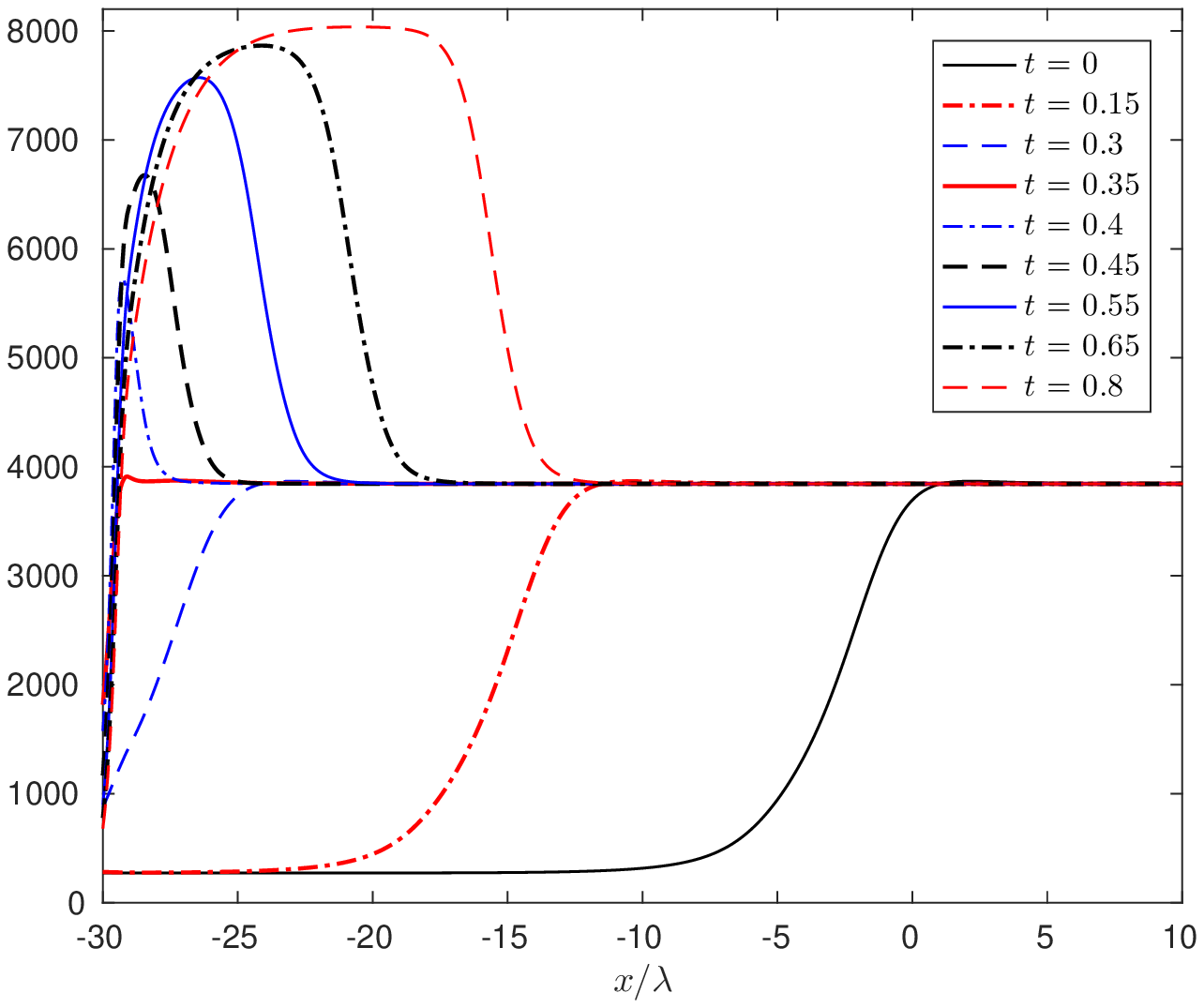}} \\
  \subfloat[Normal stress, $\sigma_{xx}~({\rm kg \cdot m^{-1}\cdot s^{-2}})$]{\includegraphics[width=0.49\textwidth,clip]{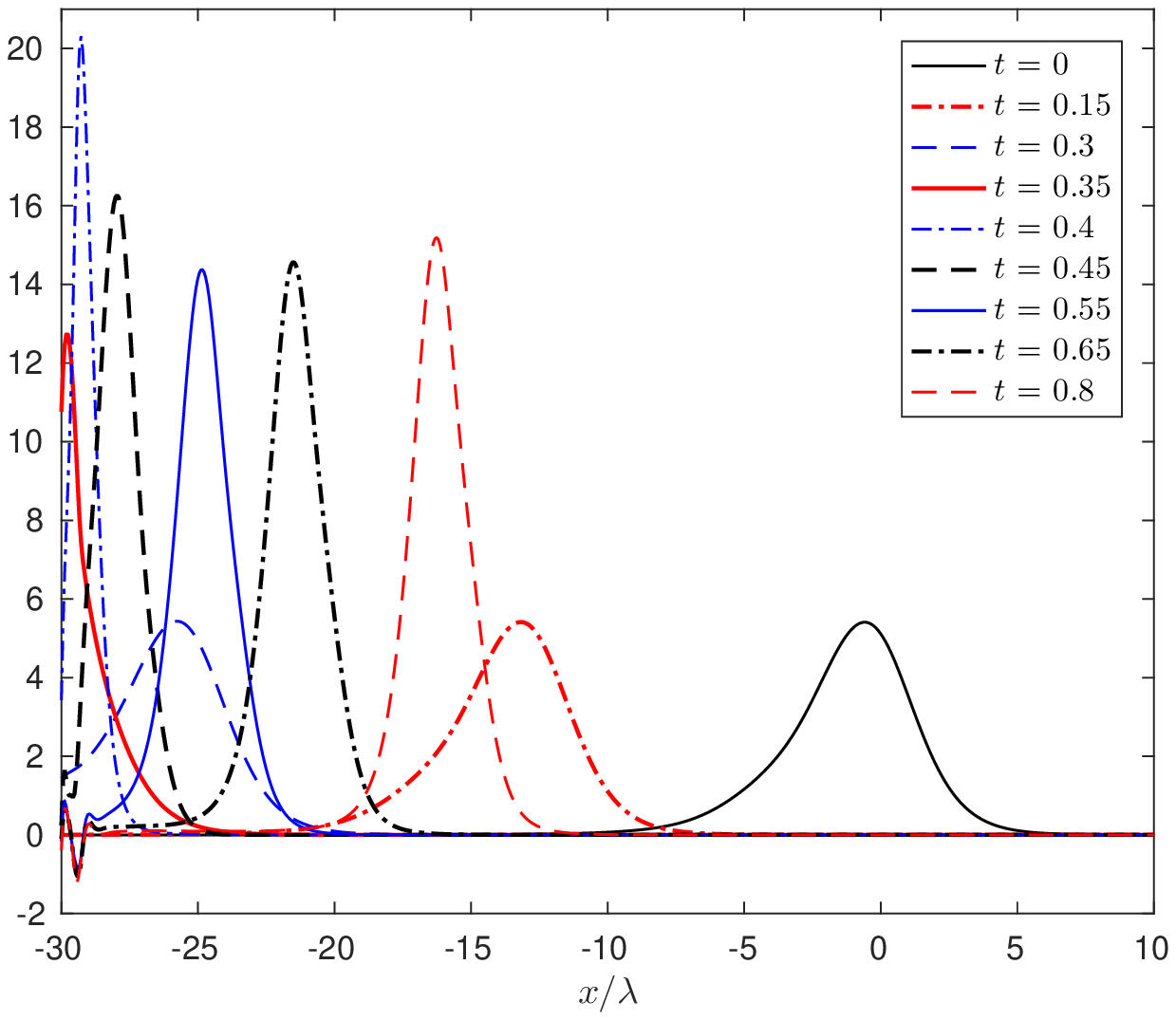}} \hfill
  \subfloat[Heat flux, $q_x~(\rm kg/s^{3})$]{\includegraphics[width=0.49\textwidth,clip]{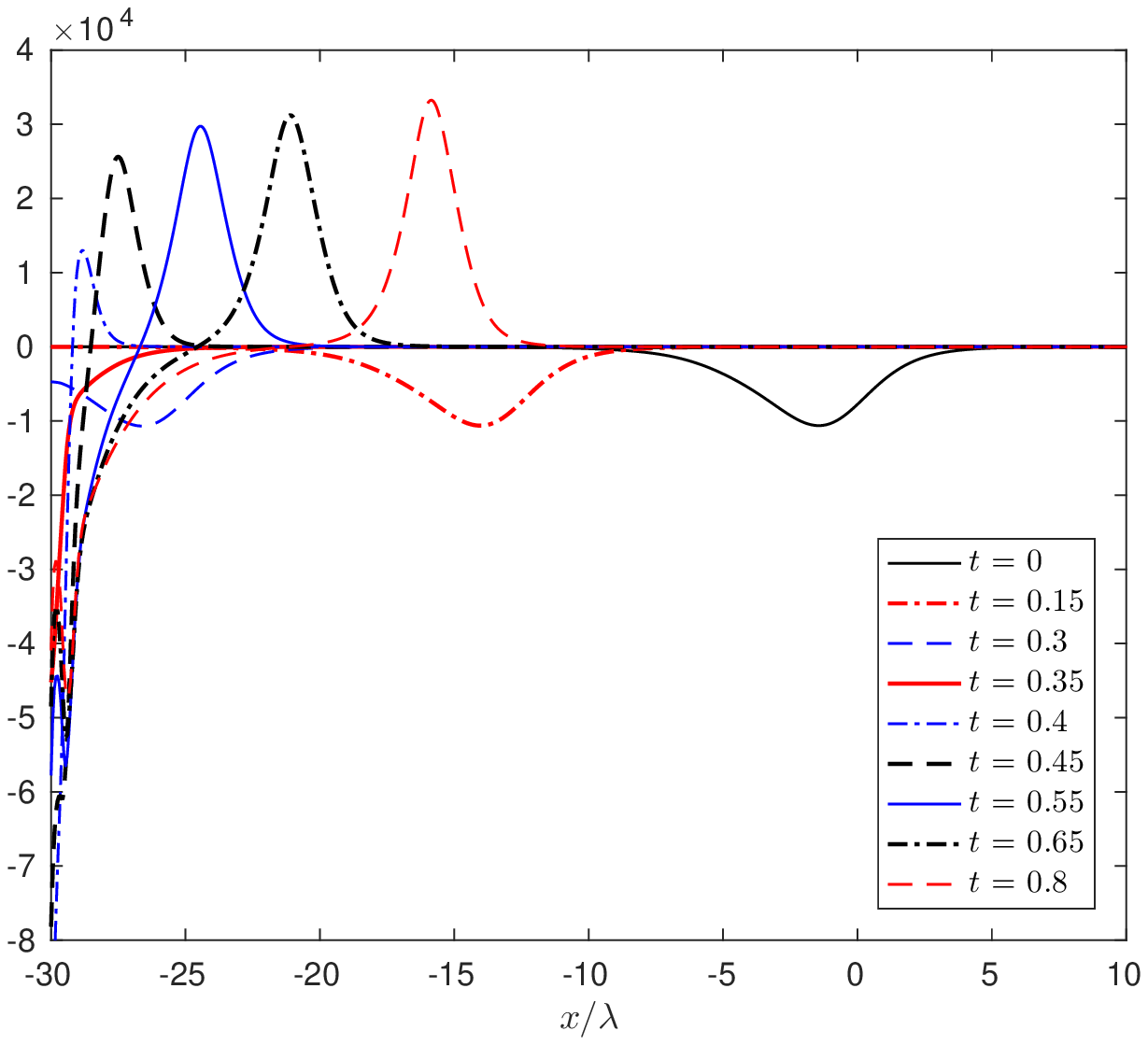}} 
  \caption{Solution for the interaction of a shock with
    $\mathit{Ma}=6.5$ and a solid wall at various time with $M_{0}=10$
    and $M=40$.}
  \label{fig:shock-wall-Ma65}
\end{figure}


\section{Conclusion} \label{sec:conclusion}
We have developed an efficient numerical scheme for the spatially inhomogeneous
Boltzmann equation based on the Burnett spectral method. Two major numerical
strategies are used: (1) we coupled the quadratic collision operator and the BGK
collision operator to balance the computational cost and the modelling
accuracy; (2) we introduced a zero term in the collision term to preserve the
steady state. Both steady and unsteady flows are solved as numerical examples.
Although only one spatial dimension (with three velocity dimensions) is
considered in the simulation, all the examples include large temperature
contrast or fast flow speed, which are still quite challenging. The
computational time shows high numerical efficiency of our method.

\appendix

\section{Inverse power law model}
In this section we provide a full description of the collision kernel of the
inverse power law model. For given $\kappa$, $\mm$ and $\eta$, we have
\begin{displaymath}
B(\bv - \bv_1, \bn) = \left( \frac{2\kappa}{\mm} \right)^{\frac{2}{\eta-1}}
  |\bv - \bv_1|^{\frac{\eta-5}{\eta-1}}
  \frac{W_0(\vartheta) W_0'(\vartheta)}{\sin \vartheta},
\end{displaymath}
where
\begin{displaymath}
\vartheta = \arccos \frac{|(\bv - \bv_1) \cdot \bn|}{|\bv - \bv_1|},
\end{displaymath}
and the function $W_0(\cdot)$ is defined by
\begin{displaymath}
W_0(\vartheta) = \sqrt{1 - y(\vartheta)}
  \left( \frac{\eta-1}{2} y(\vartheta) \right)^{-\frac{1}{\eta-1}},
\end{displaymath}
with $y(\cdot)$ being defined implicitly as
\begin{equation} \label{eq:y}
\int_0^1 \left( 1 - x^2 [1-y(\vartheta)] - x^{\eta-1} y(\vartheta) \right)^{-1/2}
  \sqrt{1 - y(\vartheta)} \,\mathrm{d}x = \vartheta.
\end{equation}
The viscosity coefficient $\overline{\mu}$ can be derived from the first-order
Chapman-Enskog expansion. In \eqref{eq:mu}, the function $A_2(\eta)$ is defined
by
\begin{equation} \label{eq:A_2}
A_2(\eta) = \int_0^{\pi/2} W_0(\vartheta) W_0'(\vartheta)
  \sin^2 (2\vartheta) \,\mathrm{d} \vartheta.
\end{equation}

When $\eta = 5$, the corresponding inverse power law model is also called the
Maxwell model. In this case, for any $l,m,n$, the function $p_{lmn}(\bv^*)
\omega(\bv^*)$ is the eigenfunction of the linearized collision operator $L[f]
= Q(f,\omega) + Q(\omega,f)$. Here we complete the definition of $p_{lmn}$ by
providing the definition of the Laguerre polynomials and spherical harmonics,
which appear in \eqref{eq:p_lmn}:
\begin{equation} \label{eq:LY}
L_n^{(\alpha)}(x) = \frac{x^{-\alpha} \exp(x)}{n!}
  \frac{\mathrm{d}^n}{\mathrm{d}x^n} [x^{n+\alpha} \exp(-x)], \quad
Y_l^m(\bn) = Y_l^m(\theta,\varphi) =
  \sqrt{\frac{2l+1}{4\pi} \frac{(l-m)!}{(l+m)!}}
  P_l^m(\cos \vartheta) \exp(\mathrm{i} m \varphi),
\end{equation}
where $(\vartheta,\varphi)$ is the spherical coordinates of $\bn$, i.e. $\bn =
(\cos \vartheta, \sin \vartheta \cos \varphi, \sin \vartheta \sin \varphi)^T$,
and $P_l^m(\cdot)$ is the associated Legendre function defined by
\begin{displaymath}
P_l^m(x) = \frac{(-1)^m}{2^l l!} (1-x^2)^{m/2}
  \frac{\mathrm{d}^{l+m}}{\mathrm{d}x^{l+m}} [(x^2-1)^l].
\end{displaymath}
When $m = 0$, the polynomial $P_l^0(x)$ is the Legendre polynomial of degree
$l$, which is often denoted by $P_l(x)$.

\section{Coefficients in the collision operator for inverse power law models}
To complete the description of the ODE system \eqref{eq:ode}, we summarize in
this section the results in \cite{Wang2019, Cai2019}, and provide the complete
process to compute these coefficients. The computational cost for computing all
these coefficients is $O(M^{14})$ (see \cite{Cai2019} for details), but this
needs to be done only once before the numerical simulation.

Suppose the index $\eta$ in the power potential is given, below we are going to
provide a sequence of formulas, by which the coefficients $A_{lmn}^{l_1 m_1 n_1,
l_2 m_2 n_2}$ can be computed step by step:
\begin{itemize}
\item Compute the following integral $\mathcal{I}_k$ for $k \leqslant 2M$:
  \begin{displaymath}
  \mathcal{I}_k = \int_0^1 \big[ P_k \big({-}\cos 2\vartheta(y) \big) - 1 \big]
    [2(1-y) + (\eta-1)y] [(\eta-1)y]^{-\frac{\eta+1}{\eta-1}} \,\mathrm{d} y,
  \end{displaymath}
  where $\vartheta(y)$ is the inverse function of $y(\vartheta)$ defined in
  \eqref{eq:y}, and $P_k$ is the Legendre polynomial of degree $k$.
\item Compute $K_{mn}^{kl}$ for $k \leqslant 2M$, $l \leqslant M$, $m\leqslant
  \lfloor k/2 \rfloor$, $n\leqslant \lfloor l/2 \rfloor$, $k -2m = l - 2n$ by
  \begin{displaymath}
  \begin{split}
  K_{mn}^{kl} &= (-1)^{m+n} 2^{\frac{\eta-3}{\eta-1} + k-2m}
    \Gamma \left( k-2m+2 - \frac{2}{\eta-1} \right) C(\eta)
    \mathcal{I}_{k-2m} \times {} \\
  & \qquad \sum_{i=0}^{\min(m,n)}
      \begin{pmatrix} \frac{1}{2} - \frac{2}{\eta-1} \\ m-i \end{pmatrix}
      \begin{pmatrix} \frac{1}{2} - \frac{2}{\eta-1} \\ n-i \end{pmatrix}
      \begin{pmatrix} k-2m+i+1 - \frac{2}{\eta-1} \\ i \end{pmatrix},
  \end{split}
  \end{displaymath}
  where
  \begin{displaymath}
  C(\eta) = \frac{5}{4^{\frac{3}{2} - \frac{2}{\eta-1}}
    \sqrt{\pi} A_2(\eta)\Gamma(4-2(\eta-1))}.
  \end{displaymath}
\item For all $k \leqslant M$, define the homogeneous polynomial $S_k(\bv,\bw)$
  for $\bv,\bw \in \mathbb{R}^3$ by the following recursive formulas:
  \begin{displaymath}
  \begin{gathered}
    S_0(\bv, \bw) = 1, \qquad S_1(\bv, \bw) = \bv\cdot \bw,\\
    S_{k+1}(\bv, \bw) = \frac{2k+1}{k+1}(\bv\cdot \bw)S_k(\bv, \bw) -
    \frac{k}{k+1}(|\bv||\bw|)^2S_{k-1}(\bv, \bw), \quad k \geqslant 2.
  \end{gathered}
  \end{displaymath}
  By the above definition, write $S_k(\cdot, \cdot)$ as
  \begin{displaymath}
  S_k(\bv, \bw) = \sum_{k_1 + k_2 + k_3 = k} \sum_{l_1 + l_2 + l_3 = k}
    S_{k_1 k_2 k_3}^{l_1 l_2 l_3}
      v_1^{k_1} v_2^{k_2} v_3^{k_3} w_1^{l_1} w_2^{l_2} w_3^{l_3},
  \end{displaymath}
  and find all the coefficients $S_{k_1 k_2 k_3}^{l_1 l_2 l_3}$.
\item Compute $C_{m_1 m_2 m_3}^{k_1 k_2 k_3}$ for $k_1 + k_2 + k_3 \leqslant 2M$
  and $m_1 + m_2 + m_3 \leqslant M$ by
  \begin{displaymath}
  C_{m_1 m_2 m_3}^{k_1 k_2 k_3} =
    \frac{(-1)^{m_1 + m_2 + m_3} 4\pi (m_1 + m_2 + m_3)!}
      {(2(k_1 + k_2 + k_3 - m_1 - m_2 - m_3) + 1)!}
    \frac{k_1! k_2! k_3!}{m_1! m_2! m_3!}.
  \end{displaymath}
\item Compute $\gamma_{k_1 k_2 k_3}^{l_1 l_2 l_3}$ for $k_1 + k_2 + k_3
  \leqslant 2M$ and $l_1 + l_2 + l_3 \leqslant M$ by
  \begin{displaymath}
  \begin{split}
  \gamma_{k_1 k_2 k_3}^{l_1 l_2 l_3} &=
    \sum_{m_1=0}^{\lfloor k_1/2 \rfloor}
    \sum_{m_2=0}^{\lfloor k_2/2 \rfloor}
    \sum_{m_3=0}^{\lfloor k_3/2 \rfloor}
    \sum_{n_1=0}^{\lfloor l_1/2 \rfloor}
    \sum_{n_2=0}^{\lfloor l_2/2 \rfloor}
    \sum_{n_3=0}^{\lfloor l_3/2 \rfloor}(2k-4m+1) \times \\
  & \qquad C_{m_1 m_2 m_3}^{k_1 k_2 k_3} C_{n_1 n_2 n_3}^{l_1 l_2 l_3}
    S_{k_1-2m_1,k_2-2m_2,k_3-2m_3}^{l_1-2n_1,l_2-2n_2,l_3-2n_3}
    K_{mn}^{kl},
  \end{split}
  \end{displaymath}
  where
  \begin{displaymath}
  k = k_1 + k_2 + k_3, \quad
  l = l_1 + l_2 + l_3, \quad
  m = m_1 + m_2 + m_3, \quad
  n = n_1 + n_2 + n_3.
  \end{displaymath}
\item Compute $a_{i'j'}^{ij}$ for $i\leqslant M$, $j \leqslant M$, $i'\leqslant
  2M$, $j'\leqslant 2M$, $i+j = i'+j'$ by
  \begin{displaymath}
  a_{i'j'}^{ij} = 2^{-(i'+j')/2} i! j!
  \sum_{s=\max(0,i'-j)}^{\min(i',i)}
  \frac{(-1)^{j'-i+s}}{s!(i-s)!(i'-s)!(j'-i+s)!}.
  \end{displaymath}
\item Compute $B_{k_1 k_2 k_3}^{i_1 i_2 i_3, j_1 j_2 j_3}$ for $i_1 + i_2 + i_3
  \leqslant M$, $j_1 + j_2 + j_3 \leqslant M$, $k_1 + k_2 + k_3 \leqslant M$ by
  \begin{displaymath}
  B_{k_1k_2k_3}^{i_1 i_2 i_3, j_1 j_2 j_3}
    = \sum_{i'_1=0}^{\min(i_1+j_1,k_1)}
      \sum_{i'_2=0}^{\min(i_2+j_2,k_3)}
      \sum_{i'_3=0}^{\min(i_3+j_3,k_3)}
      \frac{ 2^{-k/2}}{2^3\pi^{3/2}}\frac{1}{l'_1!l'_2!l'_3!}
      a_{i'_1j'_1}^{i_1j_1}a_{i'_2j'_2}^{i_2j_2}
      a_{i'_3j'_3}^{i_3j_3}\gamma_{j'_1j'_2j'_3}^{l'_1l'_2l'_3},
  \end{displaymath}
  where
  \begin{displaymath}
  j_s' = i_s + j_s - i_s', \quad l_s' = k_s - i_s', \qquad s=1,2,3.
  \end{displaymath}
\item Given that $C_{000}^{000} = 1$ and
  \begin{displaymath}
  C_{lmn}^{k_1 k_2 k_3} = 0, \qquad |m| > l \text{ or } l < 0 \text{ or } n < 0,
  \end{displaymath}
  compute the coefficients $C_{lmn}^{k_1 k_2 k_3}$ for $k_1 + k_2 + k_3
  \leqslant M$, $l+2n = k_1 + k_2 + k_3$, $-l \leqslant m \leqslant l$ by
  solving the following equations:
  \begin{equation} \label{eq:C}
  \begin{aligned}
    a_{l,m+1,n}^{(-1)}C_{l+1, m, n}^{k_1 k_2 k_3} + b_{l,m+1,n}^{(-1)}C_{l-1, m,
        n+1}^{k_1 k_2 k_3} &= \frac{1}{2}k_1C_{l,m+1,n}^{k_1-1,k_2,k_3} -
      \frac{\mathrm{i}}{2} k_2 C_{l,m+1,n}^{k_1,k_2-1,k_3}, \\
    a_{l,m,n}^{(0)}C_{l+1, m, n}^{k_1 k_2 k_3} + b_{l,m,n}^{(0)}C_{l-1, m,
        n+1}^{k_1 k_2 k_3} &= k_3C_{l,m,n}^{k_1,k_2,k_3-1}, \\
    a_{l,m-1,n}^{(1)}C_{l+1, m, n}^{k_1 k_2 k_3} + b_{l,m-1,n}^{(1)}C_{l-1, m,
        n+1}^{k_1 k_2 k_3} &= -\frac{1}{2}k_1C_{l,m-1,n}^{k_1-1,k_2,k_3} -
      \frac{\mathrm{i}}{2} k_2C_{l,m-1,n}^{k_1,k_2-1,k_3},
  \end{aligned}     
  \end{equation}
  where
  \begin{align*}
  a_{lmn}^{(\mu)} &= \frac{1}{2^{|\mu|}}
    \sqrt{\frac{(2(n+l) +3)[l+(2\delta_{1,\mu}-1)m + \delta_{1,\mu} + 1]
      [l+(2\delta_{1,\mu}-1)m + \delta_{1,\mu} + 1]}{(2l+1)(2l+3)}}, \\
  b_{lmn}^{(\mu)} &= \frac{(-1)^{\mu+1}}{2^{|\mu|}}
    \sqrt{\frac{2(n+1)[l-(2\delta_{1,\mu}-1)m - \delta_{1,\mu}]
      [l-(2\delta_{1,\mu}-1)m - \delta_{1,\mu}]}{(2l-1)(2l+1)}},
  \quad \mu = -1, 0, 1.
  \end{align*}
  The equations \eqref{eq:C} can be applied recursively. Note that \eqref{eq:C}
  includes three equations. They are always consistent so that we can use
  \eqref{eq:C} to solve two coefficients $C_{l+1, m, n}^{k_1 k_2 k_3}$ and
  $C_{l-1,m,n+1}^{k_1 k_2 k_3}$ based on the knowledge of
  $C_{l,m-1,n}^{k_1-1,k_2,k_3}$, $C_{l,m-1,n}^{k_1,k_2-1,k_3}$ and
  $C_{l,m-1,n}^{k_1,k_2,k_3-1}$.
\item Compute the coefficients $A_{lmn}^{l_1 m_1 n_1, l_2 m_2 n_2}$ for $l +
  2n \leqslant M$, $l_1 + 2n_1 \leqslant M$, $l_2 + 2n_2 \leqslant M$, $-l_1
  \leqslant m_1 \leqslant l_1$, $-l_2 \leqslant m_2 \leqslant l_2$, $m = m_1 +
  m_2$ by
  \begin{displaymath}
  \begin{split}
    A_{lmn}^{l_1m_1n_1, l_2m_2n_2} = \sum_{k_1 + k_2 + k_3 = l+2n}
      \sum_{i_1 + i_2 + i_3 = l} \sum_{j \in I_{l_2 + 2n_2}}
      \overline{C_{lmn}^{k_1k_2k_3}} C_{l_1m_1n_1}^{i_1 i_2 i_3}
        C_{l_2m_2n_2}^{j_1 j_2 j_3} B_{k_1 k_2 k_3}^{i_1 i_2 i_3, j_1 j_2 j_3}.
  \end{split}
  \end{displaymath}
\end{itemize}

\section{Some lemmas}
In this section we present some lemmas as the preparation for the proof of
Theorem \ref{thm:Maxwellian}. In the equations appearing in the lemmas below,
we always assume that the indices $l,m,n$ satisfy $(l,m,n) \in \mathbb{N}
\times \mathbb{Z} \times \mathbb{N}$, and $-l \leqslant m \leqslant l$. Any
quantity with indices $l',m',n'$ is considered as zero if $|m'| > l'$ or $l' <
0$ or $n' < 0$ (e.g. $p_{l+1,m,n}$ with $l = m$).

\begin{lemma} \label{lem:p_property}
Burnett polynomials $p_{lmn}(\cdot)$ satisfy the following properties:
\begin{gather}
\label{eq:recurrence}
\begin{split}
v_x^* p_{lmn}(\bv^*) &= 
  \sqrt{n+l+3/2} \gamma_{l+1,m} p_{l+1,m,n}(\bv^*)
  - \sqrt{n} \gamma_{l+1,m} p_{l+1,m,n-1}(\bv^*) \\
& \qquad + \sqrt{n+l+1/2} \gamma_{lm} p_{l-1,m,n}(\bv^*)
  - \sqrt{n+1} \gamma_{lm} p_{l-1,m,n+1}(\bv^*),
\end{split} \\
\label{eq:recurrence0}
v_x^* p_{lm0}(\bv^*) = \sqrt{l+3/2} \gamma_{l+1,m} p_{l+1,m,0}(\bv^*) +
  \sqrt{\frac{1}{l+1/2}} \gamma_{lm} \frac{|\bv^*|^2}{2} p_{l-1,m,0}(\bv^*), \\
\label{eq:derivative}
\frac{\mathrm{d}p_{lmn}(\bv^*)}{\mathrm{d}v_x^*} =
  \sqrt{n+l+1/2} \gamma_{lm} p_{l-1,m,n}(\bv^*) -
    \sqrt{n} \gamma_{l+1,m} p_{l+1,m,n-1}(\bv^*),
\end{gather}
where the $\gamma$ symbol is defined in \eqref{eq:gamma}.
\end{lemma}

\begin{proof}
The proof of this lemma requires the following identities \cite[see equations
(8.5.3)(22.7.31)(22.7.30)]{Abramowitz}:
\begin{align*}
(l-m+1) P_{l+1}^m(x) &= (2l+1) x P_l^m(x) - (l+m) P_{l-1}^m(x), \\
x L_n^{(\alpha+1)}(x) &= (n+\alpha+1) L_n^{(\alpha)}(x) - (n+1) L_{n+1}^{(\alpha)}(x), \\
L_n^{(\alpha-1)}(x) &= L_n^{(\alpha)}(x) - L_{n-1}^{(\alpha)}(x).
\end{align*}
Representing $\bv^*$ by spherical coordinates $(r \cos \vartheta, r \sin
\vartheta \cos \varphi, r \sin \vartheta \sin \varphi)$, we get
\begin{displaymath}
\begin{split}
& v_x^* p_{lmn}(\bv^*) = r \cos \vartheta \cdot
  \sqrt{\frac{2^{1-l} \pi^{3/2} n!}{\Gamma(n+l+3/2)}}
  L_n^{(l+1/2)} \left( \frac{r^2}{2} \right) r^l \cdot
  \sqrt{\frac{2l+1}{4\pi} \frac{(l-m)!}{(l+m)!}}
  P_l^m(\cos \vartheta) \exp(\mathrm{i} m \varphi) \\
& \quad = \sqrt{\frac{2^{1-l} \pi^{3/2} n!}{\Gamma(n+l+3/2)}}
  \sqrt{\frac{2l+1}{4\pi} \frac{(l-m)!}{(l+m)!}}
  L_n^{(l+1/2)} \left( \frac{r^2}{2} \right) r^{l+1} \times {} \\
& \quad \hspace{50pt} \left(
    \frac{l-m+1}{2l+1} P_{l+1}^m(\cos \vartheta) +
    \frac{l+m}{2l+1} P_{l-1}^m(\cos \vartheta)
  \right) \exp(\mathrm{i} m \varphi) \\
& \quad = \sqrt{\frac{2^{1-l} \pi^{3/2} n!}{\Gamma(n+l+3/2)}}
  \sqrt{\frac{2l+1}{4\pi} \frac{(l-m)!}{(l+m)!}} \Bigg(
  \frac{l-m+1}{2l+1} \left[
    L_n^{(l+3/2)} \left( \frac{r^2}{2} \right) -
    L_{n-1}^{(l+3/2)} \left( \frac{r^2}{2} \right)
  \right] r^{l+1} P_{l+1}^m(\cos \vartheta) \\
& \quad \qquad + \frac{2(l+m)}{2l+1} \left[
    (n+l+1/2) L_n^{(l-1/2)} \left( \frac{r^2}{2} \right)
    - (n+1) L_{n+1}^{(l-1/2)}\left( \frac{r^2}{2} \right)
  \right] r^{l-1} P_{l-1}^m(\cos \vartheta) \Bigg) \exp(\mathrm{i} m \varphi) \\
& \quad = \sqrt{\frac{2(l-m+1)(l+m+1)(n+l+3/2)}{(2l+1)(2l+3)}} p_{l+1,m,n}(\bv^*)
  - \sqrt{\frac{2(l-m+1)(l+m+1)n}{(2l+1)(2l+3)}} p_{l+1,m,n}(\bv^*) \\
& \quad \qquad + \sqrt{\frac{2(l-m)(l+m)(n+l+1/2)}{(2l+1)(2l-1)}} p_{l-1,m,n}(\bv^*)
  - \sqrt{\frac{2(l-m)(l+m)(n+1)}{(2l+1)(2l-1)}} p_{l-1,m,n+1}(\bv^*).
\end{split}
\end{displaymath}
Equation \eqref{eq:recurrence} is a direct result of the above equality by
inserting the definitions of $\gamma_{lm}$ \eqref{eq:gamma}. When $n=0$,
\begin{displaymath}
\begin{split}
v_x^* p_{lm0}(\bv^*) &= \sqrt{\frac{2^{1-l} \pi^{3/2}}{\Gamma(l+3/2)}}
  \sqrt{\frac{2l+1}{4\pi} \frac{(l-m)!}{(l+m)!}} r^{l+1} \left(
    \frac{l-m+1}{2l+1} P_{l+1}^m(\cos \vartheta) +
    \frac{l+m}{2l+1} P_{l-1}^m(\cos \vartheta)
  \right) \exp(\mathrm{i} m \varphi) \\
&= \sqrt{\frac{(l-m+1)(l+m+1)(l+3/2)}{(2l+1)(2l+3)}} p_{l+1,m,0}(\bv^*) +
  \sqrt{\frac{(l-m)(l+m)}{2(l+1/2)(2l+1)(2l-1)}} r^2 p_{l-1,m,0}(\bv^*).
\end{split}
\end{displaymath}
Again, the equality \eqref{eq:recurrence0} can be obtained by inserting the
definition of $\gamma_{lm}$.

Now we prove \eqref{eq:derivative}. It is clear that the equality holds for $l
= n = 0$. If $l > 0$ or $n > 0$, we just need to compute the following integral
for any $l',m',n'$ satisfying $l'+2n' < l+2n$:
\begin{displaymath}
\begin{split}
& \int_{\mathbb{R}^3} [p_{l'm'n'}(\bv^*)]^{\dagger}
  \frac{\mathrm{d} p_{lmn}(\bv^*)}{\mathrm{d} v_x^*} \omega(\bv^*)
  \,\mathrm{d}\bv^* \\
={} & -\int_{\mathbb{R}^3}
  \left( \frac{\mathrm{d} p_{l'm'n'}(\bv^*)}{\mathrm{d}v_x^*} \right)^{\dagger}
  p_{lmn}(\bv^*) \omega(\bv^*) \,\mathrm{d}\bv^* -
  \int_{\mathbb{R}^3} [p_{l'm'n'}(\bv^*)]^{\dagger} p_{lmn}(\bv^*)
    \frac{\mathrm{d} \omega(\bv^*)}{\mathrm{d}v_x^*} \,\mathrm{d}\bv^*,
\end{split}
\end{displaymath}
where the first term on the right-hand side is zero since $p_{l'm'n'}(\cdot)$
is an orthogonal polynomial. Thus,
\begin{displaymath}
\begin{split}
\int_{\mathbb{R}^3} [p_{l'm'n'}(\bv^*)]^{\dagger}
  \frac{\mathrm{d} p_{lmn}(\bv^*)}{\mathrm{d} v_x^*} \omega(\bv^*)
  \,\mathrm{d}\bv^* &=
-\int_{\mathbb{R}^3} [p_{l'm'n'}(\bv^*)]^{\dagger} p_{lmn}(\bv^*)
  \frac{\mathrm{d} \omega(\bv^*)}{\mathrm{d}v_x^*} \,\mathrm{d}\bv^* \\
&= \int_{\mathbb{R}^3} [p_{l'm'n'}(\bv^*)]^{\dagger} v_x^* p_{lmn}(\bv^*)
  \omega(\bv^*) \,\mathrm{d}\bv^*.
\end{split}
\end{displaymath}
Now we insert \eqref{eq:recurrence} to the above equation. By the orthogonality
of $p_{lmn}$, it is not difficult to see that
\begin{displaymath}
\int_{\mathbb{R}^3} [p_{l'm'n'}(\bv^*)]^{\dagger}
  \frac{\mathrm{d} p_{lmn}(\bv^*)}{\mathrm{d} v_x^*} \omega(\bv^*)
  \,\mathrm{d}\bv^* =
\sqrt{n+l+1/2} \gamma_{lm} \delta_{l-1,l'} \delta_{mm'} \delta_{nn'} -
  \sqrt{n} \gamma_{l+1,m} \delta_{l+1,l'} \delta_{mm'} \delta_{n+1,n'},
\end{displaymath}
which implies the equality \eqref{eq:derivative}.
\end{proof}

An immediate corollary of the above lemma is the parallel properties for
$p_{lmn}^{[\lu,\lth]}(\cdot)$:
\begin{lemma}
The polynomials $p_{lmn}^{[\lu,\lth]}$ defined in \eqref{eq:p_omega} satisfy
the following properties:
\begin{gather}
\label{eq:vp}
\begin{split}
v_x p_{lmn}^{[\lu,\lth]}(\bv) = \overline{u}_x p_{lmn}^{[\lu,\lth]}(\bv) &+
  \lth \left(
    \sqrt{n+l+3/2} \gamma_{l+1,m} p_{l+1,m,n}^{[\lu,\lth]}(\bv) -
    \sqrt{n+1} \gamma_{-l,m} p_{l-1,m,n+1}^{[\lu,\lth]}(\bv)
  \right) \\
  & +\sqrt{n+l+1/2} \gamma_{-l,m} p_{l-1,m,n}^{[\lu,\lth]}(\bv) -
    \sqrt{n} \gamma_{l+1,m} p_{l+1,m,n-1}^{[\lu,\lth]}(\bv).
\end{split} \\
\label{eq:vp0}
v_x p_{lm0}^{[\lu,\lth]}(\bv) =
  \sqrt{l+3/2} \gamma_{l+1,m} p_{l+1,m,0}^{[\lu,\lth]}(\bv) +
  \sqrt{\frac{1}{l+1/2}} \gamma_{lm} \frac{|\bv|^2}{2}
    p_{l-1,m,0}^{[\lu,\lth]}(\bv), \\
\label{eq:differential}
\frac{\mathrm{d} p_{lmn}^{[\lu,\lth]}(\bv)}{\mathrm{d}v_x} =
  \lth^{-1} \left[ \sqrt{n+l+1/2} \gamma_{lm} p_{l-1,m,n}^{[\lu,\lth]}(\bv)
    -\sqrt{n} \gamma_{l+1,m} p_{l+1,m,n-1}^{[\lu,\lth]}(\bv) \right].
\end{gather}
\end{lemma}
This equations can be directly obtained from Lemma \ref{lem:p_property} by
replacing $\bv^*$ with $(\bv - \lu) / \sqrt{\lth}$. The detail of the proof is
omitted. Note that the equation \eqref{eq:vp} is the same as
\eqref{eq:recursion}.

\begin{lemma}
For
\begin{displaymath}
\mathcal{M}(\bv) = \frac{\rho}{\mm(2\pi \theta)^{3/2}}
  \exp \left( -\frac{|\bv - \bu|^2}{2\theta} \right),
\end{displaymath}
it holds that
\begin{equation} \label{eq:n=0}
\mm \int_{\mathbb{R}^3} \left[ p_{lm0}^{[\lu,\lth]}(\bv) \right]^{\dagger}
  \mathcal{M}(\bv) \,\mathrm{d}\bv
= \rho \left[ p_{lm0}^{[\lu,\lth]}(\bu) \right]^{\dagger}.
\end{equation}
\end{lemma}

\begin{proof}
The proof uses the following formula \cite{Caola1978}:
\begin{displaymath}
p_{lm0}^{[\lu,\lth]}(\bv) = \sum_{\lambda=0}^l \sum_{\mu=-\lambda}^{\lambda}
  \begin{pmatrix} l+m \\ \lambda+\mu \end{pmatrix}^{1/2}
  \begin{pmatrix} l-m \\ \lambda-\mu \end{pmatrix}^{1/2}
  \sqrt{\frac{\Gamma(l-\lambda+1/2) \Gamma(\lambda+1/2)}
    {\sqrt{\pi} \Gamma(l+1/2)}} p_{\lambda\mu0}^{[\lu,\lth]}(\bu)
  p_{l-\lambda,m-\mu,0}^{[\bu,\lth]}(\bv).
\end{displaymath}
Since
\begin{displaymath}
\begin{split}
& \mm \int_{\mathbb{R}^3}
  \left[ p_{lm0}^{[\bu,\lth]}(\bv) \right]^{\dagger}
  \mathcal{M}(\bv) \,\mathrm{d}{\bv} \\
={} & \frac{\rho}{(2\pi \theta)^{3/2}}
  \sqrt{\frac{2^{1-l} \pi^{3/2}}{\Gamma(l+3/2)}}
  \int_{\mathbb{R}^3} \left| \frac{\bv - \bu}{\sqrt{\lth}} \right|^l
    \left[ Y_l^m \left( \frac{\bv - \bu}{|\bv - \bu|} \right) \right]^{\dagger}
    \exp \left( -\frac{|\bv - \bu|^2}{2\theta} \right)
  \,\mathrm{d}\bv \\
={} & \frac{\rho}{(2\pi \theta)^{3/2}}
  \sqrt{\frac{2^{1-l} \pi^{3/2}}{\Gamma(l+3/2)}}
  \int_0^{+\infty} \left( \frac{r}{\sqrt{\lth}} \right)^l
    \exp \left(-\frac{r^2}{2\theta} \right) r^2 \,\mathrm{d}r
  \int_{\mathbb{S}^2} Y_l^m(\bn) \,\mathrm{d}\bn \\
={} & \frac{\rho}{\theta^{3/2}} \sqrt{\frac{2^{-l}}{\Gamma(l+3/2) \sqrt{\pi}}}
  \delta_{l0} \delta_{m0} \int_0^{+\infty} \left( \frac{r}{\sqrt{\lth}} \right)^l
    \exp \left(-\frac{r^2}{2\theta} \right) r^2 \,\mathrm{d}r
= \rho \delta_{l0} \delta_{m0},
\end{split}
\end{displaymath}
we obtain
\begin{displaymath}
\begin{split}
& \mm \int_{\mathbb{R}^3}
  \left[ p_{lm0}^{[\lu,\lth]}(\bv) \right]^{\dagger}
  \mathcal{M}(\bv) \,\mathrm{d}{\bv} \\
={} & \sum_{\lambda=0}^l \sum_{\mu=-\lambda}^{\lambda}
  \begin{pmatrix} l+m \\ \lambda+\mu \end{pmatrix}^{1/2}
  \begin{pmatrix} l-m \\ \lambda-\mu \end{pmatrix}^{1/2}
  \sqrt{\frac{\Gamma(l-\lambda+1/2) \Gamma(\lambda+1/2)}{\sqrt{\pi} \Gamma(l+1/2)}}
  \left[ p_{\lambda\mu0}^{[\lu,\lth]}(\bu) \right]^{\dagger}
  \rho \delta_{l-\lambda,0} \delta_{m-\mu,0}
    = \rho \left[ p_{lm0}^{[\lu,\lth]}(\bu) \right]^{\dagger},
\end{split}
\end{displaymath}
which completes the proof.
\end{proof}

\begin{lemma}
For $\mathcal{M}(\bv)$ defined in \eqref{eq:general_Maxwellian}, it holds that
\begin{equation} \label{eq:vpM}
\begin{split}
& \mm \int_{\mathbb{R}^3} v_x \left[ p_{lmn}^{[\lu,\lth]}(\bv) \right]^{\dagger}
  \mathcal{M}(\bv) \,\mathrm{d}\bv \\
={} & \mm \int_{\mathbb{R}^3} \left[
  u_x p_{lmn}^{[\lu,\lth]}(\bv) + 
  \sqrt{n+l+1/2} \gamma_{lm} (\theta / \lth) p_{l-1,m,n}^{[\lu,\lth]}(\bv) -
  \sqrt{n} \gamma_{l+1,m} (\theta / \lth) p_{l+1,m,n-1}^{[\lu,\lth]}(\bv)
\right]^{\dagger} \mathcal{M}(\bv) \,\mathrm{d}\bv.
\end{split}
\end{equation}
\end{lemma}

\begin{proof}
By integration by parts,
\begin{displaymath}
\begin{split}
\mm \int_{\mathbb{R}^3} v_x \left[ p_{lmn}^{[\lu,\lth]}(\bv) \right]^{\dagger}
  \mathcal{M}(\bv) \,\mathrm{d}\bv
&= \mm \int_{\mathbb{R}^3} u_x \left[ p_{lmn}^{[\lu,\lth]}(\bv) \right]^{\dagger}
  \mathcal{M}(\bv) \,\mathrm{d}\bv
- \mm \theta \int_{\mathbb{R}^3} \left[ p_{lmn}^{[\lu,\lth]}(\bv) \right]^{\dagger}
  \frac{\mathrm{d} \mathcal{M}(\bv)}{\mathrm{d} v_x} \,\mathrm{d}\bv \\
&= \mm \int_{\mathbb{R}^3} \left[
  u_x p_{lmn}^{[\lu,\lth]}(\bv) +
  \theta \frac{\mathrm{d} p_{lmn}^{[\lu,\lth]}(\bv)}{\mathrm{d}v_x}
\right]^{\dagger} \mathcal{M}(\bv) \,\mathrm{d}\bv.
\end{split}
\end{displaymath}
Then \eqref{eq:vpM} can be obtained by inserting \eqref{eq:differential} into
the above equation.
\end{proof}

\section{Proof of Theorem \ref{thm:Maxwellian}}
\begin{proof}
Since $\lth > \theta/2$, by straightforward calculation, we get
\begin{displaymath}
\int_{\mathbb{R}^3} [\mathcal{M}(\bv)]^2
  [\omega^{[\lu,\lth]}(\bv)]^{-1} \,\mathrm{d}\bv
= \frac{\rho^2 \lth^3}{\mm [\theta (2\lth - \theta)]^{3/2}}
  \exp \left( \frac{|\bu - \lu|^2}{2\lth - \theta} \right) < +\infty.
\end{displaymath}
Therefore the expansion \eqref{eq:M_expansion} holds for
\begin{equation} \label{eq:M_lmn}
\tilde{\mathcal{M}}_{lmn} = \mm \lth^{l+2n} \int_{\mathbb{R}^3}
  \left[ p_{lmn}^{[\lu,\lth]} (\bv) \right]^{\dagger} \mathcal{M}(\bv)
\,\mathrm{d}\bv.
\end{equation}
When $n > 0$, we can apply \eqref{eq:vp} to get
\begin{displaymath}
\begin{split}
\tilde{\mathcal{M}}_{lmn} &= \frac{1}{\sqrt{n} \gamma_{l+1,m}} \Bigg(
  \sqrt{n+l+3/2} \gamma_{l+2,m} \tilde{\mathcal{M}}_{l+2,m,n-1} -
  \sqrt{n-1} \gamma_{l+2,m} \lth \tilde{\mathcal{M}}_{l+2,m,n-2} \\
& \qquad + \sqrt{n+l+1/2} \gamma_{l+1,m} \lth \tilde{\mathcal{M}}_{l,m,n-1}
  + \bar{u}_x \tilde{\mathcal{M}}_{l+1,m,n-1}
  - \mm \int_{\mathbb{R}^3}
    v_x \left[ p_{l+1,m,n-1}^{[\lu,\lth]}(\bv) \right]^{\dagger}
    \mathcal{M}(\bv) \,\mathrm{d}\bv \Bigg).
\end{split}
\end{displaymath}
The integral term in the above equation can be calculated by applying
\eqref{eq:vpM}, and the result will be \eqref{eq:recur} after simplification.
Similarly, combining \eqref{eq:vp0}, \eqref{eq:n=0} and \eqref{eq:M_lmn} yields
the iterative formula \eqref{eq:tildeM0}.

The proof of the initial condition \eqref{eq:tildeMll0} requires the following
formula:
\begin{displaymath}
|\bv^*|^l Y_l^{\pm l} \left( \frac{\bv^*}{|\bv^*|} \right) =
  \frac{(\mp 1)^l}{2^l l!} \sqrt{\frac{(2l+1)!}{4\pi}}
  (v_y^* \pm \mathrm{i} v_z^*)^l,
\end{displaymath}
from which we know that
\begin{equation} \label{eq:p_ll0}
p_{l,\pm l, 0}^{[\lu,\lth]}(\bv) = \lth^{-l} \frac{(\pm 1)^l}{2^l l!}
  \sqrt{\frac{2^{1-l} \pi^{3/2}}{\Gamma(l+3/2)}} \sqrt{\frac{(2l+1)!}{4\pi}}
  (v_y^* \pm \mathrm{i} v_z^*)^l =
(\mp 1)^l \sqrt{\frac{1}{2^l l!}} \, \lth^{-l} (v_y^* \pm \mathrm{i} v_z^*)^l,
\end{equation}
where we have used
\begin{displaymath}
\Gamma(l+3/2) = \frac{(2l+2)!}{4^{l+1} (l+1)!} \sqrt{\pi}
\end{displaymath}
in the second equality of \eqref{eq:p_ll0}. Then \eqref{eq:tildeMll0} is an
immediate result of \eqref{eq:n=0}, \eqref{eq:M_lmn} and \eqref{eq:p_ll0}.
\end{proof}

\bibliographystyle{amsplain}
\bibliography{../article}

\providecommand{\bysame}{\leavevmode\hbox to3em{\hrulefill}\thinspace}
\providecommand{\MR}{\relax\ifhmode\unskip\space\fi MR }
\providecommand{\MRhref}[2]{%
  \href{http://www.ams.org/mathscinet-getitem?mr=#1}{#2}
}
\providecommand{\href}[2]{#2}
\begin{thebibliography}{10}

\bibitem{Abramowitz}
M.~Abramowitz and I.A. Stegun, \emph{Handbook of mathematical functions with
  formulas, graphs, and mathematical tables}, Dover, New York, 1964.

\bibitem{Alekseenko2014}
A.~Alekseenko and E.~Josyula, \emph{Deterministic solution of the spatially
  homogeneous {B}oltzmann equation using discontinuous {G}alerkin
  discretizations in the velocity space}, J. Comput. Phys. \textbf{272} (2014),
  170--188.

\bibitem{Alekseenko2019}
A.~Alekseenko and J.~Limbacher, \emph{Evaluating high order discontinuous
  {G}alerkin discretization of the {B}oltzmann collision integral in {$O(N^2)$}
  operations using the discrete fourier transform}, Kin. Rel. Models
  \textbf{12} (2019), no.~4, 703--726.

\bibitem{Alonso2017}
R.~Alonso, I.~M. Gamba, and M.~Taskovi{\`c}, \emph{Exponentially-tailed
  regularity and time asymptotic for the homogeneous {B}oltzmann equation},
  arXiv:1711.06596 (2017).

\bibitem{Bhatnagar1954}
P.~L. Bhatnagar, E.~P. Gross, and M.~Krook, \emph{A model for collision
  processes in gases. {I}. small amplitude processes in charged and neutral
  one-component systems}, Phys. Rev. \textbf{94} (1954), no.~3, 511--525.

\bibitem{Bird1963}
G.~A. Bird, \emph{Approach to translational equilibrium in a rigid sphere gas},
  Phys. Fluids \textbf{6} (1963), no.~10, 1518--1519.

\bibitem{Bird1994}
\bysame, \emph{Molecular gas dynamics and the direct simulation of gas flows},
  Oxford: Clarendon Press, 1994.

\bibitem{Bobylev1997}
A.~Bobylev and S.~Rjasanow, \emph{Difference scheme for the {B}oltzmann
  equation based on the fast {F}ourier transform}, Eur. J. Mech. B Fluids
  \textbf{16} (1997), no.~2, 293--306.

\bibitem{Burnett1936}
D.~Burnett, \emph{The distribution of molecular velocities and the mean motion
  in a non-uniform gas}, Proc. London Math. Soc. \textbf{40} (1936), no.~1,
  382--435.

\bibitem{Cai2019}
Z.~Cai, Y.~Fan, and Y.~Wang, \emph{Burnett spectral method for the spatially
  homogeneous {B}oltzmann equation}, arXiv:1810.07804 (2019), submitted.

\bibitem{Cai2015}
Z.~Cai and M.~Torrilhon, \emph{Approximation of the linearized {B}oltzmann
  collision operator for hard-sphere and inverse-power-law models}, J. Comput.
  Phys. \textbf{295} (2015), 617--643.

\bibitem{Cai2018}
\bysame, \emph{Numerical simulation of microflows using moment methods with
  linearized collision operator}, J. Sci. Comput. \textbf{74} (2018), 336--374.

\bibitem{Caola1978}
M.~J. Caola, \emph{Solid harmonics and their addition theorems}, J. Phys. A:
  Math. Gen. \textbf{11} (1978), no.~2, L23--L25.

\bibitem{Dechriste2016}
G.~Dechrist{\'e} and L.~Mieussens, \emph{A {C}artesian cut cell method for
  rarefied flow simulations around moving obstacles}, J. Comput. Phys.
  \textbf{314} (2016), 465--488.

\bibitem{Dimarco2019}
G.~Dimarco, C.~Hauck, and R.~Loub{\`e}re, \emph{A class of low dissipative
  schemes for solving kinetic equations}, J. Sci. Comput. \textbf{78} (2019),
  393--432.

\bibitem{Dimarco2018}
G.~Dimarco, R.~Loub{\`e}re, J.~Narski, and T.~Rey, \emph{An efficient numerical
  method for solving the {B}oltzmann equation in multidimensions}, J. Comput.
  Phys. \textbf{353} (2018), 46--81.

\bibitem{Filbet2011}
F.~Filbet and S.~Jin, \emph{An asymptotic preserving scheme for the {ES}-{BGK}
  model of the {B}oltzmann equation}, J. Sci. Comput. \textbf{46} (2011),
  204--224.

\bibitem{Filbet2015}
F.~Filbet, L.~Pareschi, and T.~Rey, \emph{On steady-state preserving spectral
  methods for homogeneous {B}oltzmann equations}, Comptes Rendus Mathematique
  \textbf{353} (2015), no.~4, 309--314.

\bibitem{Gamba2017}
I.~M. Gamba, J.~R. Haack, C.~D. Hauck, and J.~Hu, \emph{A fast spectral method
  for the {B}oltzmann collision operator with general collision kernels}, SIAM
  J. Sci. Comput. \textbf{39} (2017), no.~14, B658--B674.

\bibitem{Gamba2018}
I.~M. Gamba and S.~Rjasanow, \emph{{G}alerkin-{P}etrov approach for the
  {B}oltzmann equation}, J. Comput. Phys. \textbf{366} (2018), 341--365.

\bibitem{Gamba2009}
I.~M. Gamba and S.~H. Tharkabhushanam, \emph{Spectral-{L}agrangian methods for
  collisional models of non-equilibrium statistical states}, J. Comput. Phys.
  \textbf{228} (2009), no.~6, 2012--2036.

\bibitem{Grad1949}
H.~Grad, \emph{On the kinetic theory of rarefied gases}, Comm. Pure Appl. Math.
  \textbf{2} (1949), no.~4, 331--407.

\bibitem{Hesthaven}
J.~S. Hesthaven and T.~Warburton, \emph{Nodal discontinuous {G}alerkin
  methods}, Springer, 2008.

\bibitem{Hu2019}
Z.~Hu, Z.~Cai, and Y.~Wang, \emph{Numerical simulation of microflows using
  {H}ermite spectral methods}, SIAM J. Sci. Comput. (2019), To appear.

\bibitem{Huang2012}
J.~Huang, K.~Xu, and P.~Yu, \emph{A unified gas-kinetic scheme for continuum
  and rarefied flows {II}: Multi-dimensional cases}, Commun. Comput. Phys.
  \textbf{12} (2012), no.~3, 662--690.

\bibitem{Kitzler2019}
G.~Kitzler and J.~Schr{\"o}berl, \emph{A polynomial spectral method for the
  spatially homogeneous {B}oltzmann equation}, SIAM J. Sci. Comput. \textbf{41}
  (2019), no.~1, B27--B49.

\bibitem{Mouhot2006}
C.~Mouhot and L.~Pareschi, \emph{Fast algorithms for computing the {B}oltzmann
  collision operator}, Math. Comp. \textbf{75} (2006), no.~256, 1833--1852.

\bibitem{Pareschi1996}
L.~Pareschi and B.~Perthame, \emph{A {F}ourier spectral method for homogeneous
  {B}oltzmann equations}, Transport Theor. Stat. \textbf{25} (1996), no.~3--5,
  369--382.

\bibitem{Shakhov1968}
E.~M. Shakhov, \emph{Generalization of the {K}rook kinetic relaxation
  equation}, Fluid Dyn. \textbf{3} (1968), no.~5, 95--96.

\bibitem{Tang1993}
T.~Tang, \emph{The {H}ermite spectral method for {G}aussian-type functions},
  SIAM J. Sci. Comput. \textbf{14} (1993), no.~3, 594--606.

\bibitem{Timokhin2017}
M.~Yu. Timokhin, H.~Struchtrup, A.~A. Kokhanchik, and Ye.~A. Bondar,
  \emph{Different variants of {R}13 moment equations applied to the shock-wave
  structure}, Phys. Fluids \textbf{29} (2017), 037105.

\bibitem{Torrilhon2015}
M.~Torrilhon, \emph{Convergence study of moment approximations for boundary
  value problems of the {B}oltzmann-{BGK} equation}, Commun. Comput. Phys.
  \textbf{18} (2015), no.~3, 529--557.

\bibitem{Wang2019}
Y.~Wang and Z.~Cai, \emph{Approximation of the {B}oltzmann collision operator
  based on {H}ermite spectral method}, J. Comput. Phys. \textbf{397} (2019),
  108815.

\bibitem{Wu2013}
L.~Wu, C.~White, T.~Scanlona, J.~Reese, and Y.~Zhang, \emph{Deterministic
  numerical solutions of the {B}oltzmann equation using the fast spectral
  method}, J. Comput. Phys. \textbf{250} (2013), 27--52.

\bibitem{Xu2011}
K.~Xu and J.C. Huang, \emph{An improved unified gas-kinetic scheme and the
  study of shock structures}, IMA J. Appl. Math. \textbf{76} (2011), 698--711.

\end{thebibliography}
\end{document}